\PassOptionsToPackage{unicode}{hyperref}
\PassOptionsToPackage{hyphens}{url}
\PassOptionsToPackage{dvipsnames,svgnames,x11names}{xcolor}
\documentclass[
  12pt]{article}

\usepackage{amsmath,amssymb}
\usepackage{iftex}
\ifPDFTeX
  \usepackage[T1]{fontenc}
  \usepackage[utf8]{inputenc}
  \usepackage{textcomp} 
\else 
  \usepackage{unicode-math}
  \defaultfontfeatures{Scale=MatchLowercase}
  \defaultfontfeatures[\rmfamily]{Ligatures=TeX,Scale=1}
\fi
\usepackage{lmodern}

\usepackage[linesnumbered, ruled, vlined]{algorithm2e}
\SetKwRepeat{Do}{do}{while}%

\ifPDFTeX\else  
\fi
\IfFileExists{upquote.sty}{\usepackage{upquote}}{}
\IfFileExists{microtype.sty}{
  \usepackage[]{microtype}
  \UseMicrotypeSet[protrusion]{basicmath} 
}{}
\makeatletter
\@ifundefined{KOMAClassName}{
  \IfFileExists{parskip.sty}{%
    \usepackage{parskip}
  }{
    \setlength{\parindent}{0pt}
    \setlength{\parskip}{6pt plus 2pt minus 1pt}}
}{
  \KOMAoptions{parskip=half}}
\makeatother
\usepackage{xcolor}
\makeatletter
\ifx\paragraph\undefined\else
  \let\oldparagraph\paragraph
  \renewcommand{\paragraph}{
    \@ifstar
      \xxxParagraphStar
      \xxxParagraphNoStar
  }
  \newcommand{\xxxParagraphStar}[1]{\oldparagraph*{#1}\mbox{}}
  \newcommand{\xxxParagraphNoStar}[1]{\oldparagraph{#1}\mbox{}}
\fi
\ifx\subparagraph\undefined\else
  \let\oldsubparagraph\subparagraph
  \renewcommand{\subparagraph}{
    \@ifstar
      \xxxSubParagraphStar
      \xxxSubParagraphNoStar
  }
  \newcommand{\xxxSubParagraphStar}[1]{\oldsubparagraph*{#1}\mbox{}}
  \newcommand{\xxxSubParagraphNoStar}[1]{\oldsubparagraph{#1}\mbox{}}
\fi
\makeatother

\usepackage{longtable,booktabs,array}
\usepackage{calc} 
\usepackage{etoolbox}
\makeatletter
\patchcmd\longtable{\par}{\if@noskipsec\mbox{}\fi\par}{}{}
\makeatother
\IfFileExists{footnotehyper.sty}{\usepackage{footnotehyper}}{\usepackage{footnote}}
\makesavenoteenv{longtable}
\usepackage{graphicx}
\makeatletter
\def\maxwidth{\ifdim\Gin@nat@width>\linewidth\linewidth\else\Gin@nat@width\fi}
\def\maxheight{\ifdim\Gin@nat@height>\textheight\textheight\else\Gin@nat@height\fi}
\makeatother
\setkeys{Gin}{width=\maxwidth,height=\maxheight,keepaspectratio}
\makeatletter
\def\fps@figure{htbp}
\makeatother

\makeatletter
\@ifpackageloaded{caption}{}{\usepackage{caption}}
\AtBeginDocument{%
\ifdefined\contentsname
  \renewcommand*\contentsname{Table of contents}
\else
  \newcommand\contentsname{Table of contents}
\fi
\ifdefined\listfigurename
  \renewcommand*\listfigurename{List of Figures}
\else
  \newcommand\listfigurename{List of Figures}
\fi
\ifdefined\listtablename
  \renewcommand*\listtablename{List of Tables}
\else
  \newcommand\listtablename{List of Tables}
\fi
\ifdefined\figurename
  \renewcommand*\figurename{Figure}
\else
  \newcommand\figurename{Figure}
\fi
\ifdefined\tablename
  \renewcommand*\tablename{Table}
\else
  \newcommand\tablename{Table}
\fi
}
\@ifpackageloaded{float}{}{\usepackage{float}}
\floatstyle{ruled}
\@ifundefined{c@chapter}{\newfloat{codelisting}{h}{lop}}{\newfloat{codelisting}{h}{lop}[chapter]}
\floatname{codelisting}{Listing}

\makeatother
\makeatletter
\@ifpackageloaded{caption}{}{\usepackage{caption}}
\@ifpackageloaded{subcaption}{}{\usepackage{subcaption}}
\makeatother

\ifLuaTeX
  \usepackage{selnolig}  
\fi
\usepackage[]{natbib}
\usepackage{bookmark}

\IfFileExists{xurl.sty}{\usepackage{xurl}}{} 
\hypersetup{
  pdftitle={Title},
  pdfauthor={Author 1; Author 2},
  pdfkeywords={3 to 6 keywords, that do not appear in the title},
  colorlinks=true,
  linkcolor={blue},
  filecolor={Maroon},
  citecolor={Blue},
  urlcolor={Blue},
  pdfcreator={LaTeX via pandoc}}

\newcommand{\anon}{1}

\def\1{\mbox{\bf 1}}
\def\R{\mathbb{R}}

\def\N{\mathbb{N}}
\def\P{\mathbb{P}}
\def\E{\mathbb{E}}

\def\R{\mathbb{R}}

\newcolumntype{E}{>{\centreing\arraybackslash}m{1.9cm}} 

\newcommand{\op}{\mathrm{op}}

\newtheorem{theo}{Theorem}
\newtheorem{lem}{Lemma}
\newtheorem{prop}{Proposition}

\newtheorem{Def/Prop}{Definition-Proposition}

\begin{document}

\def\spacingset#1{\renewcommand{\baselinestretch}%
{#1}\small\normalsize} \spacingset{1}


\if1\anon
{
  \title{\bf Learning Centre Partitions from Summaries}
  \author{ Zinsou Max Debaly  \\  
    Faculté des Sciences,  Département de mathématiques,   
    Université de Sherbrooke \\
    Health Data Research Network\\
    Jean-François Ethier \\
    Faculté de médecine et des sciences de la santé, 
    Université de Sherbrooke \\
    Health Data Research Network\\
    Michael H. Neumann\\
    Institut für Mathematik, \\
    Friedrich-Schiller-Universität\\ 
    Félix Camirand Lemyre \\
    Faculté des Sciences,  Département de mathématiques,    
    Université de Sherbrooke \\
    Health Data Research Network}
  \maketitle
} \fi

\if0\anon
{
  \bigskip
  \bigskip
  \bigskip
  \begin{centre}
    {\LARGE\bf Title}
\end{centre}
  \medskip
} \fi

 \begin{abstract}
Multi-centre studies increasingly rely on distributed inference, where sites share only
centre-level summaries. In practice, homogeneity of parameters across centres is often
violated, which calls for procedures that both \emph{test} equality and \emph{learn}
centre groupings prior to estimation. We develop multivariate Cochran-type tests
operating solely on summary statistics, derive their asymptotic $\chi^2$-mixture null
distributions, and provide plug-in estimators that make them fully implementable from
centre-level outputs. These tests are embedded in a sequential, test-driven
\emph{Clusters-of-Centres} (\textsc{CoC}) algorithm that merges centres only when
equality is not rejected. To improve finite-sample behaviour, we introduce a multi-round
bootstrap \textsc{CoC} procedure that re-evaluates candidate fusions across independently
resampled sets of summaries; under mild regularity conditions and a separation
requirement, the true partition is recovered with probability tending to one as the
number of rounds grows. We complement these asymptotic guarantees with explicit type-I and
type-II error bounds derived via Berry--Esseen approximations and
$\sqrt{(\log n)/n}$ deviation inequalities, characterise a
detectability threshold of the same order in terms of centre
separability, and analyse a bootstrap variant
with a shrinking rejection region under which both error rates vanish in probability
simultaneously. Simulations and an application to U.S.\ airline on-time performance
data illustrate accurate heterogeneity detection and reliable partition recovery across
a range of configurations.
\end{abstract}

\noindent%
\vfill

\newpage
\spacingset{1.8} 

\section{Introduction}

The proliferation of large-scale,  distributed data systems across domains, ranging from healthcare and finance to education,  marketing,  and the environmental sciences, has transformed the landscape of modern statistical inference. Advances in cloud computing,  digital infrastructure,  and real-time data acquisition have enabled the collection of massive datasets from geographically dispersed sources. Whether through sensor networks,  enterprise data lakes,  or nationwide administrative registries,  multi-site data generation is now the norm rather than the exception.
In parallel,  growing societal and legal concerns over privacy,  confidentiality,  and ethical data use have led to the development of stringent data protection frameworks,  such as the General Data Protection Regulation (GDPR) in Europe and the Health Insurance Portability and Accountability Act (HIPAA) in the United States. These frameworks often prohibit the pooling of sensitive,  individual-level data into centralized repositories,  especially in high-stakes domains involving personal,  financial,  or health information.
As a result,  federated or distributed inference paradigms, where data remain local to each site and only aggregate statistics,  model summaries,  or encrypted representations are exchanged, have emerged as crucial tools for collaborative analytics. These approaches   enable joint modeling,  estimation,  and hypothesis testing across multiple centres. They are particularly well-suited to studies of rare events,  subgroup-specific dynamics,  or heterogeneous populations,  where a single site may lack sufficient data,  but combined information can yield robust and generalizable insights.

The first distributed inference methods typically relied on straightforward aggregation rules. For instance,  the naive average pools local parameter estimates from each centre without sophisticated weighting,  whereas inverse-variance weighted (IVW) averaging  incorporates estimation uncertainties from each centre to improve statistical efficiency. More advanced approaches have subsequently been developed,  such as the communication-efficient surrogate likelihood (CSL) proposed by \cite{Jordan03042019} and the aggregated estimating equations (AEE) introduced by \cite{lin2011aggregated}. These methods significantly enhance estimation accuracy compared to simplistic averaging and reduce communication burdens. For a comprehensive overview,  see the recent review works by \cite{gao2022review} and \cite{lemyre2024distributed}.

A critical challenge in multi-centre analyses is inherent heterogeneity among data centres. For instance, variations in data collection protocols,  patient demographics,  and equipment calibration lead to systematic differences across sites,  complicating straightforward aggregation. While extensive harmonization  largely resolves operational heterogeneity (workflows, coding, measurement protocols), statistical distribution heterogeneity, reflected in divergent covariate distributions and outcome relationships, often persists across sites.  Ignoring such heterogeneity can introduce biases and obscure   meaningful variations. 

Traditionally,  Cochran's Q test (see,  for example,  \cite{kulinskaya2011testing}) has been employed to detect inter-study heterogeneity in meta-analyses and can be adapted for multi-centre studies. This test evaluates whether discrepancies among site-specific estimates exceed expectations from sampling variability alone. However,  Cochran's Q test and its random-effects model extensions \citep{higgins2002quantifying} are inherently univariate,  assessing single parameters or outcomes independently and thus overlooking covariance structures among multiple correlated parameters. Multivariate extensions \citep{gasparrini2012multivariate} primarily test for the existence of random effects rather than directly assessing the equality of parameter vectors across multiple centres. But,  no multivariate Cochran-type test explicitly designed to test parameter equality in a distributed inference context exists in the literature,  likely due to traditional meta-analysis settings' limited availability of detailed summary statistics,  prompting reliance on parameter-by-parameter univariate tests.

 Recent works in distributed inference address heterogeneity primarily through two distinct approaches. The first involves designing estimation procedures that aim to identify a prevailing common parameter among the centre-specific parameters,  treating the remaining parameters as outliers. For instance,  \cite{minsker2019distributed} proposed a robust estimator at the aggregation stage,  while \cite{guo2025robust} developed a two-step procedure that first selects centres sharing the prevailing parameter via a data-driven criterion and subsequently performs post-selection inference using resampling methods. These approaches generally assume that at least $K/2$ centres share the same true parameter; see Assumption~1 in \cite{guo2025robust} or Section~2.4.1 in \cite{minsker2019distributed}. Our goal is different: we learn the full centre partition. Our paper introduces a novel perspective suitable for highly heterogeneous settings,  explicitly accommodating the extreme scenario in which each centre defines its own parameter.
The second approach posits a global model encompassing both centre-specific nuisance parameters and globally common parameters. Examples include the fully parametric distributed modeling framework of \cite{gu2023distributed} and \cite{DuanHeterogeneityAware} and  the semi-parametric partially linear framework studied by \cite{ZhaoAOSpartiallyLinearFramework}. However,  such frameworks currently lack explicit guidelines to differentiate between parameters that are genuinely common and those that are heterogeneous. 

Closely related to our setting,  \citet{chen2024heterogeneity} recently considered clustering of centres with a focus on high–dimensional models. They propose a penalized estimator that simultaneously performs variable selection (e.g.,  via a SCAD penalty) and convex–clustering–based fusion. However,  the stated guarantees require nontrivial prior knowledge of the oracle grouping structure, such as lower and upper bounds on cluster sizes, to specify tuning parameters appropriately \citep[Theorems~1–2]{chen2024heterogeneity}. Moreover,  parameter consistency does not by itself ensure clustering consistency: although convex clustering admits a unique solution for each tuning level,  that solution depends on the tuning parameter,  and for sufficiently large values all centres may be merged into a single cluster even when the underlying parameters differ. By contrast,  our test-based
merging method is free of tuning parameters and recovers the true partition of centres (see
Theorem 1).

 \paragraph*{Overview of contributions.}
In this paper, we
(i) derive multivariate Cochran-type global tests and two-block integration tests for the equality of parameter vectors based solely on centre-level summary statistics;
(ii) propose a test-driven clustering algorithm (the CoC algorithm) that sequentially fuses blocks whenever equality is not rejected, using a deterministic largest-\(p\) tie-break rule;
(iii) establish a golden-partition recovery result for the multi-round CoC algorithm based on bootstrap resamples: under standard regularity conditions and a separation assumption between the true blocks, we show that \(\P\!\big(\widehat{\mathcal C}_n^{(R(n))}=\mathcal P\big)\to 1\), where \(\widehat{\mathcal C}_n^{(R(n))}\) denotes the partition returned after \(R(n)\) rounds and \(\mathcal P\) is the true underlying partition; and
(iv) introduce fusion tests with a shrinkage rejection region that also achieve asymptotic recovery of the true partition.

\paragraph*{Notations}
Throughout,  $\Rightarrow$ denotes convergence in distribution.
For a vector $u$ and a positive–semidefinite matrix $\Sigma$, 
$\mathcal{N}(u, \Sigma)$ denotes the multivariate normal distribution with mean $u$
and covariance $\Sigma$.
We write $\xrightarrow{p}$ for convergence in probability and use
$a_n = o_p(1)$ to indicate that the random sequence $a_n$ converges in probability to $0$;
 more generally,  $a_n = o_p(b_n)$ means $a_n/b_n \xrightarrow{p} 0$.
Let $I_p$ be the $p \times p$ identity matrix.
For square matrices $M_1, \dots, M_q$ of size $p \times p$, 
$\operatorname{diag}(M_1, \dots, M_q)$ denotes the block–diagonal matrix of order $pq$
with diagonal blocks $M_1, \dots, M_q$.
For any matrix $M$,  $M^{\top}$ denotes its transpose; if $M$ is positive–semidefinite, 
$M^{1/2}$ denotes any symmetric square root satisfying $M^{1/2} M^{1/2}  = M$, $\|M\|_{\op}
\;:=\;
\sup_{\|x\|_2=1}\,\|Mx\|_2
\;=\;
\sup_{x\neq 0}\frac{\|Mx\|_2}{\|x\|_2}
$ is the operator norm of $M$.
The symbol $0$ is used generically for the zero scalar,  the zero vector,  or the zero matrix, 
according to context.
For a set $S$,  $|S|$ denotes its cardinality.
For a vector $x=(x_1, \ldots, x_d) \in \mathbb{R}^d$, 
$\|x\| := \sqrt{\sum_{i=1}^d x_i^2}$ denotes the Euclidean norm.
For two random variables (or vectors) $X$ and $Y$,  $X \overset{d}{=} Y$ denotes equality in distribution.

\subsection{Problem Setup}

Consider a distributed data environment with a fixed  \(K\)   data centres,  each collecting \(n\) observations.  At the \(k\)th centre,  let the observed data be denoted by
\[
\{Z_{i, k}: i = 1, \dots, n\}, 
\]
where each \(Z_{i, k}\) follows a distribution \(P_{i, k}\).  We do not assume a common distribution here; however,  we assume that each \(P_{i, k}\) shares a common parameter \(\varrho_{0,k} \in \Theta \subseteq \mathbb{R}^l\). We consider a transformation of $\varrho_{0,k},  \;\theta_{0,k} = \varrho(\varrho_{0,k})$ and assume that each data centre provides an estimator \(\widehat{\theta}_{n, k}\) of \(\theta_{0,k} \in \R^p\).  We do not restrict the choice of estimation procedure.  Instead,  we suppose that the sample size \(n\) is large enough for each local estimator to admit the following Bahadur decomposition:
\begin{equation}
\label{eq::asymptoticDecomposition}
\sqrt{n}\, \bigl(\widehat{\theta}_{n, k} - \theta_{0,k}\bigr)
    = V_k^{-1}\, U_{n, k} \;+\;\varepsilon_{n, k}.
\end{equation}
We will rely on the following set of assumptions :
\begin{itemize}
  \item[(A1)] \(V_k \in \mathbb{R}^{p \times p}\): the local sensitivity positive-definite matrix;
  \item[(A2)] \(U_{n, k}\Rightarrow\mathcal{N}(0, \, Q_k)\) as \(n\to\infty\),  with \(Q_k \ne 0\) is    positive semi-definite;
  \item[(A3)] \(\varepsilon_{n, k}\xrightarrow{p}0\) as \(n\to\infty\);
  \item[(A4)] the random vectors \(U_{n, 1}, \, \ldots,  U_{n, K}\) are independent across centres \(k\) for each  $n$.
\end{itemize}
We address unequal local sample sizes later in the paper.
\paragraph*{On assumptions \textbf{(A1)}–\textbf{(A4)}} For $\theta_{0,k} = \varrho_{0,k}, $ Assumptions \textbf{(A1)}–\textbf{(A3)} hold for any M-estimator whose loss is twice continuously differentiable and whose parameter lies in the interior of  a compact set ~\(\Theta\) \citep{van2000asymptotic}.  
They also cover convex,  non-differentiable losses such as the check loss in quantile regression  \citep{koenker2005quantile}. 
For a \emph{regular } well-specified parametric model (in particular, when the support of the distribution is independent of the model's parameters),  the sensitivity and variability matrices coincide (\(V_k = Q_k\)).  
In more general M-estimation,  \(Q_k\) is the score-variance matrix and is therefore positive \emph{semi}-definite.  For an arbitrary transformation $\varrho$,  for example, a vector of odds ratios in logistic regression,  the Bahadur expansion can be obtained with the Delta method under   additional  differentiability conditions on $\varrho$, see for example \cite[chap. 3]{van2000asymptotic}. For a high-dimensional parameter $\varrho_0, $ debiased estimators for each one-dimensional component of  $\varrho_0$ that satisfy the Bahadur expansion are now available,  see for example \cite{van2014AOS} or \cite{zhang2017simultaneous} for high-dimensional regression models.  Equation \eqref{eq::aee} is  also satisfied by the broad class of profile estimating equations (see,  e.g.,  \citep[eqs. (21.5)–(21.6)]{kosorok2008introduction}) for the finite-dimensional component of the parameters of  semi-parametric models,  under assumptions \citep[A1–A4,  pp. 403–404]{kosorok2008introduction}; see \citep[Theorem 21.6]{kosorok2008introduction} for more  details. More basically,  these assumptions hold for a large class of non-degenerate U-statistics via Hoeffding’s decomposition. Independence of \(U_{n, 1}, \, \ldots,  U_{n, K}\) across centres in \textbf{(A4)} is the more restrictive assumption.  
It may be violated when centres are geographically proximate or otherwise share overlapping catchment areas,  creating inter-centre correlation.   In the Supplementary
Material (Section~\ref{sec:weakAssumptions}) we show that the asymptotic null
distribution of the test statistic remains valid under cross-centre dependence,
provided the full joint covariance $\overline{Q}$ is known. In practice, however,
the off-diagonal blocks of $\overline{Q}$ are not estimable from centre-level
summaries alone, so implementation under dependence requires either matched
individual-level records or external knowledge of the dependence structure. We
therefore maintain~\textbf{(A4)} as the working assumption throughout, and regard
the extension to dependent centres as an open problem.

We will make use of the following assumption :
\begin{description}
    \item[(A5)] For each centre $k, $ there are some matrices $\widehat V_{n, k},  \text{ and }\widehat Q_{n, k}$ that converge to $V_k$ and $Q_k$ with probability tending to $1.$
\end{description}

\citet{lin2011aggregated} defined the aggregated estimating equations (AEE) estimator by weighted aggregation using the sensitivity matrices \(V_k\):
\begin{equation}
\label{eq::aee}
\widehat{\theta}_n
  \;=\;\Bigl(\sum_{k=1}^K \widehat V_{n, k}\Bigr)^{-1}
    \bigl(\sum_{k=1}^K \widehat V_{n, k}\, \widehat{\theta}_{n, k}\bigr).
\end{equation}

\paragraph*{Computational note: AEE vs.\ inverse–variance–weighted (IVW) estimators.}
Under correct specification for regular models,  the Aggregated Estimating Equation (AEE)
and fixed–effect IVW estimators coincide. Computationally,  AEE is lightweight at the
data centres: after fitting,  each centre transmits its local estimate,  sensitivity,  and score
variance,  thereby avoiding per–centre matrix inversions; by contrast,  IVW typically requires
covariance estimates,  whose assembly and inversion are more expensive and brittle on
massive shards. At the aggregator,  both reduce to accumulating one $p\times p$ matrix and
one $p$–vector and solving a single $p\times p$ system,  so central cost is essentially identical.
These features make AEE preferable for online analytical processing and other low–latency
decision systems: point estimation is inexpensive to obtain (only the local estimators and
the sensitivity matrices are required),  and the asymptotic variance is the same as that of the
pooled–data estimator,  see for example \cite[Theorem~5.2]{lin2011aggregated} for generalized linear models maximum likelihood estimator and diverging number of centres and equations \eqref{eq::CANAEE} below in general case for fixed $K$.  Define
\( 
V = \sum_{k=1}^K V_k, \,
Q = \sum_{k=1}^K Q_k, \,
W = V^{-1}Q\, V^{-1}.
\) 
Then,  under assumptions \textbf{(A1)}--\textbf{(A5)} and the null hypothesis of homogeneity 
\(
H_0:\quad \theta_{0, 1} = \theta_{0, 2} = \cdots = \theta_{0, K} = \theta_0, 
\)
we have :
\begin{equation}
\label{eq::CANAEE}
   \sqrt{n}\, \bigl(\widehat{\theta}_n - \theta_0\bigr)
  \;\Rightarrow\;
\mathcal{N}(0, \, W). 
\end{equation}
 
Although violating the homogeneity assumption does not,  in itself,  invalidate the   aggregated estimators,  it may lead to misleading statistical conclusions, particularly when the centres correspond to meaningful subpopulations (e.g.,  administrative units,  hospitals,  or geographic regions). Consider the case of two centres with parameters \(\theta_1\) and \(\theta_2\),  and suppose that the population target is defined as the simple average \(\theta_0 = \tfrac{1}{2}(\theta_1 + \theta_2)\). If the centres are homogeneous,  i.e.,  \(\theta_1 = \theta_2\),  then \(\theta_0 = \theta_1 = \theta_2\),  and aggregation is both valid and interpretable. However,  in the heterogeneous case where \(\theta_2 = -\theta_1 \neq 0\),  the aggregate \(\theta_0 = 0\) is   valid,  yet it masks the presence of two subpopulations with opposing effects. In a regression context,  such cancellation may result in a non-significant coefficient at the population level,  even though each centre exhibits a significant effect in opposite directions. This example illustrates the critical need to assess homogeneity before aggregation: identifying such structural differences is essential for sound statistical inference and informed decision-making. The homogeneity test we propose thus serves as a safeguard to justify aggregation under cluster-level concordance,  while the associated clustering procedure offers a practical tool for uncovering and interpreting heterogeneous structure when centre-specific effects differ.

\paragraph*{Paper structure.}
The remainder of the paper is organized as follows.
Section~2 develops our core methodological contributions: (i) a multivariate Cochran-type homogeneity test; (ii) a new \emph{centre-integration} test; and (iii) the \emph{Clusters-of-Centres} algorithm (\textsc{CoC}-algorithm) (iv) the  control of type-I and type-II errors.
Section~3 is devoted to numerical experiments  
Section~4 presents an application to real data.
Concluding remarks are given in Section~5,  and    the proofs of the main results  and  several applications to some well-known models are collected in the supplementary materials.

\section{On the Homogeneity Across Data centres}
\label{sec:cochran-test}

\subsection{Cochran-Type Tests for   Homogeneity}
\label{subsec:cochran-test}

Define the $Kp\times Kp$ block matrix $H$ blockwise (with each block of size $p\times p$) by
\[
H_{k, j} =
\begin{cases}
I_p,  & \text{if } j\neq k, \\[1mm]
I_p \;-\;\Bigl(\sum_{i=1}^K V_i\Bigr)\, V_k^{-1},  & \text{if } j=k.
\end{cases}
\]

Next,  let $\overline{Q}$ be the $(Kp)\times(Kp)$ block-diagonal covariance matrix with diagonal blocks $Q_1, \dots, Q_K$,  and let $\overline{V}$ be the $(Kp)\times(Kp)$ block-diagonal matrix whose each diagonal block is 
\(
\sum_{k=1}^K V_k.
\)
Finally,  consider the statistic
\[
T_n
\;=\;
\begin{pmatrix}
(\widehat\theta_n - \widehat{\theta}_{n, 1})^\top,  &
(\widehat\theta_n - \widehat{\theta}_{n, 2})^\top,  &
\ldots,  &
(\widehat\theta_n - \widehat{\theta}_{n, K})^\top
\end{pmatrix}^\top.
\]

\begin{lem}
\label{th::testHomogeneity}
Consider the centre-specific asymptotic decomposition \eqref{eq::asymptoticDecomposition} and suppose that the assumptions   \textbf{(A1)} to \textbf{(A4)} hold.  Then,  under the null hypothesis
\(
H_0 : \theta_{0, 1} = \theta_{0, 2} = \cdots = \theta_{0, K} = \theta_0, 
\)
we have
\[
n\, T_n^\top\, \overline{V}^\top\, \overline{V}\, T_n
 \Rightarrow
\sum_{\ell=1}^{Kp} \lambda_\ell\, \chi_\ell^2, 
\]
where \(\{\lambda_\ell\}_{\ell=1}^{Kp}\) are the nonnegative eigenvalues of 
\(\overline{Q}^{1/2}\, H^\top H\, \overline{Q}^{1/2}\) 
and \(\{\chi_\ell^2\}_{\ell=1}^{Kp}\) are independent $\chi^2_1$ random variables. 
\end{lem}

Define
\[
\widehat{\overline Q}_n \;=\; \operatorname{diag}\!\bigl(\widehat Q_{n,1}, \dots, \widehat Q_{n,K}\bigr), \quad
\widehat{\overline V}_n \;=\; \operatorname{diag}\!\Bigl(\sum_{k=1}^K \widehat V_{n, k}, \dots, \sum_{k=1}^K \widehat V_{n, k}\Bigr), 
\]
and construct \(\widehat H_n\) by substituting \(\widehat V_{n, k}\) for \(V_k\) in the definition of \(H\).  The next corollary shows that replacing the unknown quantities by these plug-in estimates preserves the asymptotic size of the test.

\begin{prop}
\label{cor::testHomogemeity}
Consider the centre-specific asymptotic Bahadur representation ~\eqref{eq::asymptoticDecomposition} and suppose that the
assumptions (A1) to (A5) hold. Then,  the null hypothesis $H_0 : \theta_{0, 1} = \cdots = \theta_{0, K} = \theta_0$ is rejected at
asymptotic level~$\alpha$ whenever
\[
n\,  T_n^{\top} \,  \widehat{\overline V}_n^{\top} \,  \widehat{\overline V}_n \,  T_n \; > \; \widehat{q}_{1-\alpha, n}
 \quad \text{where} \quad 
\mathbb{P} \left( \sum_{\ell=1}^{Kp} \widehat{\lambda}_{\ell, n} \,  \chi^2_{\ell} \; \ge \; \widehat{q}_{1-\alpha, n} \right) = \alpha
\]
  where \(\{\widehat \lambda_\ell\}_{\ell=1}^{Kp}\) are the nonnegative eigenvalues of  
\(\widehat {\overline{Q}}^{1/2}_n\, \widehat H^\top_n \widehat H_n\, \widehat {\overline{Q}}^{1/2}_n\) 
and \(\{\chi_\ell^2\}_{\ell=1}^{Kp}\) are independent $\chi^2_1$ random variables. 
\end{prop}\subsubsection{Local alternatives analysis}
\paragraph*{Local (Pitman) alternatives.}
Consider a sequence of contiguous alternatives that deviates
from the null at the usual $n^{-1/2}$ rate:
\begin{equation}
    \label{eq::alt_local}
\exists\,  k \in \{1,  \dots,  K\} : \quad \theta_{0, k} - \theta_0 = \frac{\Delta_k}{\sqrt{n}},  \quad \Delta_k \neq 0.
\end{equation}
Define
$$
\overline W_k = \sum_{j=1}^K V_j \Delta_j - V \Delta_k \in \mathbb{R}^p, 
$$
and the $(Kp)$-vector
\(
\overline W = \big(\overline W_1^{\top},  \dots,  \overline W_K^{\top} \big)^{\top}.
\)
Write the spectral decomposition 
$$
H \overline Q H^{\top} = O \,  \mathrm{diag}(\lambda_1,  \dots,  \lambda_{Kp}) \,  O^{\top}.
$$
The following lemma gives the power of the statistical test of Proposition~1 under
the local Pitman alternatives~\eqref{eq::alt_local}.

\begin{lem}
\label{lem::local_alt_power}
    For the centre-specific decomposition \eqref{eq::asymptoticDecomposition},  suppose that the assumptions (A1) to (A5)
hold. Under the local alternatives \eqref{eq::alt_local},  the   statistic 
\( n\,  T_n^{\top} \widehat{\overline{V}}_n^{\top} \widehat{\overline{V}}_n\,  T_n \) converges in law to  
$$
(H \overline Q^{1/2} Z + \overline W)^{\top} (H \overline Q^{1/2} Z + \overline W)
\quad \text{where} \quad 
Z\overset{d}{=}  \mathcal{N}(0,  I_{Kp}), $$  and  
\[
(H \overline Q^{1/2} Z + \overline W)^{\top} (H \overline Q^{1/2} Z + \overline W) 
\ \overset{d}{=} \
\sum_{j:\, \lambda_j > 0} \lambda_j \,  \chi^2_{j}\!\left(\frac{\delta_j^2}{\lambda_j}\right) 
+ \sum_{j:\, \lambda_j = 0} \delta_j^2, 
\]
where \( (\delta_1,  \dots,  \delta_{Kp})^{\top} := O^{\top} \overline W \).  
For \( j \) with \( \lambda_j > 0 \),  the variables 
\( \chi^2_{j}(\delta_j^2/\lambda_j) \) are independent non-central \(\chi^2_1\) 
with non-centrality parameter \(\delta_j^2/\lambda_j\).
\end{lem}

 \subsection{Cochran-Type Tests for the Fusion of two sets of  Local Estimators}
We now assume that we have   two sets,  $S_1$ and $S_2$ of   homogeneous centres,  i.e.,   
\begin{equation}
\label{eq::homogeneousClusters}
   \forall k \in S_1,  \theta_{0, k} = \theta_0^1 \; \text{ and } \quad \forall k \in S_2,  \theta_{0, k} = \theta_0^2.  
\end{equation}
for some $\theta_0^1,  \theta_0^2 \in \R^p$. Now,  we consider the null hypothesis $\tilde H_0 : \theta_0^1 = \theta_0^2.$
For these sets of clusters,  denote $\widehat\theta_n^1$ and $\widehat\theta_n^2$ their aggregated estimators.
 
The AEE estimator combining $\widehat\theta_n^1$ and $\widehat\theta_n^2$ is
\begin{equation}
\label{eq::aeeIntegration}
\widehat \theta_n^{1, 2}
\;=\;
\Bigl(\sum_{k \in S_1 \cup S_2}  \widehat V_k\Bigr)^{-1}
\Bigl[\bigl(\sum_{k\in S_1} \widehat V_k\bigr)\, \widehat\theta_n^1
  \;+\;
  \bigl(\sum_{k\in S_2} \widehat V_k\bigr)\, \widehat\theta_n^2
\Bigr].
\end{equation}

We consider the statistic
\[
\tilde T_n
\;=\;
\begin{pmatrix}
(\widehat \theta_n^{1, 2} - \widehat\theta_n^1)^\top,  &
(\widehat \theta_n^{1, 2} - \widehat\theta_n^2)^\top
\end{pmatrix}^\top.
\]
This test reduces to the previous one with two clusters.  In this special case,  $H$ becomes
\[
\tilde H
=
\begin{pmatrix}
I_p - \bigl(\sum_{k \in S_1 \cup S_2} V_k\bigr)\bigl(\sum_{k \in S_1} V_k\bigr)^{-1}
  & I_p\\[2mm]
I_p
  & I_p - \bigl(\sum_{k \in S_1 \cup S_2} V_k\bigr)\bigl(\sum_{k \in S_2} V_k\bigr)^{-1}
\end{pmatrix}.
\]
Let $\tilde Q$ be the $(2p)\times(2p)$ block‐diagonal covariance matrix with diagonal blocks $\sum_{k\in S_1} Q_k$ and $\sum_{k\in S_2} Q_k$,  and let $\tilde V$ be the $(2p)\times(2p)$ block‐diagonal matrix whose each diagonal block equals 
\(
\sum_{k \in S_1 \cup S_2} V_k.
\)
Also,  let $\widehat{\tilde Q}_n,  \; \widehat{\tilde V}_n$ and $\widehat{\tilde H}_n$ be defined by replacing $V_k$ and $Q_k$ by with their consistent estimators,  as specified by (\textbf{A5}).

\begin{prop}
\label{prop::integration}
 Consider the centre-specific decomposition~\eqref{eq::asymptoticDecomposition} and suppose that the assumptions
(A1) to (A5) and the equality ~\eqref{eq::homogeneousClusters} hold.  
Then,  under~$\widetilde{H}_0$, 
 \[
n\, \tilde T_n^\top\, \tilde V^\top\, \tilde V\, \tilde T_n
\Rightarrow
\sum_{\ell=1}^{2p}\lambda_\ell\, \chi_\ell^2, 
\]
where \(\{\lambda_\ell\}_{\ell=1}^{2p}\) are the nonnegative eigenvalues of 
\(\tilde Q^{1/2}\, \tilde H^\top\tilde H\, \tilde Q^{1/2}\)
and \(\{\chi_\ell^2\}_{\ell=1}^{2p}\) are independent $\chi^2_1$ random variables.  Consequently,  the homogeneity hypothesis $\tilde H_0$ is rejected at asymptotic level~$\alpha$ when
\[
n\, \tilde T_n^\top\, \widehat {\tilde V}_n^\top\, \widehat {\tilde V}_n\, \tilde T_n
\;>\;
\widehat q_{1-\alpha, n}, 
\text{ and  } \;
\mathbb{P}\Bigl(\sum_{\ell=1}^{2p}\widehat\lambda_{\ell, n}\, \chi_\ell^2\ge \widehat q_{1-\alpha, n}\Bigr)
=\alpha
\]
with   \(\{\widehat \lambda_{\ell, n}\}_{\ell=1}^{2p}\) are the nonnegative eigenvalues of 
\( {\widehat{\tilde Q}_n}^{1/2}\, 
      \widehat {\tilde{H}}_n^{{\!\top}}\widehat {\tilde H}_n\, 
      {\widehat{\tilde Q}_n}^{1/2}\).
\end{prop}

\subsection*{Extension to unequal centre sample sizes}
The preceding results extend   to settings in which centres contribute different sample sizes. For \(k = 1, \ldots, K\),  let \(n_k\) denote the sample size of centre \(k\) and write \(n_k =   n_k(n) \in \N \),  where  \(n_k/n \rightarrow \rho_k \in (0, 1)\) with   \(\sum_{k=1}^{K} \rho_k = 1\) and \(n = \sum_{k=1}^{K} n_k\) is the pooled sample size. The asymptotic expansion~\eqref{eq::asymptoticDecomposition} becomes
\begin{equation}
    \label{eq::diffLenghtBahadur}
      \sqrt{n}\, \bigl(\widehat{\theta}_{n_k, k} - \theta_{0, k}\bigr)
    = V_k^{-1}\, \rho_k^{-1/2}U_{n_k, k} + \rho_k^{-1/2}\varepsilon_{n_k, k} =:  V_k^{-1}\,  \mathcal U_{n_k, k} + o_p(1) .
\end{equation} 
Accordingly,  assumptions \textbf{(A2)}  and \textbf{(A4)} are replaced by
\begin{description}
    \item[(A2')]    \(\mathcal U_{n_k, k}\Rightarrow\mathcal{N}(0, \, Q_k)\) as \(n\to\infty\),  with \(Q_k \ne 0\) is    positive semi-definite;
    \item[(A4')] the random processes  \( \mathcal U_{n_1, 1}, \, \ldots,  \mathcal U_{n_K, K} \) are independent across centres \(k\) at each $n$. 
\end{description}
Although \(n\) now refers to the pooled sample size,  the asymptotic results remain valid  for   \(\min_k n_k \to \infty\) since each  \( n_k \to \infty\) as a  consequence of \(\rho_k > 0\).

\subsection{Clusters of Centres Algorithm}

We now introduce the Clusters of Centres Algorithm (CoC-algorithm). First,  apply the global homogeneity test of Proposition~\ref{cor::testHomogemeity} to all $K$ centres. If that test does not reject homogeneity at level $\alpha$,  declare a single cluster containing every centre and stop. Otherwise,  build the partition sequentially: initialize the first cluster with centre $1$. For each centre $j=2, \dots, K$,  compute the integration test (Proposition~\ref{prop::integration}) comparing centre $j$ to each existing cluster; collect the clusters whose integration p-value with $j$ is at least $\alpha$. If this collection is nonempty,  choose the one of those clusters with the largest p-value; if it is empty,  create a new singleton cluster containing $j$. The one-shot CoC-algorithm described above is given in Algorithm~\ref{alg:CoC-algorithm_updated} in the supplementary material.

 Denote the  true underlying partition of the centres \(\mathcal P = \{\mathcal P_1,  \ldots,  \mathcal P_L\}\),  \(L \leq K\),    \emph{i.e} the partition of the index set \(\{1, \ldots, K\}\),  such that for all \(q,  \overline{q} \in \{1,  \ldots,  L\}\),  \(q \ne \overline{q}\),  we have:
    \[
    \theta_{0, k} = \theta_{0, j},  \quad k\ne j \in \mathcal P_q, 
    \quad\text{and}\quad
    \theta_{0, k} \ne \theta_{0, i},  \quad k \in \mathcal P_q, \; i \in \mathcal P_{\overline q}.
    \]

\begin{lem}\label{lem::CoC-algorithmOracle}
Assume that $L>1$. Let \(\widehat{\mathcal C}_n\) denote the partition estimated by Algorithm \ref{alg:CoC-algorithm_updated},  consisting of \(L_n\) clusters,  and
let \(\widehat{\ell}_{n, k}\in\{1, \ldots, L_n\}\) be the cluster label assigned to the \(k\)-th centre. Suppose that
the assumptions \((\textbf{A1.})\)-\((\textbf{A2})\) -\((\textbf{A3})\)  and \(\textbf{A4}\)-\(\textbf{A5}\) hold. Then,  for all \(q, \bar q\in\{1, \ldots, L\}\) with \(q\neq \bar q\) :
\begin{enumerate}
 \item If $k$ and  $j\, (k\ne j)$ belong to the same set ${\mathcal P}_q$,  then
$
 \limsup_{n\to\infty}\P\big( \widehat{\ell}_{n, k}=\widehat{\ell}_{n, j} \big) \leq 1-\alpha.
 $
\item If $k$ and $j$ belong to different sets ${\mathcal P}_q$ and ${\mathcal P}_{\bar{q}}$,  then
$
\P\big( \widehat{\ell}_{n, k}\neq\widehat{\ell}_{n, j} \big) \, \mathop{\longrightarrow}\limits_{n\to\infty}\,  1.
$
\end{enumerate}
\end{lem}

Lemma \ref{lem::CoC-algorithmOracle} gives a negative result: the probability that Algorithm \ref{alg:CoC-algorithm_updated} recovers the true partition is at most \(1-\alpha\). Indeed,  since the power under non-local alternatives approaches \(1\) and the asymptotic significance level is \(\alpha>0\),  the one-shot CoC algorithm may asymptotically fail to merge homogeneous centres,  separating them into distinct clusters with probability approximately \(\alpha\). However,  the probability of erroneously merging heterogeneous centres tends to zero. To improve on this,  we propose iterating the clustering procedure starting from the initial partition.

That said,  if the first iteration with the same data fails to merge homogeneous clusters,  repeating the aggregation step on the same data will not change the result; the partition stabilises after the first iteration. To circumvent this limitation,  we introduce bootstrap resampling of the summary statistics: only the point estimator of the parameter is recomputed on the bootstrap samples,  while the matrices \(\widehat V_{n, k}\) and \(\widehat Q_{n, k}\) are reused across all bootstrap rounds (see Algorithm~\ref{alg:CoC-algorithm_bootstrap_iterations}). \emph{This choice is communication–efficient and light on computation:} each local centre sends \(\widehat V_{n, k}\) and \(\widehat Q_{n, k}\) once to the coordinating centre and thereafter only the bootstrap point estimates \(\{\widehat\theta_{n, k}^{(r)}\}\); by contrast,  recomputing and transmitting \(\widehat V_{n, k}^{(r)}\) and \(\widehat Q_{n, k}^{(r)}\) each round shares more information and adds extra local work,  with no asymptotic benefit under A5. To keep updates reproducible and limit order effects,  we use a simple tie-breaking rule within each round: whenever multiple candidate merges satisfy \(p\ge\alpha\),  we merge with the cluster attaining the largest \(p\)-value.

Specifically,  the multi-round bootstrap CoC algorithm takes as input \(R\) independent collections of centre-specific summary statistics,  each generated via bootstrap resampling at the individual centres,  any bootstrap scheme that satisfies (\textbf{A6}). In the first round,  the algorithm applies the one-shot CoC procedure to the first summary set to produce an initial clustering. In each subsequent round \(r=2, \dots, R\),  it uses the \(r\)-th summary set to re-evaluate and merge clusters from the initial partition via the aggregation test (Proposition~\ref{prop::integration}) at level \(\alpha\). After all \(R\) rounds,  the algorithm returns the partition from the final iteration. We show that Algorithm~\ref{alg:CoC-algorithm_bootstrap_iterations} enjoys a golden-partition recovery property (see Theorem~\ref{th:CoC-algorithmOracleIterations} below). In particular,  whereas the one-shot CoC algorithm leaves a nonzero probability that a homogeneous pair \((j, k)\) fails to merge,  each bootstrap trial provides a new opportunity for that pair to merge.


\begin{algorithm}[!htb]
\SetAlgoLined
\KwIn{Significance level \(\alpha\); fixed matrices \(\{\widehat V_{n, k}, \, \widehat Q_{n, k}\}_{k=1}^K\);
      \(R\) bootstrap sets of point estimators \(\{\widehat\theta_{n, k}^{(r)}\}_{k=1}^K\) for \(r=1, \dots, R\).}
\KwOut{Final cluster partition \(C^{(R)}\).}

\(C^{(1)} \leftarrow \mathrm{one\_shot\_CoC}\Bigl(\{\widehat\theta_{n, k}^{(1)}\}_{k=1}^K, \;\{\widehat V_{n, k}\}_{k=1}^K, \;\{\widehat Q_{n, k}\}_{k=1}^K, \;\alpha\Bigr)\)\;

\For{\(r \leftarrow 2\) \KwTo \(R\)}{
  Let \(\{B_1, \dots, B_L\} \leftarrow C^{(r-1)}\) \tcp*{clusters from the previous round}
  Initialize \(C^{(r)} \leftarrow \{B_1\}\)\;
  \For{\(\ell \leftarrow 2\) \KwTo \(L\)}{
    \ForEach{cluster \(D \in C^{(r)}\)}{
      \(p_D \leftarrow \mathrm{integration\_test}\Bigl(D, \, B_\ell \, \big|\,  \{\widehat\theta^{(r)}_{n, k}\}_{k=1}^K, \;\{\widehat V_{n, k}\}_{k=1}^K, \;\{\widehat Q_{n, k}\}_{k=1}^K\Bigr)\) \tcp*{two-cluster aggregation test,  Prop.~\ref{prop::integration}}
    }
    \(\mathcal{E} \leftarrow \{\, D \in C^{(r)} : p_D \ge \alpha\, \}\);\;
    \uIf{\(\mathcal{E} = \varnothing\)}{
      \(C^{(r)} \leftarrow C^{(r)} \cup \{B_\ell\}\) \tcp*{start a new cluster}
    }
    \Else{
      \(D^\star \leftarrow \arg\max_{D \in \mathcal{E}}\,  p_D\)\;
      Merge \(B_\ell\) into \(D^\star\)\;
    }
  }
}
\Return{\(C^{(R)}\)}
\caption{Multi-round bootstrap CoC-algorithm }
\label{alg:CoC-algorithm_bootstrap_iterations}
\end{algorithm}

\subsubsection{Golden-partition recovery property of multi-round bootstrap CoC-algorithm}

Fix centres $k$ and $i$ ($k\neq i$) and set
\(
   T=\{1, \dots , K\}\setminus\{k, i\}, \qquad |T|=K-2 .
\)

\noindent A single “$k$–vs–$i$’’ integration test specifies two disjoint sets
\(
      S_k\supseteq\{k\},  \qquad
      S_i\supseteq \{i\},  \qquad
      S_k\cap S_i=\varnothing .
\)

 For every $h\in T$ there are three mutually exclusive options:
(i) ignore $h$,  (ii) add $h$ to $S_k$,  (iii) add $h$ to $S_i$.
These choices are independent across the $K-2$ labels.
The number of distinct assignments is therefore
$3^{\, |T|}=3^{K-2}$.
Because $k\in S_k$ and $i\in S_i$ are fixed and the sets are disjoint, 
each assignment yields a unique ordered pair $(S_k, S_i)$ and vice versa.
Thus exactly $3^{K-2}$ distinct “$k$–vs–$i$’’ tests exist.

 For a truly heterogeneous pair $(k, i)$ and each test $J$ for $P_q$ vs.~$P_{\bar{q}}$ with $k\in P_q$ and $i\in P_{\bar{q}}$, 
define the non-rejection probability
\begin{displaymath}
\beta_{J, n} \, :=\,  \P\big( \mbox{ homogeneity test $J$ accepts homogeneity of $(k, i)$ } \big)
\end{displaymath}
and its bootstrap analogue
\begin{displaymath}
\beta_{J, n}^*(r) \, :=\,  \P^*\big( \mbox{ bootstrap homogeneity test $J$ at round $r$ accepts homogeneity of $(k, i)$ } \big),  \, r=1, \ldots, R.
\end{displaymath}
$\P^*(A)$ denotes the conditional probability $\P(A| \mathcal Z_n)$ given the original sample ${\mathcal 
Z}_n:=(Z_{i, k}\colon\,  1\leq i\leq n,  1\leq k\leq K)$.\\
Under a non-local (fixed) alternative,  $\beta_{J, n}\xrightarrow{p} 0$ as $n\to\infty$ under the assumption  (\textbf{A6}) below.
For a  homogeneous pair $(k, j)$ and each test $J$ for $P_q$ vs.~$P_{\bar{q}}$ with $k\in P_q$ and $j\in P_{\bar{q}}$, 
define the rejection probability
 $\alpha_{J, n} \, :=\,  \P\big( \mbox{ homogeneity test J rejects homogeneity of (k, j) } \big)$
 and its bootstrap analogue
\begin{displaymath}
\alpha_{J, n}^*(r) \, :=\,  \P^*\big( \mbox{ bootstrap homogeneity test $J$ at round $r$ rejects homogeneity of $(k, j)$ } \big),\,  r=1, \ldots, R.
\end{displaymath}
 The probabilities $\alpha_{J, n}^*(1),\ldots, \alpha_{J, n}^*(R)$ depend in a complex way on the original sample; see the lemma \ref{lem::growth} in the Appendices for details. For the sake of conciseness and to avoid confusion,  we will sometimes write 
$\alpha_{J(k, j), n}$ (resp.\ $\alpha^*_{J(k, j), n}(r)$,  $\beta_{J(k, j), n}$,  $\beta^*_{J(k, j), n}(r)$) 
instead of 
$\alpha_{J, n}$ (resp.\ $\alpha^*_{J, n}(r)$,  $\beta_{J, n}$,  $\beta^*_{J, n}(r)$).

We impose  the  following  assumptions.
\begin{description}
 
\item[(A6)]  {\it Conditional bootstrap CLT.}\;
For the sample $\mathcal Z_n$ and its $r-th$ bootstrap resample     $\mathcal Z_n^{(r)}, $  let $\widehat\theta_{n, k}$ (resp. $\widehat\theta_{n, k}^{(r)}$) the
      centre–specific estimator  computed from $\mathcal Z_n $ (resp. $\mathcal Z_n^{(r)}$)  and 
       assume that, 
       conditionally on $\mathcal Z_n, $  in probability,
       $$\sqrt n(\widehat\theta^{(r)}_{n, k}-\widehat\theta_{n, k})\Rightarrow\mathcal  N(0, V_k^{-1}Q_kV_k^{-1}).$$
      
\end{description}

The assumption \textbf{A6} is mild and ensures the validity of the bootstrap procedure. It is satisfied,  for instance,  by the following  universal resampling scheme (that is independent of the specific estimator) : for each bootstrap replicate \(r\),  each sample size $n$ and centre \(k\),  take $\widehat{\theta}^{(r)}_{n, k}$ as in 
\[
\sqrt{n}\big(\widehat{\theta}^{(r)}_{n, k}-\widehat{\theta}_{n, k}\big)
\ \overset{d}{=}\ 
\mathcal{N}\!\left(0, \ \widehat{V}_{n, k}^{-1}\, \widehat{Q}_{n, k}\, \widehat{V}_{n, k}^{-1}\right).
\]
More generally,  \textbf{A6} is also satisfied by standard bootstrap schemes under classical regularity conditions. For example,   the nonparametric bootstrap is consistent under mild assumptions (see,  e.g.,  \cite[Sec.~10.3 and Sec.~13.2.3]{kosorok2008introduction})
in parametric models for estimating-equation estimators. 
For the weighted (multiplier) bootstrap,  where the weights $\xi_1, \ldots, \xi_n$ are \emph{i.i.d.} copies of a square–integrable random variable $\xi$ with mean $1$ and variance $1$, \textbf{A6} holds for both parametric and semiparametric models. In the latter setting,  one may consider profile estimating-equation estimators with within-centre data assumed \emph{i.i.d.}; see \citep[Theorem~21.7]{kosorok2008introduction}. In addition,  \textbf{A6} holds for classes of M-estimators under the regularity conditions stated by \cite{arcones1992bootstrap} for parametric and semiparametric bootstrap. For instance,  \cite{hahn1995bootstrapping} uses \cite{arcones1992bootstrap} to establish \textbf{A6} for quantile regression under a semiparametric bootstrap for deterministic covariates and nonparametric bootstrap for random covariates.

We are now in position to state our Golden-Partition Recovery property. We will need of the following assumption.

\begin{description}
    \item[(A7)] {\it Growth of the number of rounds of CoC-algorithm.} As $n$ tends to $\infty, $
    $$
    R(n) \rightarrow +\infty, \quad 
     R(n)\max_{S_i,  S_j} \max_{i, j} \beta^*_{J(i, j), n}(1)    = O(1)\text{ a.s }  \text{ and }\quad R(n)\max_{S_i,  S_j} \max_{i, j} \beta^*_{J(i, j), n}(1)  \xrightarrow{p} 0.
    $$
\end{description}
\begin{theo} 
\label{th:CoC-algorithmOracleIterations}
Let   $\widehat{C}_n^{(r)}$ denote the partition produced by Algorithm \ref{alg:CoC-algorithm_bootstrap_iterations}
after $r$ bootstrap iterations.
If \((\textbf{A1})\) to \((\textbf{A7})\)   are fulfilled,  then
$
\P\big( \widehat{C}_n^{(R(n))} \, =\,  \mathcal P \big) \, \mathop{\longrightarrow}\limits_{n\to\infty}\,  1
$
where ${\mathcal P}$ denotes the true underlying partition.
\end{theo}

\paragraph*{Comment on Theorem \ref{th:CoC-algorithmOracleIterations}.}
Theorem \ref{th:CoC-algorithmOracleIterations} is   asymptotic and therefore provides limited guidance on the finite-sample behavior of our procedure. Actually, under the separation condition
\(
\theta_{0,k} = \theta_{0,j}, \quad k \neq j \in \mathcal P_q,
\quad\text{and}\quad
\theta_{0,k} \neq \theta_{0,i}, \quad k \in \mathcal P_q,\; i \in \mathcal P_{\overline q},
\)
we expect that, for sufficiently large $n$, the power of the fusion test is close to $1$, so that two distinct clusters of centres can be consistently distinguished. However, for a finite sample size $n$, the probability of a false merge remains strictly positive. A similar limitation arises when considering the test level: since the finite-sample level   is always strictly smaller than $1$, increasing the number of replications may help fuse two centres that truly belong to the same cluster. Nevertheless, Theorem~\ref{th:CoC-algorithmOracleIterations} does not quantify the probability of such correct fusion in finite samples. For this reason, we provide below a detailed analysis of the type-I and type-II errors in the CoC-algorithm.

\noindent In what follows, the class of sets $\mathcal{A}$ is the collection of Borel subsets of $\mathbb{R}^p$ satisfying
\(
\sup_{B \in \mathcal{A}} \Phi((\partial B)^{\varepsilon}) = O(\varepsilon) \quad \text{as} \quad \varepsilon \downarrow 0.
\)
Here $\Phi$ denotes the normal distribution with mean $0$ and dispersion matrix being the identity matrix. The set $(\partial B)^{\varepsilon}$ is defined as the $\varepsilon$-neighborhood of the boundary of $B$:
\(
(\partial B)^{\varepsilon} = \left\{ x \in \mathbb{R}^p : \inf_{y \in \partial B} \|x - y\| \le \varepsilon \right\}.
\) The set of convex sets of $\R^p$, for instance, satisfies this property.

\subsubsection{Control of the type-I and type-II errors}
We rely on the following deviation and Berry--Esseen-type bounds, rather than on the Bahadur decomposition in~\eqref{eq::asymptoticDecomposition}.

\begin{description}
    \item[(A8)] For each centre $k$, assume that there exist constants $a_k>0$   such that
\begin{subequations}\label{eq::nonAsymptoticConditions}
\begin{align}
\P\!\left(\|\widehat{\theta}_{n,k}-\theta_{0,k}\|
\le a_k\sqrt{\frac{\log n}{n}}\right)
&= 1 - o\!\left(n^{-1/2}\right),
\label{eq::concentration}\\
\sup_{B\in \mathcal A}
\Big|\P\!\left(\widehat{H}_{n,k}\in B\right)-\Phi(B)\Big|
&= O\!\left(n^{-1/2}\log n\right),
\label{eq::BerryEssen}
\end{align}
\end{subequations}
where
\(
\widehat{H}_{n,k}
:= \sqrt{n}\,\widehat{V}_{n,k}\big(\widehat{\theta}_{n,k}-\theta_{0,k}\big),
\)
and $\Phi$ denotes the distribution of the corresponding Gaussian limit.
\end{description}Define
\begin{equation}\label{eq::Hagg_theta}
\widehat{H}_n
:=\frac{1}{\sqrt K}\sum_{k=1}^K \widehat{H}_{n,k}
 \end{equation}
Here, we specifically adopt the joint bounds~\eqref{eq::concentration}--\eqref{eq::BerryEssen}, which are derived simultaneously in \citet{bhattacharya1978validity} for Z-estimators. We emphasize that \citet{bhattacharya1978validity} states the result for a \emph{local} solution of the estimating equation; under convexity of the contrast function, this local solution coincides with the global minimizer. This setting includes, for instance, one-parameter generalized linear models with canonical link. See, e.g., Theorem~1 of \citet{das2025pebble} for logistic regression.
Finally, the Berry--Esseen bound in~\eqref{eq::BerryEssen} can be sharpened to $o(n^{-1/2})$ under additional regularity conditions, as the Cramer's condition on characteristic function for  non-lattice distribution; see, for example, \citet[Prop.~2.1 or Thm.~2.2]{lahiri1994two} or \citet[Thm.~3]{bhattacharya1978validity}. In contrast, logistic regression corresponds to a lattice-type setting, and \citet[Thm.~1]{das2025pebble} establishes the rate $O(n^{-1/2}\log n)$, while their proposed bootstrap attains $o(n^{-1/2})$. In the present work, we adopt the least favorable rate $O(n^{-1/2}\log n)$ in order to cover both lattice and non-lattice cases within a unified framework.

We also require the  inequality for the random matrix
$\widehat V_{n,k}$.
\begin{description}
    \item[(A9)] For each $k$, assume that there exists a constant $b_k>0$ such that
\begin{equation}\label{eq::concentrationV}
\P\!\left(\|\widehat{V}_{n,k}-V_{k}\|_{\op}
\le b_k\sqrt{\frac{\log n}{n}}\right)
= 1 - o\!\left(n^{-1/2}\right),
\end{equation}
\end{description}

for some positive-definite matrix $V_k$.

A sufficient set of conditions for~\eqref{eq::concentrationV} is given as follows. Suppose that
\(
\widehat{V}_{n,k}=\frac1n\sum_{i=1}^n v_k\!\left(Z_{i,k},\widehat{\theta}_{n,k}\right),
\qquad
V_k:=\E\!\left[v_k\!\left(Z_{1,k},\theta_{0,k}\right)\right].
\)
\begin{enumerate}
\item[\textnormal{SC1.}]
For any $\theta_1,\theta_2$ and any $z$,
\(
\big\|v_k(z,\theta_1)-v_k(z,\theta_2)\big\|_{\op}
\le w_k(z)\,\|\theta_1-\theta_2\|,
\)
for some nonnegative measurable function $w_k$.

\item[\textnormal{SC2.}]
For each $k$, the observations $(Z_{i,k})_{i=1}^n$ are i.i.d., $\E[w_k(Z_{1,k})]>0$, and there exists $s\geq 3$ such that
\(
\E\!\left[w_k(Z_{1,k})^s\right]<\infty,
\qquad
\E\!\left[\big\|v_k(Z_{1,k},\theta_{0,k})\big\|_{\op}^s\right]<\infty.
\)
\end{enumerate}

\begin{lem}\label{lem::concentrationV}
Assume that~\eqref{eq::concentration} holds for $\widehat{\theta}_{n,k}$, and that \textnormal{SC1--SC2} hold.
Then \eqref{eq::concentrationV} holds for some constant $b_k>0$.
\end{lem}
\[\text{Define  }\quad
\overline{\widehat V}_n:=\frac1K\sum_{k=1}^K \widehat V_{n,k},
\qquad
\overline V:=\frac1K\sum_{k=1}^K V_k,
\qquad
r_n:=\sqrt{\frac{\log n}{n}},
\qquad
\theta_{0,V}:= V^{-1}\sum_{k=1}^K V_k\,\theta_{0,k}.
\]

 \begin{prop}\label{prop::AEE_conc_BE}
Assume that, for each $k\in\{1,\ldots,K\}$, (A8) and
 (A9) hold.
The following hold.

\noindent\textnormal{(i).}
There exists a constant $C>0$ such that
\begin{equation}\label{eq::conc_theta_agg}
\P\!\left(\|\widehat{\theta}_n-\theta_{0,V}\|\le C\,r_n\right)
=1-o\!\left(n^{-1/2}\right).
\end{equation}

 \noindent\textnormal{(ii).}
We have
\begin{equation}\label{eq::BE_theta_agg}
\sup_{B\in\mathcal  A}
\big|\P(\widehat{H}_n\in B)-\Phi(B)\big|
=
O\!\left(n^{-1/2}\log n\right).
\end{equation}

 \noindent
In particular, when
$\theta_{0,1}=\cdots=\theta_{0,K}= \theta_0$, we have $\theta_{0,V}=\theta_0$ and
\(
\widehat{H}_n
=
\sqrt{nK}\,\overline{\widehat V}_n\big(\widehat{\theta}_n-\theta_0\big)
 \)
and  \eqref{eq::Hagg_theta}--\eqref{eq::BE_theta_agg} coincide with the homogeneous-case formulation.
\end{prop}

 This result  makes explicit how
\emph{local}  large sample results stated as an asymptotic  deviation bound and a Berry--Esseen bound for each
$\widehat\theta_{n,k}$ together with  the bound for $\widehat V_{n,k}$ \emph{propagate} to the aggregated estimator
$\widehat\theta_n$.
A key feature is the separation of the arguments: the deviation of $\widehat\theta_n$ is obtained from the local deviation
of $\widehat\theta_{n,k}$ combined with the control of $\widehat V_{n,k}$, whereas the Berry--Esseen bound for the standardized aggregated statistic is deduced
directly from the local Berry--Esseen bounds through a telescoping   argument.

Let
$
\bar A_0 = \|V\|_{\op} \Big(\sum_{k=1}^K a_k^2\Big)^{1/2}
$ and fix $\eta \in (0,1)$, and define \[
c_+(\eta):=\frac{\lambda_{\max}(V)}{(1-\eta)\lambda_{\min}(V)},
\quad
c_-(\eta):=\frac{\lambda_{\min}(V)}{\lambda_{\max}(V)+\eta\lambda_{\min}(V)}, \quad c_p:=\frac{2^{1-p/2}}{\Gamma(p/2)}  \text{ and }
\]
\begin{align*}
 \bar A_1
&:=
 \|V\|_{\op}
\Big(\sqrt K\,C+\Big(\sum_{k=1}^K a_k^2\Big)^{1/2}\Big)
 +
\Big(\sum_{k=1}^K b_k\Big) \,
\Big(\sum_{k=1}^K \|\theta_{0,V}-\theta_{0,k}\|^2\Big)^{1/2}.
\end{align*}
Define, for all $n,$
\[
\widehat W_n
\, =\,  \left( \begin{array}{c}
\widehat V_n \sqrt{n} \,(\widehat{\theta}_n-\widehat{\theta}_{n, 1}) \\
\widehat V_n \sqrt{n} \,(\widehat{\theta}_n-\widehat{\theta}_{n, 2}) \\ \vdots \\
\widehat V_n \sqrt{n} \,(\widehat{\theta}_n-\widehat{\theta}_{n, K})
\end{array} \right),
\qquad
W_{n,0}
=
\left(
\begin{array}{c}
V\sqrt n(\theta_{0,V}-\theta_{0,1})\\
V\sqrt n(\theta_{0,V}-\theta_{0,2})\\
\vdots\\
V\sqrt n(\theta_{0,V}-\theta_{0,K})
\end{array}
\right)
\]  and
\begin{equation}\label{eq::A0_def}
 \Delta_V
=
\left(\sum_{k=1}^K \|V(\theta_{0,V}-\theta_{0,k})\|^2\right)^{1/2}.
\end{equation}

\begin{theo}\label{thm::tail_bounds_WhatV2}
  Assume, for each
$k\in\{1,\ldots,K\}$,(A8) and
 (A9) hold.

1. Assume  $\theta_{0,1}=\cdots=\theta_{0,K}= \theta_0$.   Let $(u_n)_{n\ge1}$ be any deterministic bounded sequence such as $A_0 = \lim\inf u_n > 2\bar A_0$.
\text{Set  }
$
\, t_n:=u_n\sqrt{\log n}.
$
Then,   for any $\eta \in (0,1),$
\begin{align}
\P\big(\|\widehat W_n\|\ge t_n\big)
&\le
c_p\left(\frac{t_n}{2Kc_+(\eta)}\right)^{p-2}
\exp\!\left(-\frac{t_n^2}{8K^2c_+(\eta)^2}\right)
\Big(1+o(1)\Big)
+
O\!\left(n^{-1/2}\log n\right)
+
o\!\left(n^{-1/2}\right),
\label{eq::upper_tail_What_tn_V2}\\[3pt]
\P\big(\|\widehat W_n\|\ge t_n\big)
&\ge
c_p\left(\frac{3t_n}{2Kc_-(\eta)}\right)^{p-2}
\exp\!\left(-\frac{9t_n^2}{8K^2c_-(\eta)^2}\right)
\Big(1+o(1)\Big)
-
O\!\left(n^{-1/2}\log n\right)
-
o\!\left(n^{-1/2}\right).
\label{eq::lower_tail_What_tn_V2}
\end{align}
2. Assume that there exists at least one $k\neq 1$ such that
\(
\theta_{0,k}\neq \theta_{0,1}
\). For any deterministic   bounded sequence $(u_n)_{n\ge 1}$ such that $\lim\inf u_n >  \bar A_1$,  set
$
t_n  = \sqrt n\,\Delta_V+u_n\sqrt{\log n}.
$ Then,
$
\P(\|\widehat W_n\|\le t_n)\ge 1-o(n^{-1/2}).
$
In particular,  for any fixed $\varepsilon\in(0,1)$,
\begin{align}
\P\!\left(\|\widehat W_n\|\le (1-\varepsilon)\sqrt n\,\Delta_V\right)
&=
o(n^{-1/2}),
\label{eq::below_relative_case2}\\
\P\!\left(\|\widehat W_n\|\le (1+\varepsilon)\sqrt n\,\Delta_V\right)
&=
1-o(n^{-1/2}).
\label{eq::above_relative_case2}
\end{align}
\end{theo}

\paragraph*{Comment (detectability threshold).}
In the heterogeneous case, the statistic $\|\widehat W_n\|$ concentrates around the
deterministic \emph{signal level} $\|W_{n,0}\|=\sqrt n\,\Delta_V$, up to fluctuations of order
$\Gamma_{n,\mathrm{alt}}=O(\sqrt{\log n})$ with failure probability $o(n^{-1/2})$.
In contrast, under homogeneity, $\|\widehat W_n\|$ lives on a much smaller \emph{noise scale}
(of order at most $\sqrt{\log n}$), so any level-$\alpha$ critical value $c_{n,\alpha}$
calibrated under $H_0$ is typically of order $c_{n,\alpha}=O(\sqrt{\log n})$.

\noindent This yields a transparent  detectability condition.
Indeed, by the reverse triangle inequality,
\(
\|\widehat W_n\|\;\ge\;\|W_{n,0}\|-\|\widehat W_n-W_{n,0}\|,
\)
so that, for any $t>0$,
\(
\big\{\|\widehat W_n-W_{n,0}\|\le t\big\}\cap \big\{\|W_{n,0}\|\ge c_{n,\alpha}+t\big\}
\ \subseteq\
\big\{\|\widehat W_n\|\ge c_{n,\alpha}\big\}.
\)
Since $\|W_{n,0}\|=\sqrt n\,\Delta_V$ is deterministic, whenever
\(
\sqrt n\,\Delta_V\;\ge\;c_{n,\alpha}+\Gamma_{n,\mathrm{alt}}
\),
we obtain
\begin{align*}
\P_{H_1}\!\big(\|\widehat W_n\|\ge c_{n,\alpha}\big)
&\ge
\P_{H_1}\!\Big(\|\widehat W_n-W_{n,0}\|\le \Gamma_{n,\mathrm{alt}}\Big)\\
&=
1-\P_{H_1}\!\Big(\|\widehat W_n-W_{n,0}\|>\Gamma_{n,\mathrm{alt}}\Big)
=
1-o(n^{-1/2}),
\end{align*}
where one may take, for instance,
\begin{align}
\Gamma_{n,\mathrm{alt}}
&:=
\Big(\|V\|_{\op}+\Big(\sum_{k=1}^K b_k\Big)r_n\Big)
\Big(\sqrt K\,C+\Big(\sum_{k=1}^K a_k^2\Big)^{1/2}\Big)\sqrt{\log n}
\nonumber\\
&\quad+
\Big(\sum_{k=1}^K b_k\Big)\sqrt{\log n}\,
\Big(\sum_{k=1}^K \|\theta_{0,V}-\theta_{0,k}\|^2\Big)^{1/2},
\label{eq::Gamma_alt_def_comment}
\end{align}
as derived in the proof (see the bound on $\|\widehat W_n-W_{n,0}\|$ under $H_1$).

\noindent Equivalently, the test can reliably detect departures from homogeneity as soon as
\(
\Delta_V\;\ge\;\frac{c_{n,\alpha}+\Gamma_{n,\mathrm{alt}}}{\sqrt n}.
\)
In particular, since $c_{n,\alpha}=O(\sqrt{\log n})$ and $\Gamma_{n,\mathrm{alt}}=O(\sqrt{\log n})$,
this detectable separation is  of order
\(
\Delta_V \;\geq \; \kappa \sqrt{\frac{\log n}{n}}
\)   for some positive $\kappa.$

\subsubsection{CoC-algorithm with  shrinkage rejection region}

Now, we impose the following conditions for the bootstrap procedure, analogous to
\eqref{eq::concentration}--\eqref{eq::BerryEssen}--\eqref{eq::concentrationV}.

\begin{description}
    \item[(A10)]  For each centre $k$ and each bootstrap round $r$, assume that there exists a constant $a_k^*>0$ such that
\begin{subequations}
\begin{align}
\P\!\left(\|\widehat{\theta}_{n,k}^{(r)}-\widehat \theta_{n,k}\|
\le a_k^*\sqrt{\frac{\log n}{n}} \,\middle|\, \mathcal Z_n\right)
&= 1 - o_p\!\left(n^{-1/2}\right),
\label{eq::concentrationBoot}\\
\sup_{B\in \mathcal A}
\Big|\P\!\left(\widehat{H}_{n,k}^{(r)}\in B \,\middle|\, \mathcal Z_n \right)-\Phi(B)\Big|
&= O_p\!\left(n^{-1/2}\log n\right),
\label{eq::BerryEssenBoot}
\end{align}
\end{subequations}
where $\widehat{H}_{n,k}^{(r)}:=\sqrt{n}\,\widehat{V}_{n,k}^{(r)}\big(\widehat{\theta}_{n,k}^{(r)}-\widehat \theta_{n,k}\big)$.
\end{description}

We also require the following matrix concentration.
\begin{description}
    \item[(A11)] For each $k$ and $r$, assume that there exists a constant $b_k^*>0$ such that
\begin{equation}\label{eq::concentrationVBoot}
\P\!\left(\|\widehat{V}_{n,k}^{(r)}-\widehat V_{n,k}\|_{\op}
\le b_k^*\sqrt{\frac{\log n}{n}} \,\middle|\, \mathcal Z_n\right)
= 1 - o_p\!\left(n^{-1/2}\right),
\end{equation}
\end{description}
which is also readily obtained under conditions similar to SC1 and SC2.

Conditions (A10) and (A11) yield the existence of some constant $C^*$ such that for any round $r,$
$
\P\!\left(\|\widehat{\theta}_n^{(r)}-\widehat \theta_{n}\|\le C^*\,r_n\middle|\mathcal Z_n \right)
=1-o\!\left(n^{-1/2}\right).
$

Denote   $\mathcal C$ is the collection of homogeneous cluster pairs and $\mathcal D$ is the collection of inhomogeneous cluster pairs. Let $ \widehat V^*_n = \sum_{k=1}^K \widehat V_{n,k}, \,
\bar A_0^* = \|\widehat V^*_n\|_{\op} \Big(\sum_{k=1}^K a_k^{*2}\Big)^{1/2}
$ and for $\eta \in (0,1)$, \[
c_+^*(\eta):=\frac{\lambda_{\max}(\widehat V^*_n)}{(1-\eta)\lambda_{\min}( \widehat V^*_n)},
\quad
c_-^*(\eta):=\frac{\lambda_{\min}( \widehat V^*_n)}{\lambda_{\max}(\widehat V^*_n)+\eta\lambda_{\min}(\widehat V^*_n)} \text{ and }
\]
\begin{align*}
 \bar A_1^*
&:=
 \|\widehat V^*_n\|_{\op}
\Big(\sqrt K\,C^*+\Big(\sum_{k=1}^K a_k^{*2}\Big)^{1/2}\Big)
 +
\Big(\sum_{k=1}^K b_k^*\Big) \,
\Big(\sum_{k=1}^K \|\widehat \theta_{n}-\widehat \theta_{n,k}\|^2\Big)^{1/2}.
\end{align*}

For all $n,$ for each pair of disjoint clusters $(S_i,S_j)$, define the bootstrap Cochran-type statistic $Q_{n,r}^{(i,j)} = \|\widehat W_{n,r}^{(i,j)}\|^2$  where    at a fixed bootstrap round $r$,
 \(
\widehat V^{(i)}_{n,r}:=\sum_{k\in S_i}\widehat V_{n, k}^{(r)}, \quad
\widehat V^{(j)}_{n,r}:=\sum_{k\in S_j}\widehat V_{n, k}^{(r)}, \quad
\widehat V^{(i, j)}_{n,r}:=\widehat V^{(i)}_{n,r}+\widehat V^{(j)}_{n,r},
\)
 the AEE  estimators by
\(
\widehat V^{(i)}_{n,r}\, \widehat{\widehat\theta}^{(i)}_{n,r}=\sum_{k\in S_i}\widehat V_{n, k}^{(r)}\, \widehat\theta^{(r)}_{n, k}, \quad
\widehat V^{(j)}_{n,r}\, \widehat{\widehat\theta}^{(j)}_{n,r}=\sum_{k\in S_j}\widehat V_{n, k}^{(r)}\, \widehat\theta^{(r)}_{n, k},
\)
and the non-bootstrap versions $\widehat{\widehat\theta}^{(i)}_n, \widehat{\widehat\theta}^{(j)}_n$ where  $\widehat\theta_{n, k}$ replaces $\widehat\theta^{(r)}_{n, k}$;  the combined AEE satisfies
\(
\widehat V^{(i, j)}_{n,r}\, \widehat{\widehat\theta}^{(i, j)}_{n,r}
= \widehat V^{(i)}_{n,r} \widehat{\widehat\theta}^{(i)}_{n,r} + \widehat  V^{(j)}_{n,r} \widehat{\widehat\theta}^{(j)}_{n,r}.
\) and consequently define its non-bootstrap version    $\widehat{\widehat\theta}^{(i, j)}_{n}$. We set $\widehat V^{(i,j)}_{n} = \sum_{k  \in S_i \cup S_j} \widehat V_{n,k},$
\[
\widehat W_{n,r}^{(i,j)}
\, =\,  \left( \begin{array}{c}
\widehat V_{n,r}^{(i,j)} \sqrt{n} \,(\widehat {\widehat{\theta}}_{n,r}^{(i,j)}-\widehat {\widehat{\theta}}_{n, r}^{(i)}) \\
\widehat V_{n,r}^{(i,j)} \sqrt{n} \,(\widehat {\widehat{\theta}}_{n,r}^{(i,j)}-\widehat {\widehat{\theta}}_{n, r}^{(j)})
\end{array} \right), \,  \widehat \Delta_V^{(i,j)}
=
\left(  \|\widehat V^{(i,j)}_{n}(\widehat{\widehat\theta}^{(i, j)}_{n}-\widehat{\widehat\theta}^{(i)}_n)\|^2 +  \|\widehat V^{(i,j)}_{n}(\widehat{\widehat\theta}^{(i, j)}_{n}-\widehat{\widehat\theta}^{(j)}_n)\|^2\right)^{1/2}
  \]
 \begin{align*}
\text{ and } \quad \Gamma_{n,\mathrm{alt}}^*
& = \left[\bar A_1^*
 +r_n \Big(\sum_{k=1}^K b_k^*\Big)
\Big(\sqrt K\,C^*+\Big(\sum_{k=1}^K a_k^{*2}\Big)^{1/2}\Big)\right]\sqrt{\log n}.
 \end{align*}

\begin{theo}\label{theo:boot_merge_split_control}
Assume (A8) to (A11).
Consider the deterministic threshold $t_n=u_n\sqrt{\log n}$ with a bounded sequence
$u_n$ that satisfies, conditionally on $\mathcal Z_n$,
$\liminf u_n > \max(2\bar A_0^*,\bar A_1^*)$,
and fix any number of bootstrap rounds $R(n)\ge 1$.
Define the rejection event
\(
\mathcal R_{n,r}^{(i,j)}:=\{Q_{n,r}^{(i,j)} > t_n^2\}.
\)
Let
\[
\overline\alpha_n
:=\max_{1\le r\le R(n)}\ \max_{(S_i,S_j)\in\mathcal C}\
\P\!\big(\mathcal R_{n,r}^{(i,j)} \,\big|\, \mathcal Z_n\big),
\quad
\overline\beta_n
:=\max_{1\le r\le R(n)}\ \max_{(S_i,S_j)\in\mathcal D}\
\P\!\big((\mathcal R_{n,r}^{(i,j)})^{c} \,\big|\, \mathcal Z_n\big).
\]

\medskip
\noindent\textnormal{(i) Uniform false splitting (type-I) bound.}
For any $\eta \in (0,1)$, define
\[
\zeta_{n}^{\mathrm{I}}(\eta)
:=
c_p\Big(\frac{t_n}{2Kc_+^*(\eta)}\Big)^{p-2}
\exp\!\Big(-\frac{t_n^2}{8K^2c_+^*(\eta)^2}\Big).
\]
Then
\begin{equation}\label{eq::type1_bound}
\P\!\Big(\,\overline\alpha_n
\;\le\;
\zeta_{n}^{\mathrm{I}}(\eta)\,\bigl(1+o_p(1)\bigr)
+
O_p\!\big(n^{-1/2}\log n\big)
+
o_p\!\big(n^{-1/2}\big)\Big)
\;\longrightarrow\; 1.
\end{equation}

\medskip
\noindent\textnormal{(ii) Uniform false merging (type-II) bound.}
Uniformly over $1\le r\le R(n)$ and $(S_i,S_j)\in\mathcal D$,
\begin{equation}\label{eq::type2_concentration}
\P\!\Big(\big|\,\|\widehat W_{n,r}^{(i,j)}\|-\sqrt n\,\widehat \Delta_{V}^{(i,j)}\,\big|
\le \Gamma_{n,\mathrm{alt}}^{*}
\,\Big|\,\mathcal Z_n\Big)
\;=\;
1-o_p\!\big(n^{-1/2}\big).
\end{equation}
Assume in addition that there exists a deterministic sequence $\underline\Delta_n>0$
such that, for every $(S_i,S_j)\in\mathcal D$, conditionally on $\mathcal Z_n$,
\(
\widehat \Delta_{V}^{(i,j)}\ge \underline\Delta_n.
\)
If the detectability condition
\begin{equation}\label{eq::detectability_boot}
\sqrt n\,\underline\Delta_n \;\ge\; t_n+\Gamma_{n,\mathrm{alt}}^{*}
\end{equation}
holds conditionally on $\mathcal Z_n$ for all large $n$, then
\begin{equation}\label{eq::type2_bound}
\P\!\Big(\,\overline\beta_n
\;\le\;
o_p\!\big(n^{-1/2}\big)\Big)
\;\longrightarrow\; 1.
\end{equation}
\end{theo}

\paragraph*{Comment on Theorem~\ref{theo:boot_merge_split_control}.}
Unlike our previous setup, the rejection region in this theorem shrinks as $n$ increases.
Nevertheless, both error rates vanish in probability simultaneously, ensuring that the
true partition is recovered with probability tending to $1$. In finite samples, the
multi-round algorithm reduces the false splitting of homogeneous centres, though this
comes at the cost of increased false merging of heterogeneous centres. This trade-off
is most pronounced when the clusters are not well separated, specifically when the
signal level $\widehat \Delta_V^{(i,j)}$ does not substantially exceed $\underline \Delta_n$.
Regarding the number of rounds, since $\overline\beta_n = o_p(n^{-1/2})$ by
part~(ii), the accumulated type-II error over $R(n)$ rounds satisfies
$R(n)\,\overline\beta_n = o_p(1)$ as soon as $R(n) = o_p(\sqrt{n})$.
In particular, any sequence $R(n)\to\infty$ with $R(n)/\sqrt{n}\to 0$ is sufficient
to ensure that both error rates vanish while still providing the golden-partition
recovery guarantee of Theorem~\ref{th:CoC-algorithmOracleIterations}.

\section{Simulation study}
\label{subsec:mc-CoC-algorithm-ari}

\textbf{Data-generating process (DGP).}
We simulate $K$ distributed centres partitioned into $L$    true clusters : 
centre $k$ belongs to a true cluster label $c(k)\in\{1,\dots,L\}$.
Within each cluster $\ell$, centers share a common logistic model
\(
\Pr(Y=1\mid X)=\sigma(X^\top \beta_\ell),\quad \sigma(t)=\frac{1}{1+e^{-t}} .
\)
Each center $k$ contains $n$ i.i.d.\ observations $(X_{ki},Y_{ki})_{i=1}^n$ with
\(
X_{ki}=(1, Z_{ki}^\top)^\top\in\mathbb{R}^{p},
\quad Z_{ki}\sim \mathcal{N}(0,\Sigma_\rho),
\quad (\Sigma_\rho)_{ab}=\rho^{|a-b|} \ \ (a,b\in\{1,\dots,p-1\}),
\)
so the covariates follow an AR(1)-type correlation structure controlled by $\rho$.
Cluster-specific coefficients are constructed as
\(
\beta_\ell=\beta_0+\delta\, d_\ell,
\)
where $\beta_0$ is a baseline vector (with an intercept component), $d_\ell$ is a random direction
(with zero intercept component) normalized to unit Euclidean norm, and $\delta>0$ controls
the separation between clusters (larger $\delta$ means better-separated clusters).

\textbf{Local estimation and Hessian.}
At each center $k$, we fit a logistic regression by MLE, producing a local estimator $\hat\theta_k$
and its   Hessian matrix
\(
\hat V_k=\frac{1}{n}\,X_k^\top D_k X_k,
\quad D_k=\mathrm{diag}\big(\hat p_{ki}(1-\hat p_{ki})\big)_{i=1}^n,
\quad \hat p_{ki}=\sigma(X_{ki}^\top \hat\theta_k).
\)

\textbf{Bootstrap scheme.}
For each center $k$, we run an \emph{i.i.d.\ nonparametric bootstrap}:
for each bootstrap round $r\in\{1,\dots,R\}$, we resample $n$ observations with replacement from the
centre's data, refit the local logistic model, and obtain bootstrap replicates
\(
\hat\theta_k^{(r)} \quad \text{and} \quad \hat V_k^{(r)}.
\)

 We run  the experiment  on two configurations:
\(
(K,L)\in\{(20,4),\ (40,6)\},
\)
with fixed $p=12$, $\rho=0.3$, and balanced cluster sizes.
We vary:
\(
n\in\{500,1000,2000,5000\},\quad
\delta\in\{0.5,1.0\},\quad
u_n\in\{1,2,4\},\quad
R\in\{50,100\}.
\)
The threshold used by the CoCA merging rule is
\(
t_n = u_n \sqrt{\log n}.
\)
For each grid point, we repeat the simulation for $B = 100$   Monte Carlo replicates, generating independent datasets and bootstrap resamples each time.

\textbf{Evaluation metrics.}
The  averages across the Monte-Carlo experiments of the  following metrics are reported at each grid point:
   {ARI (Adjusted Rand Index):} $\mathrm{ARI}(\widehat{C}_n^{(R(n))} \, ,   \mathcal P )\in[-1,1]$,
with $\mathrm{ARI}=1$ indicating perfect recovery.
   {False merge rate:} the fraction of pairs of centers belonging to \emph{different}
true clusters that are incorrectly grouped together,
  {False split rate:} the fraction of pairs of centers belonging to the \emph{same}
true cluster that are incorrectly separated,

 \paragraph*{Results : Overall patterns.}
Across both configurations, the three plots convey a consistent message:
\begin{itemize}
\item \textbf{Sample size effect.}
As $n$   increases, ARI increases monotonically in almost all settings (Figures~\ref{fig:ariL4K20} and \ref{fig:ariL6K40}), and the \emph{false split} rate decreases sharply (Figures~\ref{fig:false-splitL4K20} and \ref{fig:false-splitL6K40}). This indicates that the dominant error at small $n$ is \emph{over-splitting} (too conservative merging), which fades as local estimation noise decreases.

\item \textbf{Separation effect.}
Increasing the cluster separation $\delta$ from $0.5$ to $1$ uniformly improves performance: ARI rises and false split drops (compare the $(\delta,u_n)$ curves in Figures~\ref{fig:ariL4K20}, \ref{fig:ariL6K40} and \ref{fig:false-splitL4K20} and \ref{fig:false-splitL6K40}). In practice, $\delta=1$ behaves as an ``easy'' regime where the method can recover the partition well for moderate to large $n$.

\item \textbf{Threshold tuning $u_n$ induces a regime-dependent merge--split trade-off.}

While the \emph{optimal} $u_n$ can depend on the configuration $(K,L)$ (and on the difficulty level $\delta$), the qualitative mechanism is stable within each regime: increasing $u_n$ tends to reduce false splits but can inflate false merges when separation is weak.
In the two regimes considered here, this appears as follows:
\begin{enumerate}
\item \emph{$u_n=1$ (conservative).}
Fusions are rarely accepted: false merge is essentially $0$, but false split remains very large (often around $0.75$--$0.88$), so ARI stays low even as $n$ increases.
This is visible in both regimes: the $(\delta,u_n)=(\cdot,1)$ curves stay low in Figures~\ref{fig:ariL4K20}--\ref{fig:ariL6K40}, and the corresponding false-split curves stay high in Figures~\ref{fig:false-splitL4K20}--\ref{fig:false-splitL6K40}.

\item \emph{$u_n=4$ (aggressive).}
False split decreases because more fusions are accepted, but when separation is small (notably $\delta=0.5$) false merge can become non-negligible and may dominate the error, which pulls ARI down.
This pattern is most apparent in Figures~\ref{fig:false-mergeL4K20}--\ref{fig:false-mergeL6K40}. For example, for $(K,L)=(20,4)$ at $n=500$ and $\delta=0.5$, false merge is large (about $0.33$ when $R=50$), coinciding with poor ARI.

\item \emph{$u_n=2$ (intermediate).}
Across the settings explored, $u_n=2$ provides the best compromise: false merge stays close to $0$ while false split is substantially reduced, yielding the highest ARI overall.
In the easier separation regime $\delta=1$, performance is already strong at moderate $n$ and improves further with $n$ (Figures~\ref{fig:ariL4K20}--\ref{fig:ariL6K40}), while the false-merge and false-split levels remain small (Figures~\ref{fig:false-mergeL4K20}--\ref{fig:false-splitL4K20}--\ref{fig:false-mergeL6K40}--\ref{fig:false-splitL6K40}).
\end{enumerate}

\item \textbf{Effect of the number of bootstrap rounds.}
Comparing line styles (solid $R=50$ vs dashed $R=100$) in Figures~\ref{fig:ariL4K20}--\ref{fig:ariL6K40}--\ref{fig:false-splitL4K20}--\ref{fig:false-splitL6K40}, increasing $R$ from $50$ to $100$ yields \emph{modest but systematic} gains in many regimes: ARI increases slightly and false split decreases, most noticeably in the difficult/finite-sample settings (small $n$, smaller $\delta$, and/or $u_n\in\{1,2\}$). The false merge rate can increase slightly in some cases (because additional rounds allow more opportunities for fusions), but the overall changes remain small compared to the dominant effects of $(n,\delta,u_n)$.
\end{itemize}

\paragraph*{Configuration comparison: $(K,L)=(20,4)$ vs $(40,6)$ (Figures~\ref{fig:L4K20}--\ref{fig:L6K40}).}
The second configuration is uniformly harder, and the figures show this primarily through \textbf{slower ARI improvement} and \textbf{higher residual splitting} at comparable $(n,\delta,u_n)$:
\begin{itemize}
\item \textbf{Lower ARI at fixed $(n,\delta,u_n)$.}
For many cells, ARI under $(40,6)$ is below the corresponding $(20,4)$ value (e.g., at $\delta=1$, $u_n=4$, $n=2000$: ARI $\approx 0.962$--$0.951$ for $(40,6)$ vs $\approx 0.991$--$0.994$ for $(20,4)$). This is the expected scaling effect: with more centres and more true groups, there are more boundaries to recover and more chances to commit either a split or a merge error.

\item \textbf{Same qualitative trade-off in $u_n$.}
Despite the increased difficulty, the ranking of $u_n$ is unchanged: $u_n=1$ remains over-conservative (high false split), $u_n=4$ can be too aggressive when $\delta$ is small (visible false merge), and $u_n=2$ is typically closest to optimal in ARI.  

\item \textbf{Near-perfect recovery in the easy regime.}
When $\delta=1$ and $n$ is large, both configurations achieve essentially perfect clustering for $u_n=4$ (ARI $\approx 1$ and false merge/split $\approx 0$ at $n=5000$), and very high performance for $u_n=2$. This indicates that the method scales well when separation is sufficient and local sample size is large enough.

\item \textbf{Where the configurations differ the most.}
The largest gaps appear in the intermediate regimes (e.g., $\delta=1$ with $n\in\{500,1000,2000\}$): $(40,6)$ tends to retain more splitting (higher false split) and thus slightly lower ARI. In contrast, for very large $n$ the gap shrinks, consistent with the idea that the remaining errors are driven mainly by finite-sample variability rather than a structural limitation.
\end{itemize}
 
\begin{figure}[t]
  \centering

  \begin{subfigure}[t]{0.32\textwidth}
    \centering
    \includegraphics[width=\linewidth]{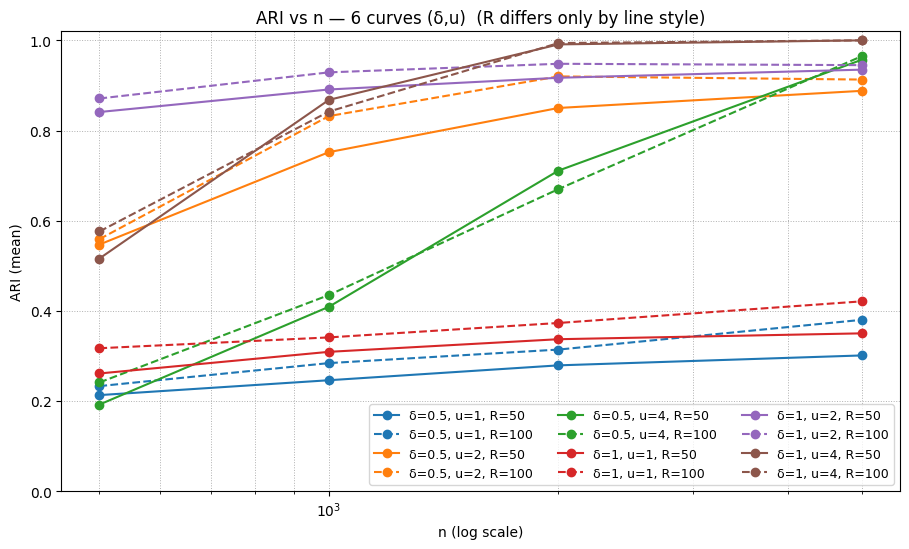}
    \caption{ARI}
    \label{fig:ariL4K20}
  \end{subfigure}\hfill
  \begin{subfigure}[t]{0.32\textwidth}
    \centering
    \includegraphics[width=\linewidth]{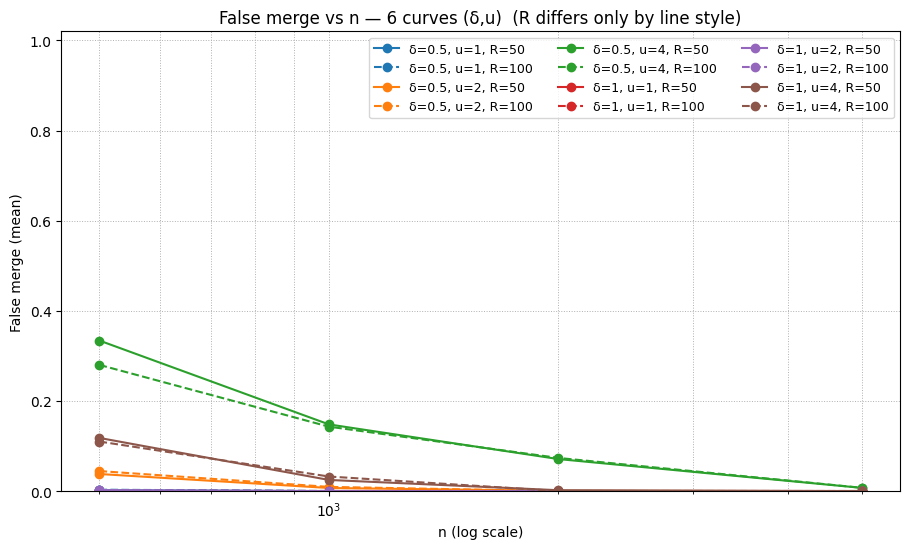}
    \caption{False merge}
    \label{fig:false-mergeL4K20}
  \end{subfigure}\hfill
  \begin{subfigure}[t]{0.32\textwidth}
    \centering
    \includegraphics[width=\linewidth]{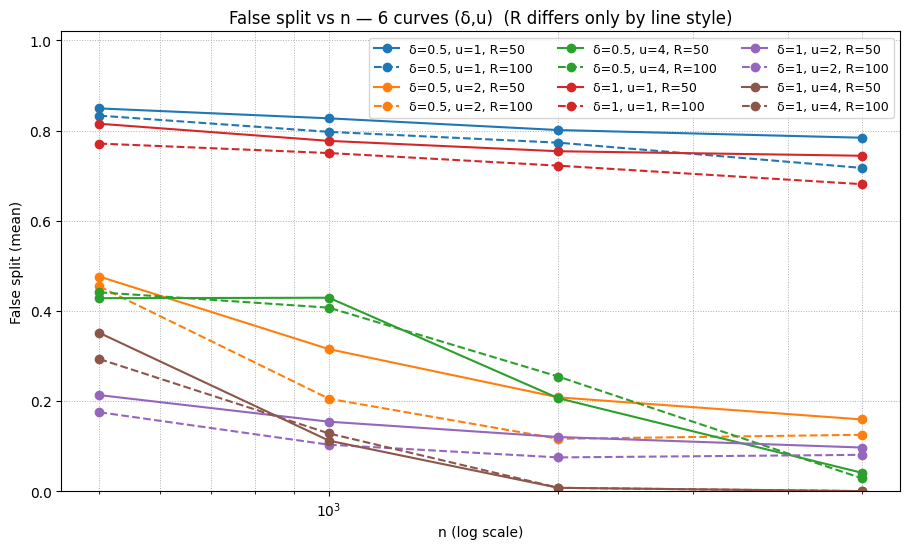}
    \caption{False split}
    \label{fig:false-splitL4K20}
  \end{subfigure}

  \caption{Performance curves versus $n$ (solid: $R=50$, dashed: $R=100$), for each $(\delta,u_n)$ setting for $L= 4$ and $K= 20$.}
  \label{fig:L4K20}
\end{figure}

\begin{figure}[t]
  \centering

  \begin{subfigure}[t]{0.32\textwidth}
    \centering
    \includegraphics[width=\linewidth]{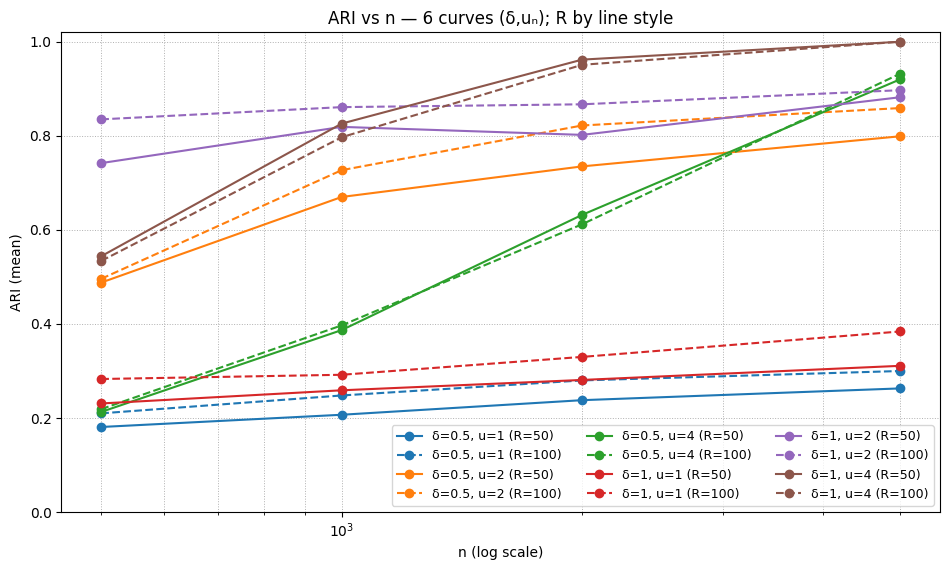}
    \caption{ARI}
    \label{fig:ariL6K40}
  \end{subfigure}\hfill
  \begin{subfigure}[t]{0.32\textwidth}
    \centering
    \includegraphics[width=\linewidth]{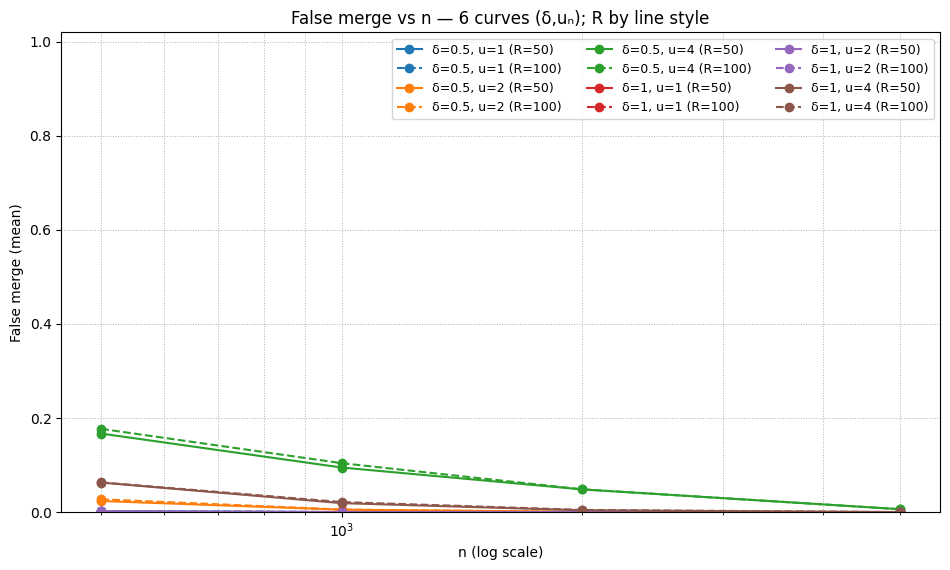}
    \caption{False merge}
    \label{fig:false-mergeL6K40}
  \end{subfigure}\hfill
  \begin{subfigure}[t]{0.32\textwidth}
    \centering
    \includegraphics[width=\linewidth]{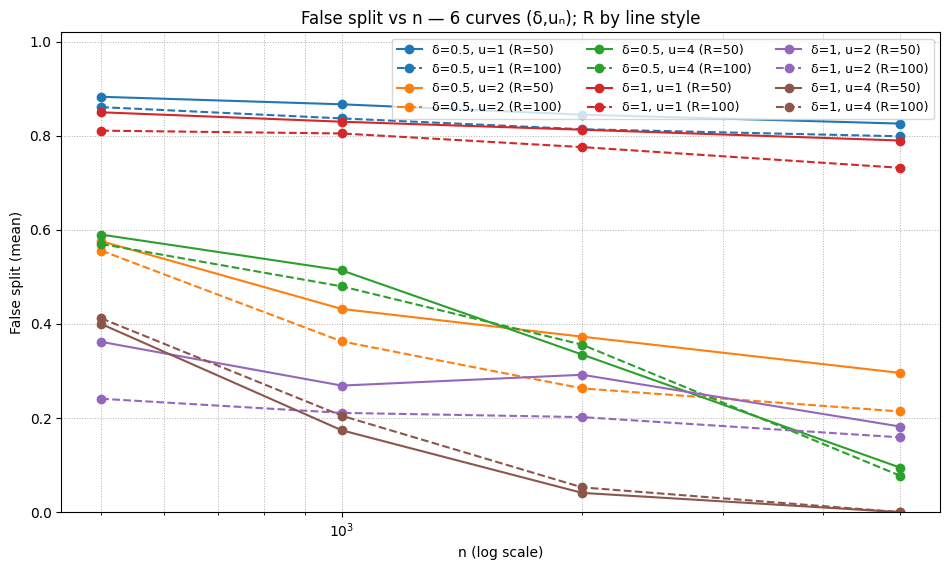}
    \caption{False split}
    \label{fig:false-splitL6K40}
  \end{subfigure}

  \caption{Performance curves versus $n$ (solid: $R=50$, dashed: $R=100$), for each $(\delta,u_n)$ setting for $L= 6$ and $K= 40$.}
  \label{fig:L6K40}
\end{figure}

\section{Real Data Application}
We illustrate the proposed methods with a real, large–scale study of commercial flights. The analysis is designed to mirror a federated setting in which only centre–level summaries are exchanged while raw data remain local.

\subsection{Arrival Delays Across Major U.S.\ Destination Airports (2007)}
\paragraph*{Data source and construction.}
We use the public U.S.\ airline on–time performance data for calendar year 2007 (ASA Data Expo 2009).\footnote{Download: Harvard Dataverse, ``Airline On-Time Performance'' (ASA Data Expo 2009), persistent identifier \texttt{doi:10.7910/DVN/HG7NV7}. A convenient landing page is \url{https://dataverse.harvard.edu/dataset.xhtml?persistentId=doi:10.7910/DVN/HG7NV7}.}
Starting from all 2007 records, we (i) remove \emph{cancelled} flights and (ii) retain observations with a non–missing arrival delay. We then define \emph{centres} as \emph{destination airports} and keep only destinations with at least $100{,}000$ usable flights after the above filters. For each retained destination $k$, we draw a simple random sample of exactly $100{,}000$ flights without replacement, yielding $K=22$ centres and a total of $N=2{,}200{,}000$ rows in the working dataset. Empirical delay rates vary substantially across centres (e.g., $\approx 38\%$ at EWR/LGA versus $\approx 19\%$ at SLC).

\paragraph*{Response and covariates.}
The binary response is
\(
Y=1\{\text{ArrDelay}\ge 15\},
\)
indicating an arrival delay of at least 15 minutes. Covariates comprise:
\begin{itemize}
\item \textbf{Distance} (miles), \textbf{DayOfWeek} ($1$--$7$), \textbf{Month} ($1$--$12$);
\item \textbf{ArrHourClass}: five bins from scheduled arrival time (CRSArrTime),
\(
\{0\!-\!5,\,6\!-\!8,\,9\!-\!14,\,15\!-\!21,\,22\!-\!00\},
\)
encoded with $N\!-\!1=4$ dummies using \(\texttt{0--5}\) as the reference.
\end{itemize}
We intentionally \emph{exclude} route/operator identifiers such as \textbf{Origin} and \textbf{UniqueCarrier} to avoid high–cardinality features with uneven support across centres and to ensure well–conditioned centre–specific information matrices. All indicator variables are coded as $\{0,1\}$, and the destination identifiers (\texttt{Centre}=\texttt{Dest}) are not used as predictors. In total, the design matrix has
\(
p \;=\; 23\;(\text{Distance, DayOfWeek, Month, ArrHour})\;=\;23
\)
covariates (excluding the intercept).

\paragraph*{Clustering results.}
The CoC-algorithm   produced the same partition for all tuning choices considered (namely $u_n \in \{1,2,4\}$ and $R \in \{50,100\}$): each airport was assigned to a singleton cluster. In other words, under the fitted local logistic-delay models, the procedure finds no statistically supported fusions between airport ``centres,'' suggesting that each airport exhibits a distinct delay profile with respect to the covariates included in the model.

This conclusion should nonetheless be interpreted with caution. Treating airports as fully independent units is a strong assumption that is unlikely to hold in an air-traffic system: common shocks and network effects (e.g., large-scale weather events, airline routing and schedule propagation, and air-traffic-control constraints) can induce cross-airport dependence that is not captured by purely local fits. Consequently, the singleton partition should be understood as evidence of separability \emph{within the imposed modelling and independence framework}, rather than as definitive proof that airports are intrinsically unrelated in their delay dynamics.

\section{Concluding Remarks}

This work develops a test-driven \emph{Clusters-of-Centres} (CoC) methodology for distributed inference when only centre-level summaries are available. The approach combines (i) multivariate Cochran-type global and two-block integration tests computed from summary statistics, and (ii) a multi-round bootstrap CoC algorithm that repeatedly re-evaluates candidate fusions across independently resampled summary sets. Under standard regularity conditions and a separation assumption between the true blocks, we establish a \emph{true-partition recovery} property: for a number of rounds \(R(n)\) growing with \(n\), the estimated partition matches the true partition with probability tending to one. Another proposed clustering procedure, based on a suitably chosen shrinkage rejection region, also recovers the true partition with probability tending to one.

Several questions remain open. First, from a practical perspective, our simulations indicate that the threshold \(t_n=u_n\sqrt{\log n}\) is the main tuning lever, as it governs the merge--split balance through the fusion rule; thus, an adaptive selection procedure deserves investigation. Second, because clustering is data-driven, establishing valid post-selection inference for the aggregated model fitted within each detected block is an important next step. Finally, it would be useful to extend the recovery guarantees and the error control of the procedure to regimes where the number of centres \(K\) increases with \(n\), which would require a careful analysis of multiplicity and separation conditions in high-combinatorial clustering sequences.

\bibliography{bibliography.bib}
\newpage

\section{Supplementary material}
\subsection{Proofs of the main results}

\subsubsection{Proof of Lemma~\ref{th::testHomogeneity}}
Define the aggregated estimator
\begin{displaymath}
\check{\theta}_n \, =\,  \Big( \sum_{k=1}^K V_k \Big)^{-1} \sum_{k=1}^K V_k \widehat{\theta}_{n, k}
\end{displaymath}
and define its feasible analogue
\begin{displaymath}
\widehat{\theta}_n \, =\,  \Big( \sum_{k=1}^K \widehat{V}_{n, k} \Big)^{-1} \sum_{k=1}^K \widehat{V}_{n, k} \widehat{\theta}_{n, k}.
\end{displaymath}
 
Since each $\widehat{V}_{n, k}$ converges in probability to its respective positive definite counterpart~$V_k$, 
the sum $\widehat{V}_n:=\sum_{k=1}^K \widehat{V}_{n, k}$ is positive definite with a probability tending to one.
Therefore,  $\widehat{\theta}_n$ is well defined with a probability tending to one.

Let $V:=\sum_{k=1}^K V_k$. Then
\begin{eqnarray*}
\sqrt{n}\big( \check{\theta}_n \, -\,  \theta_0 \big)
& = & \sqrt{n} \Big( V^{-1} \sum_{k=1}^K V_k \widehat{\theta}_{n, k} \quad - \quad \theta_0 \Big) \\
& = & V^{-1} \sum_{k=1}^K V_k \sqrt{n} \big( \widehat{\theta}_{n, k} \, -\,  \theta_0 \big) \\
& = & V^{-1} \sum_{k=1}^K U_{n, k} \, +\,  o_P(1).
\end{eqnarray*}
Furthermore,  since $\widehat{V}_{n, k}-V_k=o_P(1)$ and $\sqrt{n}\big(\widehat{\theta}_{n, k}-\theta_0\big)=O_P(1)$,  we obtain  
\begin{eqnarray*}
\lefteqn{ \sqrt{n}\big( \widehat{\theta}_n \, -\,  \check{\theta}_n \big) } \\
& = & \widehat{V}_n^{-1} \sum_{k=1}^K (\widehat{V}_{n, k}-V_k) \sqrt{n} (\widehat{\theta}_{n, k}-\theta_0)
\, +\,  \big( \widehat{V}_n^{-1} \, -\,  V^{-1} \big) \sum_{k=1}^K V_k (\widehat{\theta}_{n, k}-\theta_0) \\
& = & o_P(1)
\end{eqnarray*}
and so
\begin{displaymath}
\sqrt{n}\big( \widehat{\theta}_n \, -\,  \theta_0 \big) \, =\,  V^{-1} \sum_{k=1}^K U_{n, k} \, +\,  o_P(1).
\end{displaymath}
\vspace*{0.5cm}
 
It follows that
\begin{eqnarray*}
V \sqrt{n} (\widehat{\theta}_n \, -\,  \widehat{\theta}_{n, k})
& = & V\big( \sqrt{n}(\widehat{\theta}_n-\theta_0) \, -\,  \sqrt{n}(\widehat{\theta}_{n, k}-\theta_0) \big) \\
& = & \sum_{j=1}^K U_{n, j} \, -\,  V V_k^{-1} U_{n, k} \, +\,  o_P(1).
\end{eqnarray*}
Stacking these equations for $k=1, \ldots, K$ we obtain the $Kp$-dimensional vector
\begin{displaymath}
\sqrt{n} \overline{V} T_n
\, =\,  \left( \begin{array}{c} V \sqrt{n} (\widehat{\theta}_n-\widehat{\theta}_{n, 1}) \\
V \sqrt{n} (\widehat{\theta}_n-\widehat{\theta}_{n, 2}) \\ \vdots \\
V \sqrt{n} (\widehat{\theta}_n-\widehat{\theta}_{n, K}) \end{array} \right)
\, =\,  H \overline{U}_n \, +\,  o_P(1), 
\end{displaymath}
where $\overline{U}_n=(U_{n, 1}^\top \,  U_{n, 2}^\top  \cdots U_{n, K}^\top )^\top $.
 
Since $\overline{U}_n\Rightarrow N(0, \overline{Q})$,  where $\overline{Q}=\mbox{\rm diag}(Q_1, \ldots, Q_K)$, 
we obtain that
\begin{displaymath}
\sqrt{n} \overline{V} T_n \Rightarrow H \overline{Q}^{1/2} Z, 
\end{displaymath}
where $Z\stackrel{d}{=} N(0,  I_{Kp})$,  and thus
\begin{displaymath}
nT_n^\top \overline{V}^\top \overline{V}T_n \, \Rightarrow\,  Z^\top  \overline{Q}^{1/2} H^\top  H \overline{Q}^{1/2} Z
\, \stackrel{d}{=}\,  \sum_{\ell=1}^{Kp} \lambda_\ell \chi^2_\ell. 
\end{displaymath}
\hfill $\square$

 \subsubsection{Proof of the Proposition \ref{cor::testHomogemeity}}
     Since $\widehat{V}_{n, k}$ converges in probability to $V_k$  
and   $\sqrt{n}T_n$ is asymptotically normally distributed,  it follows from Lemma \ref{th::testHomogeneity}  that
\begin{displaymath}
nT_n^\top \widehat{\overline{V}}_n^\top \widehat{\overline{V}}_nT_n \, =\,  nT_n^\top \overline{V}^\top \overline{V}T_n \, +\,  o_P(1)
\, \Rightarrow\,  \sum_{\ell=1}^{Kp} \lambda_\ell \chi^2_\ell =: S.
\end{displaymath}

It remains to show that the empirical quantile 
$\hat{q}_{1-\alpha, n}$ converges to
    $q_{1-\alpha}$.
\\ Since  the eigenvalues of a symmetric matrix depend continuously on its entries,  and  
\begin{displaymath}
\widehat{\overline{Q}}^{1/2}_n\widehat{H}_n^\top \widehat{H}_n\widehat{\overline{Q}}^{1/2}_n
\, \stackrel{p}{\longrightarrow}\,  \overline{Q}^{1/2}H^\top H\overline{Q}^{1/2}
\end{displaymath}
it follows that the eigenvalues $\widehat{\lambda}_{1, n}\geq \widehat{\lambda}_{2, n}\geq\ldots\geq \widehat{\lambda}_{Kp}$
converge in probability to the respective counterparts $\lambda_1, \ldots, \lambda_{Kp}$.
Hence, 
\begin{displaymath}
\widehat{S}_n \, :=\,  \sum_{l=1}^{Kp} \widehat{\lambda}_{l, n} \chi^2_\ell
\, \Rightarrow\,  S.
\end{displaymath}
Since $H\overline{Q}^{1/2}\neq 0$ we have that $\lambda_1>0$ which implies that~$S$ has a continuous distribution.
Hence
\begin{equation}
\label{eq1}
\sup_x \big| \P\big( \widehat{S}_n \leq x \big) \, -\,  \P\big( S\leq x \big) \big| \, \stackrel{p}{\longrightarrow}\,  0
\end{equation}
see for example \citep[lemma 2.11]{van2000asymptotic}.
Let $q_{1-\alpha}$ and $\widehat{q}_{1-\alpha, n}$ be the respective $(1-\alpha)$-quantile of $S$ and $\widehat{S}_n$.
Since~$S$ has a strictly positive and continuous density on $(0, \infty)$ we obtain from (\ref{eq1}) that
\begin{displaymath}
\widehat{q}_{1-\alpha.n} \, \stackrel{p}{\longrightarrow}\,  q_{1-\alpha}.
\end{displaymath}
Let $Y_n:=nT_n^\top \widehat{\overline{V}}_n^\top  \widehat{\overline{V}}_nT_n - \widehat{q}_{1-\alpha, n} + q_{1-\alpha}$.
Then $Y_n\Rightarrow S$,  and we obtain that
\begin{eqnarray*}
\lefteqn{ \P\big( nT_n^\top \widehat{\overline{V}}_n^\top  \widehat{\overline{V}}_nT_n > \widehat{q}_{1-\alpha, n} \big) } \\
& = & \P\big( Y_n > q_{1-\alpha} \big) \, \mathop{\longrightarrow}\limits_{n\to\infty}\,  \P\big( S > q_{1-\alpha} \big)
\, =\,  \alpha.
\end{eqnarray*}
\hfill $\square$

\subsubsection{Proof of Lemma~\ref{lem::local_alt_power}}
It follows from \eqref{eq::asymptoticDecomposition} and \eqref{eq::alt_local} that
\begin{displaymath}
\sqrt{n} \widehat{\theta}_{n, k} \, =\,  V_k^{-1} U_{n, k} \, +\,  \sqrt{n} \theta_0 \, +\,  \Delta_k \, +\,  \varepsilon_{n, k}
\end{displaymath}
and
\begin{eqnarray*}
\sqrt{n} \widehat{\theta}_n & = & V^{-1} \sum_{j=1}^K V_j \sqrt{n} \widehat{\theta}_{n, j} \, +\,  o_P(1) \\
& = & V^{-1} \sum_{j=1}^K U_{n, j} \, +\,  \sqrt{n}\theta_0 \, +\,  V^{-1} \sum_{j=1}^K V_j \Delta_j \, +\,  o_P(1).
\end{eqnarray*}
This implies
\begin{eqnarray*}
\lefteqn{ V \sqrt{n} \big( \widehat{\theta}_n \, -\,  \widehat{\theta}_{n, k} \big) } \\
& = & \sum_{j=1}^K U_{n, j} \, -\,  V V_k^{-1} U_{n, k}
\, +\,  \sum_{j=1}^K V_j \Delta_j \, -\,  V\Delta_k \, +\,  o_P(1) 
\end{eqnarray*}
Therefore, 
\begin{displaymath}
\sqrt{n} \overline{V} T_n \, =\,  H \overline{U}_n \, +\,  \overline{W} \, +\,  o_P(1), 
\end{displaymath}
where $\overline{W}=(\overline W_1^\top , \ldots, \overline W_K^\top )^\top $.\\
It follows by the continuous mapping theorem that
\begin{displaymath}
nT_n^\top \widehat{\overline{V}}_n^\top  \widehat{\overline{V}}_nT_n = nT_n^\top \overline{V}^\top  \overline{V}T_n + o_p(1)
\, \Rightarrow\,  \big( H\overline{Q}^{1/2}Z + \overline{W} \big)^\top  \big( H\overline{Q}^{1/2}Z + \overline{W} \big), 
\end{displaymath}
where $Z\sim N(0, I_{Kp})$.

To see what the limit distribution is,  note first that 
$H\overline{Q}^{1/2}Z \stackrel{d}{=} \big(H\overline{Q}H^\top \big)^{1/2}Z\overset{d}{=}  N(0, H\overline{Q}H^\top )$.
The matrix $H\overline{Q}H^\top $ has the same set $\{\lambda_1, \ldots, \lambda_{Kp}\}$ of eigenvalues as $\overline{Q}^{1/2}H^\top H\overline{Q}^{1/2}$
and a spectral decomposition yields that
\begin{displaymath}
H\overline{Q}H^\top  \, =\,  O {\rm diag}(\lambda_1, \ldots, \lambda_{Kp}) O^\top 
\end{displaymath}
and,  hence, 
\begin{displaymath}
\big( H\overline{Q}H^\top  \big)^{1/2} \, =\,  O {\rm diag}(\sqrt{\lambda_1}, \ldots, \sqrt{\lambda_{Kp}}) O^\top , 
\end{displaymath}
where $O^\top O=OO^\top =I_{Kp}$.
We obtain that
\begin{eqnarray*}
\lefteqn{ \big( H\overline{Q}^{1/2}Z + \overline{W} \big)^\top  \big( H\overline{Q}^{1/2}Z + \overline{W} \big) } \\
& \stackrel{d}{=} & \big( O {\rm diag}(\sqrt{\lambda_1}, \ldots, \sqrt{\lambda_{Kp}}) O^\top Z \, +\,  OO^\top  \overline{W} \big)^\top 
\big( O {\rm diag}(\sqrt{\lambda_1}, \ldots, \sqrt{\lambda_{Kp}}) O^\top Z \, +\,  OO^\top  \overline{W} \big) \\
& = & \big( {\rm diag}(\sqrt{\lambda_1}, \ldots, \sqrt{\lambda_{Kp}}) O^\top Z \, +\,  O^\top  \overline{W} \big)^\top 
\big( {\rm diag}(\sqrt{\lambda_1}, \ldots, \sqrt{\lambda_{Kp}}) O^\top Z \, +\,  O^\top  \overline{W} \big) \\
& = & \sum_{j=1}^{Kp} \big( \sqrt{\lambda_j} Z_j^* \, +\,  \delta_j \big)^2 \\
& \stackrel{d}{=} & \sum_{j\colon\,  \lambda_j>0} \lambda_j \chi_j^2(\delta_j^2/\lambda_j) \, +\,  \sum_{j\colon\,  \lambda_j=0} \delta_j^2, 
\end{eqnarray*}
where $(\delta_1, \ldots, \delta_{Kp})^\top :=O^\top \overline{W}$ and $(Z_1^*, \ldots, Z_{Kp}^*)^\top :=O^\top Z\overset{d}{=}  N(0, I_{Kp})$.
For $j$ with $\lambda_j>0$,  $\chi_j^2(\delta_j^2/\lambda_j)$ are independent non-central $\chi_1^2$ with non-centrality parameter $\delta^2/\lambda$.
\hfill $\square$

\subsubsection{Proof of the lemma \ref{lem::CoC-algorithmOracle}}
   Let $\mathcal E_{1, n}$ be the event "at the sample sizes $n, $ the global homogeneity equality $\theta_{0, 1} = \cdots = \theta_{0, K}$ is rejected." Since $L>1,  \lim_{n \rightarrow \infty}\P(\mathcal E_{1, n}^c) = 0.$ 

    We now work on  $\mathcal E_{1, n}$.
    Let the later of \(\{i, j\}\),  say $j$ without  loss of generality,  be inserted when there are \(d \ge 1\) pure \(\mathcal P_q\)-clusters; label the one containing \(i\) as “1” without loss of generality. Let $p_{r, n}$ be the p-value of the $r-th$ test of the fusion   between $i$ and the $r-$th \(\mathcal P_q\)-clusters,  $r \leq d.$
Let
\(
A_n:=\{p_{1, n}\ge \alpha\}, \quad 
N:=\sum_{r=1}^d  1\{p_{r, n}\ge \alpha\}\in\{0, 1, \dots, d\}.
\)
We have, 
\(
\P(\text{pick cluster 1}\mid d)
=\P(A_n\mid R)\, \P(\text{pick cluster "1"}\mid A_n,  d) \leq \P(A_n\mid d), 
\)
since \(\P(\text{pick cluster "1"}\mid A_n^{c},  d)=0\).
Also,  $\lim_{n \rightarrow \infty}{\P(A_n\mid d)} =  1-\alpha, $ since the asymptotic level of the test is $\alpha.$
  Then,  for sufficiently large $n, $
\[
 \P\big(\widehat\ell_{n, i}=\widehat\ell_{n, j}\big)\ \le\ 1-\alpha.
\]
The second part follows since the power of the test of the proposition \ref{prop::integration} under the fixed non-local alternative tends to $1.$
\hfill $\square$

\begin{lem} 
\label{lem::growth}
Under Assumptions \textbf{A1}--\textbf{A6}, fix two disjoint clusters $S_i$ and $S_j$
($S_i\cap S_j=\emptyset$). Consider the test $J(i,j)$ that rejects for large values of
$Q_{n,r}$ using a deterministic cutoff $q>0$ (independent of $i$ and $j$).

There exists a constant $\bar\alpha=\bar\alpha(q)\in(0,1)$ such that under $H_0$,
\[
\limsup_{n\to\infty}\ \alpha^*_{J(i,j),n}(1)\ \le\ \bar\alpha
\qquad\text{in probability},
\]
and under any fixed alternative $\Delta_{i,j}\neq 0$,
\[
\beta^*_{J(i,j),n}(1)\xrightarrow{p}0 .
\]
\end{lem}

\subsubsection*{Proof of the lemma \ref{lem::growth}}
Work at a fixed bootstrap round $r$.
Set
\(
\widehat V^{(i)}_n:=\sum_{k\in S_i}\widehat V_{n, k}, \qquad
\widehat V^{(j)}_n:=\sum_{k\in S_j}\widehat V_{n, k}, \qquad
\widehat V^{(i, j)}_n:=\widehat V^{(i)}_n+\widehat V^{(j)}_n.
\)
Define the AEE  estimators by
\(
\widehat V^{(i)}_n\, \widehat{\widehat\theta}^{(i, r)}_n=\sum_{k\in S_i}\widehat V_{n, k}\, \widehat\theta^{(r)}_{n, k}, \qquad
\widehat V^{(j)}_n\, \widehat{\widehat\theta}^{(j, r)}_n=\sum_{k\in S_j}\widehat V_{n, k}\, \widehat\theta^{(r)}_{n, k}, 
\)
and the non–bootstrap versions $\widehat{\widehat\theta}^{(i)}_n, \widehat{\widehat\theta}^{(j)}_n$ with $\widehat\theta_{n, k}$ in place of $\widehat\theta^{(r)}_{n, k}$. The combined AEE satisfies
\(
\widehat V^{(i, j)}_n\, \widehat{\widehat\theta}^{(i, j, r)}_n
= \widehat V^{(i)}_n \widehat{\widehat\theta}^{(i, r)}_n + \widehat  V^{(j)}_n \widehat{\widehat\theta}^{(j, r)}_n.
\)
Define the stacked test vector
$$
\widehat T^{(i, j, r)}_n
:=
\begin{pmatrix}
\widehat{\widehat\theta}^{(i, j, r)}_n-\widehat{\widehat\theta}^{(i, r)}_n\\[1mm]
\widehat{\widehat\theta}^{(i, j, r)}_n-\widehat{\widehat\theta}^{(j, r)}_n
\end{pmatrix}\in\mathbb R^{2p}, 
$$
the cluster half–difference
$
\widehat\Delta^{\mathrm{cl}}_{i, j}:=\tfrac12\big(\widehat{\widehat\theta}^{(i)}_n-\widehat{\widehat\theta}^{(j)}_n\big), 
$
and the $2p$–vector of centrewise bootstrap differences (with $\sqrt n$ inside)
$$
B_n^{(i, j, r)}:=
\begin{pmatrix}
\displaystyle\sum_{k\in S_i}\sqrt n\, \widehat V_{n, k}\big(\widehat\theta^{(r)}_{n, k}-\widehat\theta_{n, k}\big)\\[2mm]
\displaystyle\sum_{k\in S_j}\sqrt n\, \widehat V_{n, k}\big(\widehat\theta^{(r)}_{n, k}-\widehat\theta_{n, k}\big)
\end{pmatrix}.
\text{ Also set the shift }
s^{(\mathrm{cl})}_{i, j}:=
\begin{pmatrix}
-2\sqrt n\, \widehat V^{(j)}_n\, \widehat\Delta^{\mathrm{cl}}_{i, j}\\[1mm]
\phantom{-}2\sqrt n\, \widehat V^{(i)}_n\, \widehat\Delta^{\mathrm{cl}}_{i, j}
\end{pmatrix}.
$$
A short add–subtract argument using
\[
\widehat V^{(i, j)}_n\big(\widehat{\widehat\theta}^{(i, j, r)}_n-\widehat{\widehat\theta}^{(i, r)}_n\big)
=\widehat V^{(j)}_n\big(\widehat{\widehat\theta}^{(j, r)}_n-\widehat{\widehat\theta}^{(i, r)}_n\big), 
\quad
\widehat V^{(i, j)}_n\big(\widehat{\widehat\theta}^{(i, j, r)}_n-\widehat{\widehat\theta}^{(j, r)}_n\big)
=\widehat V^{(i)}_n\big(\widehat{\widehat\theta}^{(i, r)}_n-\widehat{\widehat\theta}^{(j, r)}_n\big), 
\]
and
\(
\sqrt n\, \widehat V^{(i)}_n\big(\widehat{\widehat\theta}^{(i, r)}_n-\widehat{\widehat\theta}^{(i)}_n\big)
=\sum_{k\in S_i}\sqrt n\, \widehat V_{n, k}\big(\widehat\theta^{(r)}_{n, k}-\widehat\theta_{n, k}\big), 
\)
together with $I_p-\widehat V^{(i, j)}_n(\widehat V^{(i)}_n)^{-1}=-\widehat V^{(j)}_n(\widehat V^{(i)}_n)^{-1}$ and its counterpart for $j$,  yields the   stacked identity
\begin{equation}\label{eq:stacked-identity}
\sqrt n\, \widehat V^{(i, j)}_n\, \widehat T^{(i, j, r)}_n
=
\begin{pmatrix}
-\widehat V^{(j)}_n(\widehat V^{(i)}_n)^{-1} & I_p\\[1mm]
I_p & -\widehat V^{(i)}_n(\widehat V^{(j)}_n)^{-1}
\end{pmatrix}
B_n^{(i, j, r)}
\;+\; s^{(\mathrm{cl})}_{i, j}.
\end{equation}
For  fixed bootstrap round $r$, the quadratic form statistic satisfies
\(
Q_{n,r}
=\big\|\,M_n\,B_n^{(i,j,r)}+s^{(\mathrm{cl})}_{i,j}\,\big\|^2,
\quad
s^{(\mathrm{cl})}_{i,j}=2\sqrt n\,\widehat c_n^{(i,j)},
\)
where $M_n$ and $s^{(\mathrm{cl})}_{i,j}$ are $\mathcal  Z_n$--measurable
(hence fixed conditional on the data), while $B_n^{(i,j,r)}$ carries the bootstrap randomness.

\paragraph*{Size under $H_0$ is bounded away from $1$.}
Under  $H_0$,  $\sqrt n\,\widehat\Delta^{\mathrm{cl}}_{i,j}=O_p(1)$.
Under \textbf{A5},
it follows that
\(
\|s^{(\mathrm{cl})}_{i,j}\| = O_p(1).
\)
Fix a deterministic cutoff $q>0$ and write $\alpha^*_{J(i,j),n}(1):=\P(Q_{n,r}>q\mid \mathcal Z_n)$, where
$Q_{n,r}=\|M_nB_n^{(i,j,r)}+s^{(\mathrm{cl})}_{i,j}\|^2$.
Let $\eta\in(0,1)$ be arbitrary. By tightness of $s^{(\mathrm{cl})}_{i,j}$, there exists $R=R(\eta)$ such that
$\P(\|s^{(\mathrm{cl})}_{i,j}\|>R)\le \eta$ for all large $n$.
Since $s^{(\mathrm{cl})}_{i,j}$ is $\mathcal Z_n$--measurable, we have the decomposition
\[
\alpha^*_{J(i,j),n}(1)
\le \mathbf 1\{\|s^{(\mathrm{cl})}_{i,j}\|\le R\}\,\P(Q_{n,r}>q\mid \mathcal Z_n)
+\mathbf 1\{\|s^{(\mathrm{cl})}_{i,j}\|>R\}.
\]
On $\{\|s^{(\mathrm{cl})}_{i,j}\|\le R\}$, set $u:=s^{(\mathrm{cl})}_{i,j}$ so that $\|u\|\le R$.
By \textbf{A5--A6}, conditionally on $\mathcal Z_n$, we have $B_n^{(i,j,r)}\Rightarrow \mathcal N(0,\overline Q_{i,j})$ and
$M_n\xrightarrow{p}\overline V_{i,j}$, so by conditional Slutsky (in probability),
$M_nB_n^{(i,j,r)}\Rightarrow G_{i,j}:=\overline V_{i,j}Z$ with $Z\sim\mathcal N(0,\overline Q_{i,j})$.
The limit is non-degenerate, since $\Sigma_{i,j}:=\overline V_{i,j}\overline Q_{i,j}\overline V_{i,j}^\top$ is positive definite.
Then $G_{i,j}$ has a strictly positive continuous density on $\R^{2p}$ and the map
$u\mapsto \P(\|G_{i,j}+u\|\le \sqrt q)$ is continuous. Since $\{u:\|u\|\le R\}$ is compact, we may define
$\varepsilon_{R,q}:=\inf_{\|u\|\le R}\P(\|G_{i,j}+u\|\le \sqrt q)>0$.
Therefore, on $\{\|s^{(\mathrm{cl})}_{i,j}\|\le R\}$,
\[
\P(Q_{n,r}>q\mid \mathcal Z_n)
=\P(\|M_nB_n^{(i,j,r)}+u\|>\sqrt q\mid \mathcal Z_n)
\le 1-\varepsilon_{R,q}+o_p(1),
\]
and inserting this bound yields
\(
\alpha^*_{J(i,j),n}(1)
\le \mathbf 1\{\|s^{(\mathrm{cl})}_{i,j}\|\le R\}\,(1-\varepsilon_{R,q}+o_p(1))
+\mathbf 1\{\|s^{(\mathrm{cl})}_{i,j}\|>R\}.
\)
In particular, for any $\delta\in(0,\varepsilon_{R,q})$,
\[
\P\!\left(\alpha^*_{J(i,j),n}(1)>1-\varepsilon_{R,q}+\delta\right)
\le \P(\|s^{(\mathrm{cl})}_{i,j}\|>R)+o(1)\le \eta+o(1),
\]
so $\alpha^*_{J(i,j),n}(1)$ is bounded away from $1$ with probability at least $1-\eta$ for all large $n$.
\paragraph*{Power under fixed alternatives.}
Assume $\Delta_{i,j}\neq 0$ (fixed non-local alternative). Then
$\widehat\Delta^{\mathrm{cl}}_{i,j}\xrightarrow{p}\Delta^{\mathrm{cl}}_{i,j}\neq 0$ under \textbf{A5},
so
\(
\|s^{(\mathrm{cl})}_{i,j}\| = 2\sqrt n\,\|\widehat c_n^{(i,j)}\| \xrightarrow{p}\infty.
\)
Meanwhile, by \textbf{A5--A6}, conditionally on the data,
$M_nB_n^{(i,j,r)}=O_p(1)$ (in probability). Therefore,
\(
Q_{n,r}=\|M_nB_n^{(i,j,r)}+s^{(\mathrm{cl})}_{i,j}\|^2 \xrightarrow{p}\infty,
\)
hence for any fixed $q>0$,
\[
\beta^*_{J(i,j),n}(1)
=\mathbb P(Q_{n,r}\le q\mid \mathcal Z_n)\xrightarrow{p}0.
\]
This proves the claim. 
\hfill$\square$

\subsubsection{Proof of Theorem~\ref{th:CoC-algorithmOracleIterations}}
Let $C=\big\{\{1\}, \{2\}, \ldots\{K\}\big\}$ be the original partition, 
$\widehat{C}_n^{(0)}$ be the partition after the one-shot CoC-Algorithm (consisting of the
homogeneity test according to Proposition \ref{cor::testHomogemeity} and the Cochran-type test
according to Proposition \ref{prop::integration}). 
Furthermore,  denote by $\widehat{C}_n^{(r)}$ the partition obtained after $r$ bootstrap loops. 
In what follows we show that
\begin{itemize}
\item[(i)\quad] The probability that any inhomogeneous pair $(k, i)$ is merged tends to zero as $n\to\infty$.
\item[(ii)\quad] The probability that $\widehat{C}_n^{(R(n))}=\mathcal{P}$ tends to one as $n\to\infty$.
\end{itemize}

{\sl Proof of (i).}\\
First,  since the tests in Propositions \ref{cor::testHomogemeity} and \ref{prop::integration} are consistent,
and by the same arguments used in the proof of Lemma \ref{lem::growth},  we conclude that in the first
step of the multi-round CoC-algorithm, the probability that any inhomogeneous pair \((k, i)\) is
merged tends to zero as \(n\to\infty\).

Let $(k, i)$ be an inhomogeneous pair and let
\begin{displaymath}
J_{k, i} \, =\,  \big\{ \{S_1, S_2\}\colon \; k\in S_1, \,  i\in S_2, \,  S_1, S_2\subseteq\{1, 2, \ldots, K\}, \,  S_1\cap S_2=\emptyset \big\}
\end{displaymath}
be the collection of pairs of clusters whose union leads to the event that $k$ and $i$ fall in the same cluster,  each set being a cluster of homogeneous centres.
We have
 \begin{eqnarray*}
\lefteqn{ \P\big( \mbox{ $k$ and $i$ are combined } \big) } \\
& \leq & \E \sum_{\{S_1, S_2\}\in J_{k, i}} \P\big( \mbox{ $S_1$ and $S_2$ are combined } \big) \\
& \leq & \E \left[ \underbrace{\beta^*_{J_{k, i}, n}(0)}_{\text{1st global test}}
+   \underbrace{|J_{k, i}|\beta^*_{J(k, i), n}(0)}_{\text{integration tests inside one-shot CoC}}
\, +\,   \underbrace{\sum_{r=1}^{R(n)}   |J_{k, i}|\beta^*_{J(i, j), n}(r)}_{\text{successive rounds of CoC.}}\right] \\
& = & \E \beta^*_{J_{k, i}, n}(0) +  \E |J_{k, i}|(R(n) + 1)\beta^*_{J(k, i), n}(1) \\ 
& \le  &  C    \E   R(n) \max_{S_i,  S_j} \max_{i, j} \beta^*_{J(i, j), n}(1)
\, \mathop{\longrightarrow}_{n\to\infty}\,  0
\end{eqnarray*}
where $C$ is a positive constant and the limit is obtained by the Lebesgue dominated convergence theorem
under assumption \textbf{(A7)}.  Since the number of inhomogeneous pairs is finite we obtain that
\begin{displaymath}
\P\big( \mbox{ at least one inhomogeneous pair is combined } \big) \, \mathop{\longrightarrow}_{n\to\infty}\,  0.
\end{displaymath}

{\sl Proof of (ii).}\\
By part~(i), up to an event of probability $o(1)$, no inhomogeneous pair is ever merged throughout
all $R(n)$ rounds.  It therefore suffices to show that every homogeneous pair $(k,j)$ belonging to
the same true cluster $\mathcal{P}_q$ is eventually merged with probability tending to one.

Fix such a homogeneous pair $(k,j)$.  Let $J$ denote any integration test whose two sides contain
$k$ and $j$ respectively (at least one such test is evaluated at each round in which $k$ and $j$
belong to different clusters within the algorithm).  By Lemma~\ref{lem::growth}, there exists
$\bar\alpha=\bar\alpha(q)\in(0,1)$ such that
\[
\limsup_{n\to\infty}\,\alpha^*_{J,n}(1)\;\le\;\bar\alpha \qquad\text{in probability.}
\]
Conditionally on the original data $\mathcal{Z}_n$, the $R(n)$ bootstrap resamples are independent,
so the events $\{$test~$J$ rejects at round~$r\}$ are conditionally independent across $r$.
Therefore, the conditional probability that the pair $(k,j)$ fails to merge at every one of the
$R(n)$ rounds satisfies
\[
\P\!\big(\text{$(k,j)$ never merged}\;\big|\;\mathcal{Z}_n\big)
\;\le\;
\bigl(\alpha^*_{J,n}(1)\bigr)^{R(n)}.
\]
On the event $\{\alpha^*_{J,n}(1)\le\bar\alpha\}$, whose probability tends to $1$ by
Lemma~\ref{lem::growth}, this is at most $\bar\alpha^{R(n)}$.  Since $R(n)\to\infty$ and
$\bar\alpha<1$, we have $\bar\alpha^{R(n)}\to 0$.  Hence
\[
\P\!\big(\text{$(k,j)$ never merged}\big)
\;\le\;
\P\!\big(\alpha^*_{J,n}(1)>\bar\alpha\big) + \bar\alpha^{R(n)}
\;\longrightarrow\; 0.
\]
Taking a union bound over the finitely many homogeneous pairs $(k,j)$ that need to be merged
(at most $K(K-1)/2$) yields
\[
\P\!\big(\widehat{C}_n^{(R(n))}\neq\mathcal{P}\big)
\;\le\;
\P\!\big(\text{some inhomogeneous pair is merged}\big)
+\sum_{(k,j)\,\text{homogeneous}}\P\!\big(\text{$(k,j)$ never merged}\big)
\;\longrightarrow\; 0,
\]
which proves~(ii) and completes the proof of the theorem.
\hfill $\square$

\subsubsection{Proof of Lemma~\ref{lem::concentrationV}}
 Fix $k$. Using the triangle inequality,
\begin{align}
\big\|\widehat{V}_{n,k}-V_k\big\|_{\op}
&\le
\left\|
\frac1n\sum_{i=1}^n
\Big(
v_k(Z_{i,k},\widehat{\theta}_{n,k})
-
v_k(Z_{i,k},\theta_{0,k})
\Big)
\right\|_{\op} \nonumber \\
&\qquad
+
\left\|
\frac1n\sum_{i=1}^n v_k(Z_{i,k},\theta_{0,k})
-
\E\!\left[v_k(Z_{1,k},\theta_{0,k})\right]
\right\|_{\op}.
\label{eq::decompV}
\end{align}

\medskip
\noindent\textbf{Step 1: Control of the plug-in term.}
By the triangle inequality and \textnormal{SC1},
\begin{align}
\left\|
\frac1n\sum_{i=1}^n
\Big(
v_k(Z_{i,k},\widehat{\theta}_{n,k})
-
v_k(Z_{i,k},\theta_{0,k})
\Big)
\right\|_{\op}
&\le
\frac1n\sum_{i=1}^n
\big\|v_k(Z_{i,k},\widehat{\theta}_{n,k})-v_k(Z_{i,k},\theta_{0,k})\big\|_{\op}
\nonumber\\
&\le
\left(\frac1n\sum_{i=1}^n w_k(Z_{i,k})\right)\,
\big\|\widehat{\theta}_{n,k}-\theta_{0,k}\big\|.
\label{eq::pluginBound}
\end{align}

Consider the event
$
\mathcal E_{\theta}(n,k)
:=
\left\{
\big\|\widehat{\theta}_{n,k}-\theta_{0,k}\big\|
\le a_k\sqrt{\frac{\log n}{n}}
\right\}.
$
By~\eqref{eq::concentration},
\begin{equation}
\P\!\big(\mathcal E_{\theta}(n,k)\big)=1-o\!\left(n^{-1/2}\right).
\label{eq::EventTheta}
\end{equation}

Next, define
$
\mathcal E_{w}(n,k)
:=
\left\{
\frac1n\sum_{i=1}^n w_k(Z_{i,k})
\le 2\,\E\!\big[w_k(Z_{1,k})\big]
\right\}.
$
Since $(Z_{i,k})_{i=1}^n$ are i.i.d., $\E[w_k(Z_{1,k})]>0$ (by \textnormal{SC2}), and $\E[w_k(Z_{1,k})^s]<\infty$ for some $s>2$ (\textnormal{SC2}),
a standard $s$-moment inequality (Marcinkiewicz--Zygmund inequality) for sums of i.i.d.\ centred variables implies
\[
\E\left|
\sum_{i=1}^n\Big(w_k(Z_{i,k})-\E[w_k(Z_{1,k})]\Big)
\right|^{s}
\le
C_s\,n^{s/2}\,
\E\!\left|w_k(Z_{1,k})-\E[w_k(Z_{1,k})]\right|^{s},
\]
for some constant $C_s>0$ depending only on $s$. Therefore, by Markov's inequality applied
to the centred sum at level $n\,\E[w_k(Z_{1,k})]>0$,
\begin{align}
\P\!\big(\mathcal E_{w}(n,k)^c\big)
&=
\P\!\left(
\frac1n\sum_{i=1}^n w_k(Z_{i,k}) > 2\,\E[w_k(Z_{1,k})]
\right)
\nonumber\\
&\le
\frac{
\E\left|
\sum_{i=1}^n\Big(w_k(Z_{i,k})-\E[w_k(Z_{1,k})]\Big)
\right|^{s}
}{
\big(n\,\E[w_k(Z_{1,k})]\big)^{s}
} \nonumber \\
&
\le
\frac{C_s\,\E|w_k(Z_{1,k})-\E[w_k(Z_{1,k})]|^{s}}{\E[w_k(Z_{1,k})]^{s}}\;n^{-s/2}
=
o\!\left(n^{-1/2}\right),
\label{eq::EventW}
\end{align}
because $s>2$ implies $n^{-s/2}=o(n^{-1/2})$.

Combining \eqref{eq::pluginBound}--\eqref{eq::EventW}, on the event
$\mathcal E_{\theta}(n,k)\cap \mathcal E_w(n,k)$ we have
\begin{equation}
\left\|
\frac1n\sum_{i=1}^n
\Big(
v_k(Z_{i,k},\widehat{\theta}_{n,k})
-
v_k(Z_{i,k},\theta_{0,k})
\Big)
\right\|_{\op}
\le
2\,\E\!\big[w_k(Z_{1,k})\big]\;a_k\sqrt{\frac{\log n}{n}}.
\label{eq::Term1Final}
\end{equation}

\medskip
\noindent\textbf{Step 2: Control of the empirical fluctuation at $\theta_{0,k}$.}
Write
\(
\Delta_{i,k}
:=
v_k(Z_{i,k},\theta_{0,k})-\E\!\left[v_k(Z_{1,k},\theta_{0,k})\right],
\quad i=1,\ldots,n,
\)
so that
\[
\left\|
\frac1n\sum_{i=1}^n v_k(Z_{i,k},\theta_{0,k})
-
\E\!\left[v_k(Z_{1,k},\theta_{0,k})\right]
\right\|_{\op}
=
\left\|
\frac1n\sum_{i=1}^n \Delta_{i,k}
\right\|_{\op}.
\]

Then, for some constants $C_k, D_k>0$,
\begin{equation}
\P\!\left(
\left\|
\frac1n\sum_{i=1}^n \Delta_{i,k}
\right\|_{\op}
>
C_k\sqrt{\frac{\log n}{n}}
\right)
\le
D_k\,n^{-(s-2)/2}(\log n)^{-s/2},
\label{eq::Term2Prob}
\end{equation}
which follows from a (matrix) Fuk--Nagaev type inequality applied to
$\frac1n\sum_{i=1}^n\Delta_{i,k}$.
In particular, if $s\geq 3$, then the right-hand side in~\eqref{eq::Term2Prob} is
$o(n^{-1/2})$.

Consequently, for $s\geq 3$ there exists $c_k>0$ such that
\begin{equation}
\P\!\left(
\left\|
\frac1n\sum_{i=1}^n v_k(Z_{i,k},\theta_{0,k})
-
\E\!\left[v_k(Z_{1,k},\theta_{0,k})\right]
\right\|_{\op}
\le
c_k\sqrt{\frac{\log n}{n}}
\right)
=
1-o\!\left(n^{-1/2}\right).
\label{eq::Term2Final}
\end{equation}

\medskip
\noindent\textbf{Step 3: Conclusion.}
Combining \eqref{eq::decompV}, \eqref{eq::Term1Final}, and \eqref{eq::Term2Final}, we obtain that, on the event
\[
\mathcal E_{\theta}(n,k)\cap \mathcal E_w(n,k)\cap
\left\{
\left\|
\frac1n\sum_{i=1}^n v_k(Z_{i,k},\theta_{0,k})
-
\E\!\left[v_k(Z_{1,k},\theta_{0,k})\right]
\right\|_{\op}
\le
c_k\sqrt{\frac{\log n}{n}}
\right\},
\]
\[
\big\|\widehat{V}_{n,k}-V_k\big\|_{\op}
\le
\Big(2\,\E[w_k(Z_{1,k})]\,a_k+c_k\Big)\sqrt{\frac{\log n}{n}}.
\]
Setting
\(
b_k:=2\,\E\!\big[w_k(Z_{1,k})\big]\;a_k+c_k,
\)
and using \eqref{eq::EventTheta}, \eqref{eq::EventW}, and \eqref{eq::Term2Final} yields
\(
\P\!\left(
\big\|\widehat{V}_{n,k}-V_k\big\|_{\op}
\le
b_k\sqrt{\frac{\log n}{n}}
\right)
=
1-o\!\left(n^{-1/2}\right),
\)
which is~\eqref{eq::concentrationV}.
\hfill$\square$

\subsubsection{Proof of Proposition~\ref{prop::AEE_conc_BE}}
 Subtracting $\theta_{0,V}$ and inserting $\theta_{0,k}$ yields
\begin{align}
\widehat{\theta}_n-\theta_{0,V}
&=
\overline{\widehat V}_n^{-1}\left(\frac1K\sum_{k=1}^K \widehat V_{n,k}(\widehat{\theta}_{n,k}-\theta_{0,k})\right)
+\overline{\widehat V}_n^{-1}\left(\frac1K\sum_{k=1}^K \widehat V_{n,k}(\theta_{0,k}-\theta_{0,V})\right)
\nonumber\\
&=
\overline{\widehat V}_n^{-1}\left(\frac1K\sum_{k=1}^K \widehat V_{n,k}(\widehat{\theta}_{n,k}-\theta_{0,k})\right)
+\overline{\widehat V}_n^{-1}\left(\frac1K\sum_{k=1}^K (\widehat V_{n,k}-V_k)(\theta_{0,k}-\theta_{0,V})\right),
\label{eq::decomp_theta0V}
\end{align}
where in the last step we used the identity
\[
\sum_{k=1}^K V_k(\theta_{0,k}-\theta_{0,V})
=\sum_{k=1}^K V_k\theta_{0,k}-V\theta_{0,V}
=0.
\]
Using the triangle inequality gives
\begin{align}
\|\widehat{\theta}_n-\theta_{0,V}\|
&\le
\big\|\overline{\widehat V}_n^{-1}\big\|_{\op}\;
\frac1K\sum_{k=1}^K \|\widehat V_{n,k}\|_{\op}\,\|\widehat{\theta}_{n,k}-\theta_{0,k}\|
\nonumber\\[-2pt]
&\quad+
\big\|\overline{\widehat V}_n^{-1}\big\|_{\op}\;
\frac1K\sum_{k=1}^K \|\widehat V_{n,k}-V_k\|_{\op}\,\|\theta_{0,k}-\theta_{0,V}\|.
\label{eq::det_bound_agg_theta0V}
\end{align}

\medskip
\textbf{Step 1: Control of $\|\overline{\widehat V}_n^{-1}\|_{\op}$.} We have $\overline V\succ 0.$
By the triangle inequality,
\(
\big\|\overline{\widehat V}_n-\overline V\big\|_{\op}
\le
\frac1K\sum_{k=1}^K \|\widehat V_{n,k}-V_k\|_{\op}.
\)
On the event
\(
\mathcal E_V(n):=\bigcap_{k=1}^K
\left\{\|\widehat V_{n,k}-V_k\|_{\op}\le b_k r_n\right\},
\)
it holds that
\(
\big\|\overline{\widehat V}_n-\overline V\big\|_{\op}
\le
\left(\frac1K\sum_{k=1}^K b_k\right)r_n.
\)
Choose $n$ large enough so that
$
\left(\frac1K\sum_{k=1}^K b_k\right)r_n \le \frac12\,\lambda_{\min}(\overline V).
$
Such an $n_0$ is finite for every fixed configuration of $(b_k, \lambda_{\min}(\overline V))$.
Then Weyl's inequality implies, on $\mathcal E_V(n)$,
\(
\lambda_{\min}(\overline{\widehat V}_n)
\ge
\lambda_{\min}(\overline V)-\|\overline{\widehat V}_n-\overline V\|_{\op}
\ge \frac12\,\lambda_{\min}(\overline V),
\)
hence
\begin{equation}\label{eq::inv_barV_bd_theta0V}
\big\|\overline{\widehat V}_n^{-1}\big\|_{\op}
\le
\frac{2}{\lambda_{\min}(\overline V)}.
\end{equation}
Moreover, by \eqref{eq::concentrationV} and a union bound,
\begin{equation}\label{eq::prob_EV_theta0V}
\P\big(\mathcal E_V(n)\big)=1-o\!\left(n^{-1/2}\right).
\end{equation}

\medskip
\textbf{Step 2: Control of the right-hand side of \eqref{eq::det_bound_agg_theta0V}.}
Define the event
\(
\mathcal E_\theta(n):=\bigcap_{k=1}^K
\left\{\|\widehat{\theta}_{n,k}-\theta_{0,k}\|\le a_k r_n\right\}.
\)
By \eqref{eq::concentration} and a union bound,
\begin{equation}\label{eq::prob_Etheta_theta0V}
\P\big(\mathcal E_\theta(n)\big)=1-o\!\left(n^{-1/2}\right).
\end{equation}
On $\mathcal E_V(n)$, we have $\|\widehat V_{n,k}\|_{\op}\le \|V_k\|_{\op}+b_k r_n$, so on
$\mathcal E_V(n)\cap \mathcal E_\theta(n)$,
\begin{align*}
\frac1K\sum_{k=1}^K \|\widehat V_{n,k}\|_{\op}\,\|\widehat{\theta}_{n,k}-\theta_{0,k}\|
&\le
\frac1K\sum_{k=1}^K (\|V_k\|_{\op}+b_k r_n)\,a_k r_n
=
\left(\frac1K\sum_{k=1}^K a_k\|V_k\|_{\op}\right)r_n+O(r_n^2),
\end{align*}
and also
\(
\frac1K\sum_{k=1}^K \|\widehat V_{n,k}-V_k\|_{\op}\,\|\theta_{0,k}-\theta_{0,V}\|
\le
\left(\frac1K\sum_{k=1}^K b_k\,\|\theta_{0,k}-\theta_{0,V}\|\right)r_n.
\)
Combining these bounds with \eqref{eq::det_bound_agg_theta0V} and \eqref{eq::inv_barV_bd_theta0V}
shows that there exists a constant $C>0$ such that, on $\mathcal E_V(n)\cap\mathcal E_\theta(n)$,
\(
\|\widehat{\theta}_n-\theta_{0,V}\|\le C\,r_n,
\)
which together with \eqref{eq::prob_EV_theta0V}--\eqref{eq::prob_Etheta_theta0V} yields
\eqref{eq::conc_theta_agg}.

\medskip
\textbf{Step 3: Berry--Esseen bound.}

Let $G_1,\ldots,G_K$ be independent random vectors, independent of
$(\widehat{H}_{n,1},\ldots,\widehat{H}_{n,K})$, such that
$G_k\sim \Phi$ for each $k$.
For $j=0,1,\ldots,K$, define
\(
S_j
:=\frac1{\sqrt K}\left(\sum_{k=1}^j \widehat{H}_{n,k}+\sum_{k=j+1}^K G_k\right).
\)
Then $S_K=\widehat{H}_n$ and $S_0\sim\Phi$.
Fix any   set $B\in\mathcal A$. By the triangle inequality and telescoping,
\[
\big|\P(\widehat{H}_n\in B)-\Phi(B)\big|
=
\big|\P(S_K\in B)-\P(S_0\in B)\big|
\le
\sum_{j=1}^K \big|\P(S_j\in B)-\P(S_{j-1}\in B)\big|.
\]
Now fix $j\in\{1,\ldots,K\}$ and set
\(
R_j
:=\frac1{\sqrt K}\left(\sum_{k=1}^{j-1}\widehat{H}_{n,k}+\sum_{k=j+1}^K G_k\right),
\)
so that $S_j=R_j+\frac1{\sqrt K}\widehat{H}_{n,j}$ and
$S_{j-1}=R_j+\frac1{\sqrt K}G_j$.
Using independence of $\widehat{H}_{n,j}$ and $R_j$, and of $G_j$ and $R_j$, we obtain
\begin{align*}
\big|\P(S_j\in B)-\P(S_{j-1}\in B)\big|
&=
\left|
\E\!\left[\P\!\left(\frac1{\sqrt K}\widehat{H}_{n,j}+R_j\in B \,\middle|\,R_j\right)
-\P\!\left(\frac1{\sqrt K}G_j + R_j \in B \,\middle|\,R_j\right)\right]
\right|\\
&\le
\sup_{A\in\mathcal A }
\left|
\P\!\left(\frac1{\sqrt K}\widehat{H}_{n,j}\in A\right)
-\P\!\left(\frac1{\sqrt K}G_j\in A\right)
\right|.
\end{align*}
The last term equals
$
\sup_{A\in\mathcal A }
\left|
\P\!\left(\widehat{H}_{n,j}\in A\right)
-\P\!\left(G_j\in A\right)
\right|.
$
Therefore, using the local Berry--Esseen assumption \eqref{eq::BerryEssen},
\(
\big|\P(S_j\in B)-\P(S_{j-1}\in B)\big|
\le
\sup_{A\in\mathcal A}
\big|\P(\widehat{H}_{n,j}\in A)-\Phi(A)\big|
=
O\!\left(n^{-1/2}\log n\right).
\)
Summing over $j=1,\ldots,K$ yields
\(
\sup_{B\in\mathcal A}
\big|\P(\widehat{H}_n\in B)-\Phi(B)\big|
\le
\sum_{j=1}^K O\!\left(n^{-1/2}\log n\right)
=
O\!\left(n^{-1/2}\log n\right),
\)
where the implicit constant may depend on $K$.
This proves \eqref{eq::BE_theta_agg}.
 \hfill$\square$

\subsubsection{Proof of Theorem~\ref{thm::tail_bounds_WhatV2}}

\medskip
1.  By the triangle inequality,
\begin{equation}\label{eq::tail_What_triangle_pf_V2}
\P(\|\widehat W_n\|\ge t)
\le
\P(\|W_n^\circ\|\ge t/2)
+
\P(\|\widehat W_n-W_n^\circ\|\ge t/2),
\qquad t>0.
\end{equation}
Moreover, by the reverse triangle inequality,
\(
\|\widehat W_n\|
=\|W_n^\circ+(\widehat W_n-W_n^\circ)\|
\ge \|W_n^\circ\|-\|\widehat W_n-W_n^\circ\|,
\)
so that
\(
\{\|W_n^\circ\|\ge 3t/2,\ \|\widehat W_n-W_n^\circ\|\le t/2\}
\subseteq
\{\|\widehat W_n\|\ge t\}.
\)
Hence,
\begin{equation}\label{eq::tail_What_reverse_pf_V2}
\P(\|\widehat W_n\|\ge t)
\ge
\P(\|W_n^\circ\|\ge 3t/2)
-
\P(\|\widehat W_n-W_n^\circ\|\ge t/2).
\end{equation}

 We will show  that

\smallskip
\noindent\textbf{Step A (Gaussian comparison for $W_n^\circ$).}
Throughout Case 1, under $\theta_{0,1}=\cdots=\theta_{0,K}=\theta_0$, define the
auxiliary $Kp$-vector
\[
W_n^\circ
:=
\left(
\begin{array}{c}
V\sqrt n(\widehat\theta_n-\theta_0)\\
\vdots\\
V\sqrt n(\widehat\theta_n-\theta_0)
\end{array}
\right),
\qquad
\|W_n^\circ\|
=
\sqrt K\,\|V\sqrt n(\widehat\theta_n-\theta_0)\|.
\]
This object is distinct from $W_{n,0}$ defined in the theorem statement, which equals $0$
under homogeneity and captures deterministic signal under heterogeneity. The quantity
$W_n^\circ$ is introduced solely as an intermediate step in this proof to isolate the
stochastic part of $\widehat W_n$.

Note that under homogeneity the aggregated estimator satisfies
$\widehat{H}_n = \frac{\sqrt{n}}{\sqrt{K}}\widehat{V}_n(\widehat{\theta}_n - \theta_0)$,
which we use in the computation below.
Since $\overline{\widehat V}_n=\widehat V_n/K$, we can rewrite
\(
\widehat H_n
=
\sqrt{nK}\,\overline{\widehat V}_n(\widehat\theta_n-\theta_0)
=
\frac{\sqrt n}{\sqrt K}\,\widehat V_n(\widehat\theta_n-\theta_0),
\)
hence
\(
V\sqrt n(\widehat\theta_n-\theta_0)=\sqrt K\,V\widehat V_n^{-1}\widehat H_n.
\)
Therefore,
\begin{equation}\label{eq::Wn_relation_pf_V2}
\|W_n^\circ\|=K\,\|V\widehat V_n^{-1}\widehat H_n\|.
\end{equation}

Fix $\eta\in(0,1)$ and define
\(
\mathcal F_V(n,\eta):=\left\{\|\widehat V_n-V\|_{\op}\le \eta\,\lambda_{\min}(V)\right\}.
\)
 Note that for any fixed $\eta\in(0,1)$, the event $\mathcal F_V(n,\eta)$ has probability
tending to $1$ by \eqref{eq::concentrationV} and a union bound; the threshold
$\eta\,\lambda_{\min}(V)$ is bounded away from $0$ for fixed $\eta$, so the required
sample size $n_0(\eta)$ is finite for each $\eta$.
On $\mathcal F_V(n,\eta)$, Weyl's inequality yields
\(
\lambda_{\min}(\widehat V_n)\ge (1-\eta)\lambda_{\min}(V),
\quad
\lambda_{\max}(\widehat V_n)\le \lambda_{\max}(V)+\eta\lambda_{\min}(V).
\)
Since $V$ and $\widehat V_n$ are symmetric positive definite, we obtain the bounds
\[
\|V\widehat V_n^{-1}\|_{\op}
\le
\frac{\lambda_{\max}(V)}{\lambda_{\min}(\widehat V_n)}
\le
\frac{\lambda_{\max}(V)}{(1-\eta)\lambda_{\min}(V)}
=:
c_+(\eta), \text{ and }
\]
\[
s_{\min}(V\widehat V_n^{-1})
\ge
\frac{\lambda_{\min}(V)}{\lambda_{\max}(\widehat V_n)}
\ge
\frac{\lambda_{\min}(V)}{\lambda_{\max}(V)+\eta\lambda_{\min}(V)}
=:
c_-(\eta),
\]
where $c_+(\eta)$ and $c_-(\eta)$ are the same constants as defined in the main text.
Combining with \eqref{eq::Wn_relation_pf_V2}, we obtain on $\mathcal F_V(n,\eta)$ the sandwich bounds
\begin{equation}\label{eq::sandwich_Wn_pf_V2}
Kc_-(\eta)\,\|\widehat H_n\|
\le
\|W_n^\circ\|
\le
Kc_+(\eta)\,\|\widehat H_n\|.
\end{equation}

Let $G\sim\mathcal N(0,I_p)$ and set $\varepsilon_n:=O(n^{-1/2}\log n)$. Applying
Proposition~\ref{prop::AEE_conc_BE} to the set $\{x:\|x\|\le u\}$ yields
\begin{equation}\label{eq::BE_norm_pf_V2}
\big|\P(\|\widehat H_n\|\ge u)-\P(\|G\|\ge u)\big|\le \varepsilon_n,
\qquad u>0.
\end{equation}
Using \eqref{eq::sandwich_Wn_pf_V2} and \eqref{eq::BE_norm_pf_V2}, for any $s>0$ we obtain
\begin{align}
\P(\|W_n^\circ\|\ge s)
&\le
\P\!\left(\|\widehat H_n\|\ge \frac{s}{Kc_+(\eta)}\right)+\P(\mathcal F_V(n,\eta)^c)
\nonumber\\
&\le
\P\!\left(\|G\|\ge \frac{s}{Kc_+(\eta)}\right)+\varepsilon_n+\P(\mathcal F_V(n,\eta)^c),
\label{eq::tail_Wn_upper_pf_V2}\\[4pt]
\P(\|W_n^\circ\|\ge s)
&\ge
\P\!\left(\|\widehat H_n\|\ge \frac{s}{Kc_-(\eta)}\right)-\P(\mathcal F_V(n,\eta)^c)
\nonumber\\
&\ge
\P\!\left(\|G\|\ge \frac{s}{Kc_-(\eta)}\right)-\varepsilon_n-\P(\mathcal F_V(n,\eta)^c).
\label{eq::tail_Wn_lower_pf_V2}
\end{align}

\noindent It remains to bound $\P(\mathcal F_V(n,\eta)^c)$. Since
\(
\|\widehat V_n-V\|_{\op}
=
\Big\|\sum_{k=1}^K(\widehat V_{n,k}-V_k)\Big\|_{\op}
\le
\sum_{k=1}^K\|\widehat V_{n,k}-V_k\|_{\op},
\)
we have, for any fixed $\eta\in(0,1)$,
\(
\mathcal F_V(n,\eta)^c
\subseteq
\bigcup_{k=1}^K
\left\{
\|\widehat V_{n,k}-V_k\|_{\op}
>
\frac{\eta}{K}\lambda_{\min}(V)
\right\}.
\)
Therefore, by a union bound and \eqref{eq::concentrationV},
\begin{equation}\label{eq::FVc_small_pf_V2}
\P(\mathcal F_V(n,\eta)^c)=o(n^{-1/2}).
\end{equation}

\medskip
\noindent\textbf{Step B: control of $\|\widehat W_n-W_n^\circ\|$ and choice of $t_n$.}
From \eqref{eq::concentrationV} and Proposition~\ref{prop::AEE_conc_BE}, for
\begin{equation}\label{eq::Gamma_pf_V2}
\Gamma_n
=
C\sqrt K\Big(\sum_{k=1}^K b_k\Big)\frac{\log n}{\sqrt n}
+
\Big(\|V\|_{\op}+o(1)\Big)\Big(\sum_{k=1}^K a_k^2\Big)^{1/2}\sqrt{\log n},
\end{equation}
we have
\(
\P\big(\|\widehat W_n-W_n^\circ\|\ge \Gamma_n\big)=o(n^{-1/2}).
\)
Hence, for any $t>0$,
\(
\P\big(\|\widehat W_n-W_n^\circ\|\ge t/2\big)
\le
o(n^{-1/2})+\mathbf 1\{\Gamma_n\ge t/2\}.
\)
We now take $t=t_n=u_n\sqrt{\log n}$. Since $(\log n)/\sqrt n=o(\sqrt{\log n})$, the first term in
\eqref{eq::Gamma_pf_V2} is $o(\sqrt{\log n})$, and therefore
\(
\frac{\Gamma_n}{\sqrt{\log n}}
=
\Big(\|V\|_{\op}+o(1)\Big)\Big(\sum_{k=1}^K a_k^2\Big)^{1/2}
+
o(1).
\)
In particular, by the condition on $u_n$, for $n$ large enough,
\begin{equation}\label{eq::Gamma_bound_sqrtlog_V2}
\Gamma_n\le A_0\sqrt{\log n}.
\end{equation}
Combining \eqref{eq::Gamma_bound_sqrtlog_V2} with the condition on $u_n$, we obtain that for $n$ large enough,
\(
\Gamma_n\le \frac{t_n}{2}.
\)
Consequently, for $n$ large enough,
\begin{equation}\label{eq::indicator_vanish_pf_V2}
\mathbf 1\{\Gamma_n\ge t_n/2\}=0,
\qquad
\P\big(\|\widehat W_n-W_n^\circ\|\ge t_n/2\big)=o(n^{-1/2}).
\end{equation}

\medskip
\noindent\textbf{Step C: asymptotic form of the Gaussian tail.}
Let $R:=\|G\|$. Then $R$ has the $\chi_p$ distribution (since $G\sim\mathcal N(0,I_p)$), and
\[
\P(R\ge x)
=
\frac{2^{1-p/2}}{\Gamma(p/2)}
\int_x^\infty r^{p-1}e^{-r^2/2}\,dr
=
c_p\int_x^\infty r^{p-1}e^{-r^2/2}\,dr,
\qquad x>0,
\]
where $c_p=2^{1-p/2}/\Gamma(p/2)$ as defined in the main text.
A standard Laplace-type expansion yields, as $x\to\infty$,
\begin{equation}\label{eq::chi_tail_asymp_pf_V2}
\P(\|G\|\ge x)=c_p\,x^{p-2}e^{-x^2/2}\big(1+o(1)\big).
\end{equation}
Since $t_n\to\infty$, we have
\(
\frac{t_n}{2Kc_+(\eta)}\to\infty,
\qquad
\frac{3t_n}{2Kc_-(\eta)}\to\infty,
\)
and thus \eqref{eq::chi_tail_asymp_pf_V2} applies at both arguments.

\medskip
\noindent\textbf{Conclusion: proof of \eqref{eq::upper_tail_What_tn_V2}--\eqref{eq::lower_tail_What_tn_V2}.}
Applying \eqref{eq::tail_What_triangle_pf_V2} with $t=t_n$, then \eqref{eq::tail_Wn_upper_pf_V2} with $s=t_n/2$,
and using \eqref{eq::FVc_small_pf_V2} and \eqref{eq::indicator_vanish_pf_V2}, yields
\begin{align*}
\P(\|\widehat W_n\|\ge t_n)
&\le
\P\!\left(\|G\|\ge \frac{t_n}{2Kc_+(\eta)}\right)
+
O(n^{-1/2}\log n)
+
o(n^{-1/2})\\
&=
c_p\left(\frac{t_n}{2Kc_+(\eta)}\right)^{p-2}
\exp\!\left(-\frac{t_n^2}{8K^2c_+(\eta)^2}\right)
\Big(1+o(1)\Big)
+
O(n^{-1/2}\log n)
+
o(n^{-1/2}),
\end{align*}
where we used \eqref{eq::chi_tail_asymp_pf_V2} in the second line. This proves \eqref{eq::upper_tail_What_tn_V2}.
Similarly, applying \eqref{eq::tail_What_reverse_pf_V2} with $t=t_n$, then \eqref{eq::tail_Wn_lower_pf_V2} with $s=3t_n/2$,
and using \eqref{eq::FVc_small_pf_V2} and \eqref{eq::indicator_vanish_pf_V2}, yields
\begin{align*}
\P(\|\widehat W_n\|\ge t_n)
&\ge
\P\!\left(\|G\|\ge \frac{3t_n}{2Kc_-(\eta)}\right)
-
O(n^{-1/2}\log n)
-
o(n^{-1/2})\\
&=
c_p\left(\frac{3t_n}{2Kc_-(\eta)}\right)^{p-2}
\exp\!\left(-\frac{9t_n^2}{8K^2c_-(\eta)^2}\right)
\Big(1+o(1)\Big)
-
O(n^{-1/2}\log n)
-
o(n^{-1/2}),
\end{align*}
where we used \eqref{eq::chi_tail_asymp_pf_V2} again. This proves \eqref{eq::lower_tail_What_tn_V2}.

{2. We follow the same lines of proof as in 1.}

\paragraph*{Decomposition of $\widehat W_n-W_{n,0}$.}
For each $k\in\{1,\ldots,K\}$,
$$
\widehat V_n\sqrt n(\widehat\theta_n-\widehat\theta_{n,k})
=
\widehat V_n\sqrt n(\widehat\theta_n-\theta_{0,V})
-\widehat V_n\sqrt n(\widehat\theta_{n,k}-\theta_{0,k})
+\widehat V_n\sqrt n(\theta_{0,V}-\theta_{0,k}),
$$
hence,
\begin{align}
&\widehat V_n\sqrt n(\widehat\theta_n-\widehat\theta_{n,k})
-
V\sqrt n(\theta_{0,V}-\theta_{0,k})
\nonumber\\
&\qquad=
\widehat V_n\sqrt n(\widehat\theta_n-\theta_{0,V})
-\widehat V_n\sqrt n(\widehat\theta_{n,k}-\theta_{0,k})
+(\widehat V_n-V)\sqrt n(\theta_{0,V}-\theta_{0,k}).
\label{eq::block_remainder_case2}
\end{align}
Stacking \eqref{eq::block_remainder_case2} over $k=1,\ldots,K$ yields
\begin{equation}\label{eq::W_hat_minus_W0_case2}
\widehat W_n-W_{n,0}
=
\left(
\begin{array}{c}
\widehat V_n\sqrt n(\widehat\theta_n-\theta_{0,V})
-\widehat V_n\sqrt n(\widehat\theta_{n,1}-\theta_{0,1})
+(\widehat V_n-V)\sqrt n(\theta_{0,V}-\theta_{0,1})
\\
\vdots
\\
\widehat V_n\sqrt n(\widehat\theta_n-\theta_{0,V})
-\widehat V_n\sqrt n(\widehat\theta_{n,K}-\theta_{0,K})
+(\widehat V_n-V)\sqrt n(\theta_{0,V}-\theta_{0,K})
\end{array}
\right).
\end{equation}

\paragraph*{Deterministic bound for $\|\widehat W_n-W_{n,0}\|$.}

Using \eqref{eq::W_hat_minus_W0_case2} and the triangle inequality,
\begin{align}
\|\widehat W_n-W_{n,0}\|
&\le
\left\|
\left(
\begin{array}{c}
\widehat V_n\sqrt n(\widehat\theta_n-\theta_{0,V})\\
\vdots\\
\widehat V_n\sqrt n(\widehat\theta_n-\theta_{0,V})
\end{array}
\right)\right\|
+
\left\|
\left(
\begin{array}{c}
\widehat V_n\sqrt n(\widehat\theta_{n,1}-\theta_{0,1})\\
\vdots\\
\widehat V_n\sqrt n(\widehat\theta_{n,K}-\theta_{0,K})
\end{array}
\right)\right\|
\nonumber\\
&\quad+
\left\|
\left(
\begin{array}{c}
(\widehat V_n-V)\sqrt n(\theta_{0,V}-\theta_{0,1})\\
\vdots\\
(\widehat V_n-V)\sqrt n(\theta_{0,V}-\theta_{0,K})
\end{array}
\right)\right\|
\nonumber\\
&=
\sqrt K\,\|\widehat V_n\|_{\op}\,\sqrt n\,\|\widehat\theta_n-\theta_{0,V}\|
+
\|\widehat V_n\|_{\op}
\Big(\sum_{k=1}^K n\|\widehat\theta_{n,k}-\theta_{0,k}\|^2\Big)^{1/2}
\nonumber\\
&\quad+
\|\widehat V_n-V\|_{\op}\,\sqrt n\,
\Big(\sum_{k=1}^K \|\theta_{0,V}-\theta_{0,k}\|^2\Big)^{1/2}.
\label{eq::det_bound_W_hat_minus_W0_case2}
\end{align}

\paragraph*{Concentration bound for $\|\widehat W_n-W_{n,0}\|$.}

From \eqref{eq::conc_theta_agg} (aggregated estimator), \eqref{eq::concentration}
(local estimators), and \eqref{eq::concentrationV} (sensitivity matrices),
\begin{equation}\label{eq::W_hat_minus_W0_conc_case2}
\P\!\left(\|\widehat W_n-W_{n,0}\|\le \Gamma_{n,\mathrm{alt}}\right)
=1-o(n^{-1/2}),
\end{equation}
where
\begin{align*}
\Gamma_{n,\mathrm{alt}}
&:=
\Big(\|V\|_{\op}+\Big(\sum_{k=1}^K b_k\Big)r_n\Big)
\Big(\sqrt K\,C+\Big(\sum_{k=1}^K a_k^2\Big)^{1/2}\Big)\sqrt{\log n}
\nonumber\\
&\quad+
\Big(\sum_{k=1}^K b_k\Big)\sqrt{\log n}\,
\Big(\sum_{k=1}^K \|\theta_{0,V}-\theta_{0,k}\|^2\Big)^{1/2}.
\end{align*}
In particular, $\Gamma_{n,\mathrm{alt}}=O(\sqrt{\log n})$.

\paragraph*{A lower bound for $\P(\|\widehat W_n\|\le t_n)$.}
By the triangle inequality,
\(
\|\widehat W_n\|
\le
\|W_{n,0}\|+\|\widehat W_n-W_{n,0}\|.
\)
Therefore,
$\Big\{\|\widehat W_n-W_{n,0}\|\le t_n-\|W_{n,0}\|\Big\}
\subseteq
\Big\{\|\widehat W_n\|\le t_n\Big\}.$
Since $t_n-\|W_{n,0}\|=u_n\sqrt{\log n}$ by
$t_n  = \sqrt n\,\Delta_V+u_n\sqrt{\log n}$, we obtain
\begin{equation}\label{eq::lower_prob_start_case2}
\P(\|\widehat W_n\|\le t_n)
\ge
\P\!\left(\|\widehat W_n-W_{n,0}\|\le u_n\sqrt{\log n}\right).
\end{equation}

Now recall \eqref{eq::W_hat_minus_W0_conc_case2}. Since $\Gamma_{n,\mathrm{alt}}=O(\sqrt{\log n})$,
there exists a deterministic sequence $A_n>0$ such that
$\Gamma_{n,\mathrm{alt}}\le A_n\sqrt{\log n}$
and
$A_n=O(1).$
Hence, if $u_n\ge A_n$ (which holds for all $n$ large enough as soon as $\lim\inf u_n > \bar A_1$), then
\[
\{\|\widehat W_n-W_{n,0}\|\le \Gamma_{n,\mathrm{alt}}\}
\subseteq
\{\|\widehat W_n-W_{n,0}\|\le u_n\sqrt{\log n}\}.
\]
Combining this inclusion with \eqref{eq::W_hat_minus_W0_conc_case2} and \eqref{eq::lower_prob_start_case2} yields,
for all $n$ large enough,
\begin{equation}\label{eq::final_lower_bound_case2}
\P(\|\widehat W_n\|\le t_n)
\ge
\P\!\left(\|\widehat W_n-W_{n,0}\|\le \Gamma_{n,\mathrm{alt}}\right)
=
1-o(n^{-1/2}).
\end{equation}

\medskip
\noindent\textbf{Conclusion.}
Under heterogeneity, with the threshold
$t_n=\sqrt n\,\Delta_V+u_n\sqrt{\log n}$, we obtain the bound
\(
\P(\|\widehat W_n\|\le t_n)\ge 1-o(n^{-1/2}).
\)
Since $\Gamma_{n,\mathrm{alt}}=O(\sqrt{\log n})=o(\sqrt n)$ and $\Delta_V>0$, we have
\(
\frac{\Gamma_{n,\mathrm{alt}}}{\sqrt n\,\Delta_V}\to 0.
\)
Hence for any fixed $\varepsilon\in(0,1)$, for $n$ large enough,
\(
(1-\varepsilon)\sqrt n\,\Delta_V
\le
\sqrt n\,\Delta_V-\Gamma_{n,\mathrm{alt}},
\quad
(1+\varepsilon)\sqrt n\,\Delta_V
\ge
\sqrt n\,\Delta_V+\Gamma_{n,\mathrm{alt}}.
\)
Therefore,
\begin{align*}
\P\!\left(\|\widehat W_n\|\le (1-\varepsilon)\sqrt n\,\Delta_V\right)
&=
o(n^{-1/2}),\quad
 \P\!\left(\|\widehat W_n\|\le (1+\varepsilon)\sqrt n\,\Delta_V\right)
=
1-o(n^{-1/2}).
 \end{align*}
\hfill $\square$

\paragraph{Proof of Theorem~\ref{theo:boot_merge_split_control}}
The proof follows the same three-step structure
(Gaussian comparison, remainder control, $\chi_p$ tail asymptotics)
as the proof of Theorem~\ref{thm::tail_bounds_WhatV2}, under the
following replacements throughout: population quantities $V_k$,
$\theta_{0,k}$, $c_\pm(\eta)$, $\Delta_V$, $\bar A_0$, $\bar A_1$,
$\Gamma_{n,\mathrm{alt}}$ are replaced by their sample analogues
$\widehat V_{n,k}$, $\widehat\theta_{n,k}$, $c_\pm^*(\eta)$,
$\widehat\Delta_V^{(i,j)}$, $\bar A_0^*$, $\bar A_1^*$,
$\Gamma_{n,\mathrm{alt}}^*$; assumptions (A8)--(A9) are replaced by
(A10)--(A11); and $O(\cdot)$, $o(\cdot)$ are replaced by
$O_p(\cdot)$, $o_p(\cdot)$.

Three points specific to this setting are noted.
\begin{enumerate}
\item \emph{Conditioning.} Throughout, we condition on
$\mathcal{Z}_n$, so all substituted quantities are
$\mathcal{Z}_n$-measurable and treated as fixed. The bootstrap
randomness enters only through $\widehat\theta_{n,k}^{(r)}$ and
$\widehat V_{n,k}^{(r)}$.

\item \emph{Uniformity over pairs and rounds.} Since
$\widehat V_{n,k}$ does not depend on $r$, the bounds
$c_\pm^*(\eta)$ and $\Gamma_{n,\mathrm{alt}}^*$ are the same
$\mathcal{Z}_n$-measurable expression for every round $r$ and
every pair $(S_i,S_j)$. Taking the maximum over the finite sets
$\mathcal{C}$, $\mathcal{D}$, and $\{1,\ldots,R(n)\}$ therefore
preserves the bound.

\item \emph{Random signal.} In part~(ii), $\widehat\Delta_V^{(i,j)}$
is $\mathcal{Z}_n$-measurable and plays the role of the deterministic
$\Delta_V$ in Case~2 of Theorem~\ref{thm::tail_bounds_WhatV2}.
Conditionally on $\mathcal{Z}_n$, it is fixed, and the detectability
condition \eqref{eq::detectability_boot} ensures it dominates the
threshold $t_n+\Gamma_{n,\mathrm{alt}}^*$. The reverse triangle
inequality then gives the type-II bound exactly as in Case~2.
\end{enumerate}
The conclusions \eqref{eq::type1_bound}--\eqref{eq::type2_bound}
hold with probability tending to $1$,
because (A10)--(A11) are conditional-in-probability assumptions,
so the resulting bounds are $O_p$ and $o_p$.
\hfill $\square$
\subsection{Examples of models}
We illustrate the preceding results with several applications.
Throughout this section,  for each centre $k$,  the sample $(Z_{1, k}, \ldots, Z_{n, k})$ is assumed i.i.d. 
In all examples,  the universal resampling scheme  applies without modification. We provide a unified treatment of \emph{robust linear regression} (sec. ~6.2.1) and \emph{one-parameter GLMs} (sec. ~6.2.2),  and state sufficient conditions ensuring the validity of both the nonparametric and the weighted (multiplier) bootstrap. For sec. ~6.2.3 and~6.2.4,  we rely on existing results in the literature. The proof of Proposition~3 is given in supplementary material; the proofs of Propositions~4–6 are omitted for brevity as they follow standard arguments.
 We will also need of the  following “no linear-hyperplane support” condition ensuring identifiability of \(\theta_{0, k}\)  : 

\begin{description}
    \item[Ident.]  For each $k, $ 
\(
\mathbb{P}(\nu^\top X_{1, k} =0)=1 \ \Rightarrow\ \nu=0
\quad\text{for all }\nu\in\mathbb{R}^d.
\)
\end{description}

\subsubsection{Robust Linear Regression Models}
For each centre \(k \in \{1, \ldots, K\}\),  consider the linear regression model  : 
\begin{equation}
    \label{eq::RLR}
    Y_{i, k} \;=\; X_{i, k}^\top \theta_{0, k} + \varepsilon_{i, k},  
    \qquad i=1, \ldots, n, \; \theta_{0, k}\in\Theta, 
\end{equation}
where \(\Theta \subset \mathbb{R}^p\) is compact. The estimator \(\widehat\theta_{n, k}\) of the true parameter \(\theta_{0, k}\) is defined as 
\begin{equation}
    \label{eq::RLRestimator}
    \frac{1}{n}\sum_{i=1}^n 
\psi_{k, \delta}\!\bigl(Y_{i, k} - X_{i, k}^\top \widehat\theta_{n, k}\bigr)\,  X_{i, k} \;=\; 0, 
\end{equation}
where \(\psi_{k, \delta} = \varphi'_{k, \delta}\) is the derivative of a convex loss \(\varphi_{k, \delta}\); both may depend on a tuning (hyper-)parameter \(\delta\),  possibly centre-specific. Typical choices include the Huber \citep{huber2011robust} and pseudo-Huber losses.  Although the overall influence function may be unbounded (due to leverage in the covariates), 
this estimator still exhibits resistance to   outliers: large residuals receive
capped weights whenever \(\psi_{k, \delta}\) satisfies \textbf{RLR1.}(ii) below.

We rely on the following set of centre-specific assumptions:
\begin{description}
    \item[RLR1.] \textbf{(Loss regularity)} 
    (i) The loss \(\varphi_{k, \delta}\) is non-negative,  C$1$-differentiable,  convex and coercive,  i.e.
    \(\varphi_{k, \delta}(t)\to\infty\) as \(|t|\to\infty\).
    (ii) There exists \(\tau>0\) such that
    \[
      \sup_{t\in \R} |\psi_{k, \delta}(t)| 
      \;+\;
      \sup_{u\neq v}\left|\frac{\psi_{k, \delta}(u)-\psi_{k, \delta}(v)}{u-v}\right|
      \;\le\; \tau.
    \]

   \item[RLR2.] \textbf{(Errors)} 
(i) The error \(\varepsilon_{1, k}\) is independent of \(X_{1, k}\). 
(ii) The moment conditions
\[
\mathbb{E}\, \psi_{k, \delta}(\varepsilon_{1, k})=0
\quad\text{and}\quad
\mathbb{E}\, s_{k, \delta}(\varepsilon_{1, k})>0
\]
hold,  where \(s_{k, \delta}\) is defined as follows:
\(s_{k, \delta}(t)=\psi'_{k, \delta}(t)\) on the set where \(\psi_{k, \delta}\) is differentiable,  and we assume
\(\mathbb{P}\big(\varepsilon_{1, k}\in N_{\psi}\big)=0\),  with \(N_{\psi}\) the set of nondifferentiability points of \(\psi_{k, \delta}\).

\item[RLR3.] \textbf{(Design)} 
    \(\mathbb{E}\, \|X_{1, k}\|^4<\infty\).
\end{description}

Assumption \textbf{RLR1}(i) states the standard coercivity condition in robust estimation.
Assumption \textbf{RLR1}(ii) requires \(\psi_{k, \delta}\) to be bounded and globally Lipschitz; this is satisfied by many robust losses (e.g.,  Huber,  pseudo-Huber,  log-cosh). But Lipschitz function is differentiable almost everywhere,  so no separate differentiability assumption is needed in \textbf{RLR2}.
The independence in \textbf{RLR2}(i) is standard for regression models. For \textbf{RLR2}(ii),  one may 
(a) work with the classical derivative by assuming \(\mathbb{P}(\varepsilon_{1, k}\in N_\psi)=0\),  where \(N_\psi\) is the (null) set of non differentiability points of \(\psi_{k, \delta}\),  and take \(s_{k, \delta}=\psi'_{k, \delta}\). 
The moment condition \(\mathbb{E}\, s_{k, \delta}(\varepsilon_{1, k})>0\) is then well-defined.
The moment condition in \textbf{RLR2}(ii) is,  for example,  guaranteed if (a) \(\varphi_{k, \delta}\) is even and the distribution of \(\varepsilon_{1, k}\) is symmetric about zero (so \(\mathbb{E}\, \psi_{k, \delta}(\varepsilon_{1, k})=0\)),  and
(b) there exist constants \(\nu_1, \nu_2>0\) such that
\(
\inf_{|t|\le \nu_2}\ \, \psi'_{k, \delta}(t)\ \ge\ \nu_1
\quad\text{and}\quad
\mathbb{P}(|\varepsilon_{1, k}|\le \nu_2)>0.
\)
In particular,  this   lower bound holds for the Huber loss with \(\nu_1=1\),  \(\nu_2=\delta\) (taking any selection at the kinks),  and for pseudo-Huber and log-cosh with \(\nu_1=\psi'(\nu_2)\) since these scores are smooth. The conditions in \textbf{RLR3} specify the moment requirements used to derive asymptotic normality of the estimator.  

\begin{prop}
\label{prop::GoldenRobust}
Suppose that,  for each \(k\),  \(\theta_{0, k}\) lies in the interior of \(\Theta\),  and that \textbf{RLR1.} to \textbf{RLR3.},  (\textbf{A7}) and \textbf{Ident.} hold. Then the partition \(\widehat{C}_n^{(R(n))}\) produced by Algorithm~\ref{alg:CoC-algorithm_bootstrap_iterations},  with,  for each \(k\) and each bootstrap replicate \(r=1, \ldots, R(n)\), 
$$
\frac{1}{n}\sum_{i=1}^n 
\psi_{k, \delta}\!\bigl(Y_{i, k}^{(r)} - X_{i, k}^{{(r)}^\top} \widehat{\theta}_{n, k}^{(r)}\bigr)\, X_{i, k}^{(r)} \;=\; 0, 
$$
$$
\widehat V_{n, k}  \, =\,  \frac{1}{n}\sum_{i=1}^n 
\psi'_{k, \delta}\!\bigl(Y_{i, k}   - X_{i, k}  ^\top \widehat{\theta}_{n, k}  \bigr)\, X_{i, k}  X_{i, k}^\top, 
\quad
\widehat Q_{n, k} \, =\,  \frac{1}{n}\sum_{i=1}^n 
\psi_{k, \delta}^2\!\bigl(Y_{i, k}   - X_{i, k}^{ \top }\widehat{\theta}_{n, k}  \bigr)\, X_{i, k}   X_{i, k}^{\top}, 
$$
where \(\{(Y_{i, k}^{(r)},  X_{i, k}^{(r)}) : i=1, \ldots, n\}\) comes from either
(i) the nonparametric bootstrap (resampling with replacement) or
(ii) the weighted bootstrap with weights of mean \(1\) and variance \(1\), 
enjoys the Golden-Partition Recovery property.
\end{prop}

 \paragraph{Proof of Proposition~\ref{prop::GoldenRobust}}
Relying on Theorem~10.16 of \cite{kosorok2008introduction},  we verify:
\begin{enumerate} 
\item Identifiability:
\[
\theta_{0, k} \;=\; \arg\min_{\theta\in\Theta} \; \mathbb{E}\, \varphi_{k, \delta}\!\big(Y_{1, k}-\theta^\top X_{1, k}\big)
\quad\text{is uniquely defined.}
\]
\item For $y\in\mathbb{R}$ and $x=(x_1, \dots, x_p)$,  set
$g_{i, \theta}(y, x):=\psi_{k, \delta}\!\big(y-\theta^\top x\big)\, x_i$.
Then $\mathcal{F}_i:=\{g_{i, \theta}:\theta\in\Theta\}$ is Glivenko–Cantelli.
\item For some $\eta>0$,  $\mathcal{F}_{i, \eta}:=\{g_{i, \theta}:\theta\in\Theta, \, \|\theta-\theta_{0, k}\|\le\eta\}$ is Donsker and
\[
\mathbb{E}\big\|
\psi_{k, \delta}(Y_{1, k}-X_{1, k}^\top\theta)\, X_{1, k}
-
\psi_{k, \delta}(Y_{1, k}-X_{1, k}^\top\theta_{0, k})\, X_{1, k}
\big\|^2 \;\to\; 0 \quad \text{as }\|\theta-\theta_{0, k}\|\to 0.
\]
\item The vector $\psi_{k, \delta}(Y_{1, k}-\theta_{0, k}^\top X_{1, k})X_{1, k}$ is square integrable and
$\theta \mapsto \mathbb{E}\big[\psi_{k, \delta}(Y_{1, k}-\theta^\top X_{1, k})X_{1, k}\big]$
is differentiable at $\theta_{0, k}$ with a nonsingular derivative matrix.
\end{enumerate}

\paragraph*{(i) Identifiability.}
\textit{(i–1) Coercivity of the population loss.}
Let $\theta_m$ be any sequence with $\|\theta_m\|\to\infty$. By compactness of the unit sphere, 
there is a subsequence (not relabeled) with $u_m:=\theta_m/\|\theta_m\|\to u$ and $\|u\|=1$.
Set $Z:=u^\top X_{1, k}$. By \textbf{RLR3},  $\mathbb{P}(|Z|>0)>0$.

For $\omega$ with $|Z(\omega)|>0$,  by continuity of $v\mapsto v^\top X_{1, k}(\omega)$, 
$|u_m^\top X_{1, k}(\omega)| \ge |Z(\omega)|/2$ for all large $m$. Hence
\[
|\theta_m^\top X_{1, k}(\omega)|
=\|\theta_m\|\, |u_m^\top X_{1, k}(\omega)|
\ \ge\ \tfrac12\, \|\theta_m\|\, |Z(\omega)|
\ \xrightarrow[m\to\infty]{}\ \infty.
\]
Thus $|Y_{1, k}(\omega)-\theta_m^\top X_{1, k}(\omega)|\to\infty$ for all $\omega$ with $|Z(\omega)|>0$.
Since $\varphi_{k, \delta}$ is coercive   by assumption,  it follows that
\[
\varphi_{k, \delta}\!\big(Y_{1, k}-\theta_m^\top X_{1, k}\big)\ \to\ \infty
\quad\text{pointwise on }\{|Z|>0\}.
\]

Fix $M>0$ and define the truncation $\varphi_{k, \delta}^{(M)}(t):=\min\{\varphi_{k, \delta}(t), \, M\}$.
Then,  on $\{|Z|>0\}$, 
\[
\varphi_{k, \delta}^{(M)}\!\big(Y_{1, k}-\theta_m^\top X_{1, k}\big)\ \to\ M, 
\]
while on $\{|Z|=0\}$ we   have $\liminf_{m\to\infty}\varphi_{k, \delta}^{(M)}(Y_{1, k}-\theta_m^\top X_{1, k})\ge 0$.
Therefore, 
\[
\liminf_{m\to\infty}\, 
\varphi_{k, \delta}^{(M)}\!\big(Y_{1, k}-\theta_m^\top X_{1, k}\big)
\ \ge\ M\, \mathbf{1}_{\{|Z|>0\}}.
\]
By Fatou’s lemma, 
\[
\liminf_{m\to\infty}\, 
\mathbb{E}\, \varphi_{k, \delta}^{(M)}\!\big(Y_{1, k}-\theta_m^\top X_{1, k}\big)
\ \ge\ M\, \mathbb{P}(|Z|>0).
\]
Because $\varphi_{k, \delta}^{(M)}\le \varphi_{k, \delta}$ and $M>0$ is arbitrary,  we conclude that
\[
\mathbb{E}\, \varphi_{k, \delta}\!\big(Y_{1, k}-\theta_m^\top X_{1, k}\big)\ \xrightarrow[\ \|\theta_m\|\to\infty\ ]{}\ \infty.
\]
Hence the population objective
\(
g_k(\theta):=\mathbb{E}\, \varphi_{k, \delta}\!\big(Y_{1, k}-\theta^\top X_{1, k}\big)
\)
is coercive.

\medskip
\textit{(i–2) $\theta_{0, k}$ is a minimizer.}
Because $|\psi_{k, \delta}|\le \tau$ and $\mathbb{E}\|X_{1, k}\|<\infty$, 
we may differentiate under the expectation (dominated convergence),  obtaining
\[
\nabla_\theta g_k(\theta)
= -\, \mathbb{E}\!\left[\psi_{k, \delta}\!\big(Y_{1, k}-\theta^\top X_{1, k}\big)\, X_{1, k}\right].
\]
At $\theta=\theta_{0, k}$,  with $Y_{1, k}=X_{1, k}^\top\theta_{0, k}+\varepsilon_{1, k}$ and $\varepsilon_{1, k}\, \text{ independent from } X_{1, k}$, 
\[
\nabla_\theta g_k(\theta_{0, k})
= -\, \mathbb{E}\big[\psi_{k, \delta}(\varepsilon_{1, k})\big]\;\mathbb{E}[X_{1, k}]
= 0, 
\]
since $\mathbb{E}\, \psi_{k, \delta}(\varepsilon_{1, k})=0$.
As $g_k$ is convex,  any stationary point is a (global) minimizer; hence $\theta_{0, k}$ minimizes $g_k$.

\medskip
\textit{(i–3) Local strong convexity at $\theta_{0, k}$ and uniqueness.}
When $\psi_{k, \delta}$ is differentiable (assumption \textbf{RLR2}),  the Hessian exists and
\[
\nabla_\theta^2 g_k(\theta)
= \mathbb{E}\!\left[\psi'_{k, \delta}\!\big(Y_{1, k}-\theta^\top X_{1, k}\big)\, X_{1, k}X_{1, k}^\top\right].
\]
For any nonzero $u\in\mathbb{R}^p$, 
\[
u^\top \nabla_\theta^2 g_k(\theta_{0, k})\, u
= \mathbb{E}\!\left[\psi'_{k, \delta}(\varepsilon_{1, k})\right]\; u^\top \mathbb{E}\!\left[X_{1, k}X_{1, k}^\top\right] u
\;>\; 0, 
\]
by \textbf{RLR2(ii)} (ensuring $\mathbb{E}[\psi'_{k, \delta}(\varepsilon_{1, k})]>0$) and \textbf{RLR3} (non-collinearity: $\mathbb{E}[X_{1, k}X_{1, k}^\top]\succ 0$).
Thus $g_k$ is strictly convex in a neighborhood of $\theta_{0, k}$. Since $g_k$ is convex globally, 
it cannot have two distinct minimizers.
Therefore $\theta_{0, k}$ is the \emph{unique} minimizer of $g_k$,  equivalently the unique solution to the population estimating equation
\[
\mathbb{E}\!\left[\psi_{k, \delta}\!\big(Y_{1, k}-X_{1, k}^\top\theta\big)X_{1, k}\right]=0.
\]

{\sl Proof of (ii).}

By \textbf{RLR1(ii)} (Lipschitz/ bounded slope of $\psi_{k, \delta}$),    for any $\theta_1, \theta_2\in\Theta$ and $(y, x)\in\mathbb{R}\times\mathbb{R}^p$, 
\[
\begin{aligned}
\big|g_{i, \theta_1}(y, x)-g_{i, \theta_2}(y, x)\big|
&= \big|\psi_{k, \delta}(y-\theta_1^\top x)-\psi_{k, \delta}(y-\theta_2^\top x)\big|\cdot |x_i| \\
&\le  \tau\, \|x\|^2\, \|\theta_1-\theta_2\|.
\end{aligned}
\]
 
By \textbf{RLR3},  $\mathbb{E}\, m(Y_{1, k}, X_{1, k})=\tau\, \mathbb{E}\|X_{1, k}\|^2<\infty$.
Since $\Theta$ is compact and $\theta\mapsto g_{i, \theta}(y, x)$ is continuous for each $(y, x)$, 
the family $\mathcal{F}_i:=\{g_{i, \theta}:\theta\in\Theta\}$ is $L_1(P_{1,k})$–Lipschitz in $\theta$
with an integrable envelop,  this implies  that $\mathcal{F}_i$ is $P_{1,k}$–Glivenko–Cantelli.

Now define the vector-valued map $g_\theta(y, x):=\psi_{k, \delta}(y-\theta^\top x)\, x\in\mathbb{R}^p$
and the class $\mathcal{F}:=\{g_\theta:\theta\in\Theta\}$.
Since each coordinate class $\mathcal{F}_i$ is $P_{1,k}$–Glivenko–Cantelli,  $\mathcal{F}$ is Glivenko–Cantelli in the vector sense (see,  eg. \cite[ Part~III,  Ch.~5]{bolthausen2002lectures}).

{\sl Proof of (iii).} 
Now,   we need to require $\E m^2(Y_{1, k}, X_{1, k})< \infty$ and also $\E g^2_{i, \theta}(Y_{1, k}, X_{1, k}) < \infty$ for $\theta \in \Theta.$
All of this condition are met   following \textbf{RLR1.(ii)} and  the integrability condition in  \textbf{RLR3.}.
For the second condition of (iii),  since the score is continuous in $\theta$ as $\varphi_{k, \delta}$ is C$1$-differentiable,  it is sufficient to show that  
$$\E \sup_{\|\theta - \theta_0\|\leq \eta} \|\psi_{k,  \delta}(Y_{1, k} - \theta^\top X_{1, k})X_{1, k}\|^2 < \infty$$
and the result   follows by the dominated convergence theorem.
This latter condition holds since $\psi_{k,  \delta}$ is bounded and $\E \|X_{1, k}\|^4 < \infty.$ As for the point (ii.) the class $\mathcal{F}_\eta:=\{g_\theta:\theta\in\Theta, \, \|\theta-\theta_{0, k}\|\le\eta\}$ is $P_{1,k}-$Donsker following eg. \cite[ Part~III,  Ch.~6]{bolthausen2002lectures}.

{\sl Proof of (iv).} 

We   deduce this point from \textbf{RLR3.} and since   $\psi_{n, k}$ is bounded.

{\sl Conclusion} (i),  (ii),  (iii) and (iv) allow us to obtain (\textbf{A1}) to (\textbf{A4}) and (\textbf{A6}).   (\textbf{A5})    is also verified from the continuous mapping theorem for convergence in probability following the consistency of the estimator. 
\hfill $\square$

\subsubsection{One Parameter Exponential Family GLM}

In this framework,  conditional on the  covariate \(X_{i, k}\),  the response \(Y_{i, k}\) follows a canonical one-parameter exponential family:
\begin{equation}
\label{eq::glm}
\mathbb{P}(Y_{i, k}\in\Delta\mid X_{i, k})
\;=\;
\int_\Delta
c(y)\exp\!\bigl\{y\, X_{i, k}^{\!\top}\theta_{0, k}-a(X_{i, k}^{\!\top}\theta_{0, k})\bigr\}\, \mu(dy), 
\quad
\theta_{0, k}\in\Theta, \;
\Delta\subseteq\mathcal{Y}\subseteq\mathbb{R}, 
\end{equation}
where \(\mu\) denotes Lebesgue measure for continuous \(Y_{i, k}\) or counting measure for discrete \(Y_{i, k}\),  and \(a(\cdot)\) is the cumulant function. The negative log-likelihood and its score function are given by
\[
\ell(\theta)
\;=\;
\sum_{i=1}^n\!\bigl\{a(X_{i, k}^{\!\top}\theta)-Y_{i, k}X_{i, k}^{\!\top}\theta\bigr\}, 
\qquad
s(\theta)
\;=\;
-\sum_{i=1}^n\!\bigl(Y_{i, k}-a'(X_{i, k}^{\!\top}\theta)\bigr)X_{i, k}.
\]

The estimator \(\widehat\theta_{n, k}\) of the true parameter \(\theta_{0, k}\) is defined as 
\begin{equation}
    \label{eq::GLMestimator}
    \frac{1}{n}\sum_{i=1}^n 
 \!\bigl(Y_i-a'(X_i^{\!\top}\widehat\theta_{n, k})\bigr)X_i \;=\; 0.
\end{equation}

We make use of  the following set of centre-specific assumptions:
\begin{description}
    \item[GLM1.] There exists a positive real-valued function $A$ such that 
    $$
    |a'(X_{1, k}^\top \theta_1) - a'(X_{1, k}^\top \theta_2)| \leq A(X_{1, k})\|\theta_1-\theta_2\|.
    $$
    \item[GLM2.] The following moment condition hold :
    $$
    \E a''(X_{1, k}^\top \theta_{0, k})\|X_{1, k}\|^2 < \infty \,  \text{ and } \,  \E A^2(X_{1, k})\|X_{1, k}\|^2 < \infty.
    $$
    \item[GLM3.] For some $\eta > 0, $
    $$
    \E \sup_{\|\theta - \theta_{0, k}\| \leq \eta}\|(Y_{1, k}-a'(X_{1, k}^\top \theta))X_{1, k}\|^2 < \infty.
    $$
\end{description}

\begin{prop}
\label{prop::GLM}
Suppose that,  for each \(k\),  \(\theta_{0, k}\) lies in the interior of \(\Theta\),  and that \textbf{GLM1.} to \textbf{GLM3.},  (\textbf{A7}) and \textbf{Ident.} hold. Then the partition \(\widehat{C}_n^{(R(n))}\) produced by Algorithm~\ref{alg:CoC-algorithm_bootstrap_iterations},  with,  for each \(k\) and each bootstrap replicate \(r=1, \ldots, R(n)\), 

$$
\frac{1}{n}\sum_{i=1}^n 
 \!\bigl(Y_{i, k}^{(r)} -a'(X_{i, k}^{(r)\!\top}\widehat{\theta}_{n, k}^{(r)})\bigr)X_{i, k}^{(r)} \;=\; 0, 
$$
$$
\widehat V_{n, k}  \, =\,  \frac{1}{n}\sum_{i=1}^n 
a''\!\bigl(X_{i, k}^{\, \top} \widehat{\theta}_{n, k} \bigr)\, X_{i, k}  X_{i, k}^{\, \top}, 
\quad
\widehat Q_{n, k} \, =\,  \frac{1}{n}\sum_{i=1}^n 
\bigl(Y_{i, k}  -a'(X_{i, k}^{\!\top}\widehat{\theta}_{n, k})\bigr)^2\, X_{i, k}  X_{i, k}^{\, \top}, 
$$
where \(\{(Y_{i, k}^{(r)},  X_{i, k}^{(r)}) : i=1, \ldots, n\}\) comes from either
(i) the nonparametric bootstrap  or
(ii) the weighted bootstrap with weights of mean \(1\) and variance \(1\), 
enjoys the Golden-Partition Recovery property.
\end{prop}
 
\paragraph*{Examples (logistic and Poisson,  canonical links).}
\emph{Logistic regression:} \(Y\in\{0, 1\}\),  \(a(\eta)=\log(1+e^{\eta})\), 
\(a'(\eta)=\mu(\eta)=\frac{e^{\eta}}{1+e^{\eta}}\in(0, 1)\), 
\(a''(\eta)=\mu(\eta)\{1-\mu(\eta)\}\le 1/4\).
Then
\(|a'(\eta_1)-a'(\eta_2)|\le \tfrac14|\eta_1-\eta_2|\), 
so \textbf{GLM1} holds with \(A(X)=\tfrac14\|X\|\).
For \textbf{GLM2},  since \(a''(\eta)\le 1/4\),  it suffices that \(\E\|X\|^2<\infty\) (first part) and \(\E\|X\|^4<\infty\) (second part,  via \(A(X)=\tfrac14\|X\|\)).
Because \(Y\) and \(a'\) are bounded,  \textbf{GLM3} reduces to \(\E\|X\|^2<\infty\).

\emph{Poisson regression.}
Let \(a(\eta)=e^\eta\) so that \(a'(\eta)=a''(\eta)=e^\eta\).
Since \(\Theta\) is compact,  set \(M:=\sup_{\theta\in\Theta}\|\theta\|<\infty\).
For any \(\theta_1, \theta_2\in\Theta\),  the mean–value theorem and Cauchy–Schwarz give
\[
|a'(X^\top\theta_1)-a'(X^\top\theta_2)|
\le e^{X^\top\xi}\, |X^\top(\theta_1-\theta_2)|
\le e^{M\|X\|}\, \|X\|\, \|\theta_1-\theta_2\|, 
\]
for some \(\xi\) on the segment \([\theta_1, \theta_2]\subset\Theta\).
Hence \textbf{GLM1} holds   with \(A(X)=e^{M\|X\|}\|X\|\).
For \textbf{GLM2},  it suffices that
\(
 \E\{e^{2M\|X\|}\|X\|^4\}<\infty.
\)
For \textbf{GLM3},  since \(\E(Y^2\mid X)=e^{X^\top\theta_{0, k}}+e^{2X^\top\theta_{0, k}}
\le 2 e^{2M\|X\|}\),  a sufficient condition is
\(
\E\{e^{2M\|X\|}\|X\|^2\}<\infty.
\)

\subsubsection{Quantile Regression}
 Fix a quantile level $\tau\in(0, 1)$. At centre $k$,  let $\{(Y_{i, k}, X_{i, k})\}_{i=1}^n$ be the  observations with $X_{i, k}\in\mathbb{R}^p$ and model $Y_{i, k}=X_{i, k}^\top\theta_{0, k}+ \varepsilon_{i, k}$,  where $\theta_{0, k}\in \Theta \subseteq \mathbb{R}^p$ is the true parameter and $\varepsilon_{i, k}$ has conditional $\tau$-quantile zero given $X_{i, k}$. The local estimator is the quantile regression  
\[
\widehat\theta_{n, k}\in\arg\min_{\theta\in\mathbb{R}^p}\;\frac{1}{n}\sum_{i=1}^n \rho_\tau\bigl(Y_{i, k}-X_{i, k}^\top\theta\bigr), 
\qquad
\rho_\tau(u)=u\bigl(\tau-\mathbf{1}\{u<0\}\bigr).
\]
Let $\psi_\tau(u)=\tau-\mathbf{1}\{u<0\}$ and write $\varepsilon_{i, k}=Y_{i, k}-X_{i, k}^\top\theta_{0, k}$. Under standard conditions for quantile regression, 
\[
\sqrt{n}\, \bigl(\widehat\theta_{n, k}-\theta_{0, k}\bigr)
\;=\;
V_k^{-1}\left(\frac{1}{\sqrt{n}}\sum_{i=1}^n X_{i, k}\, \psi_\tau(\varepsilon_{i, k})\right)\;+\;o_p(1), 
\]
with
\[
V_k \;:=\; \mathbb{E}\!\left[\, f_{\varepsilon|X, k}\!\left(0\, \big|\, X_{i, k}\right)\, X_{i, k}X_{i, k}^\top\right], 
\qquad
Q_k \;:=\; \tau(1-\tau)\, \mathbb{E}\!\left[\, X_{i, k}X_{i, k}^\top\right], 
\]
where $f_{\varepsilon \mid X, k}(0 \mid x)$ denotes the conditional density of $\varepsilon_{i, k}$ at $0$ given $X_{i, k}=x$,  which is assumed to be supported by $\mathbb{R}$.
Hence \textbf{A1}–\textbf{A4} hold with
\[
\varepsilon_{n, k}\;:=\;\frac{1}{\sqrt{n}}\sum_{i=1}^n X_{i, k}\, \psi_\tau(\varepsilon_{i, k})
\;\Rightarrow\; \mathcal N\!\bigl(0, Q_k\bigr), 
\quad V_k\succ 0.
\]

\emph{Consistent plug-in estimators (\textbf{A5}).}
Define
\[
\widehat Q_{n, k}\;:=\;\tau(1-\tau)\, \frac{1}{n}\sum_{i=1}^n X_{i, k}X_{i, k}^\top, 
\qquad
\widehat V_{n, k}\;:=\;\frac{1}{n}\sum_{i=1}^n \widehat f_{\varepsilon|X, k}\!\left(0\, \big|\, X_{i, k}\right)\, X_{i, k}X_{i, k}^\top, 
\]
where $\widehat f_{\varepsilon|X, k}(0\, |\, x)$ is any consistent estimator of $f_{\varepsilon|X, k}(0\, |\, x)$. Then $\widehat V_{n, k}\xrightarrow{p}V_k$ and $\widehat Q_{n, k}\xrightarrow{p}Q_k$. The sufficient conditions ensuring the validity of the nonparametric bootstrap in the next proposition are adapted from \cite[Theorem~3]{hahn1995bootstrapping}.

\begin{prop}
\label{prop::QR}
Suppose that,  for each \(k\),  \(\theta_{0, k}\) lies in the interior of \(\Theta\),    (\textbf{A7})  \textbf{Ident.} hold and that 
\begin{enumerate} 
 \item[(i)] There exists a measurable envelope $F(X_{1, k})$ such that
  \[
    f_{\varepsilon|X, k}(u\mid X_{1, k})\le F(X_{1, k})\quad\text{for all }u\in\mathbb{R}\ \text{a.s., } 
 \text{ and } 
    \mathbb{E}\big[(1+F(X_{1, k}))\, \|X_{1, k}\|^{2}\big]<\infty .
  \]
  \item[(ii)] There are positive numbers $\nu_1$ and $\nu_2$ such that $f_{\varepsilon|X, k}(u\mid X_{1, k}) \ge \nu_1$ for all $|u|\le \nu_2.$
\item[(iii)] The matrices
  \(
    V_k:=\mathbb{E}\!\left[f_{\varepsilon|X, k}(0\mid X_{1, k})\, X_{1, k}X_{1, k}^\top\right]
    \quad\text{and}\quad
    Q_k:=\tau(1-\tau)\, \mathbb{E}\!\left[X_{1, k}X_{1, k}^\top\right]
  \)
  are nonsingular.
\end{enumerate} Then the partition \(\widehat{C}_n^{(R(n))}\) produced by Algorithm~\ref{alg:CoC-algorithm_bootstrap_iterations},  with,  for each \(k\) and each bootstrap replicate \(r=1, \ldots, R(n)\), 
$$
\widehat\theta_{n, k}\in\arg\min_{\theta\in\mathbb{R}^p}\;\frac{1}{n}\sum_{i=1}^n \rho_\tau\bigl(Y_{i, k}^{(r)} -X_{i, k}^{(r)\!\top}\theta\bigr)
$$
\[
\widehat Q_{n, k}\;:=\;\tau(1-\tau)\, \frac{1}{n}\sum_{i=1}^n X_{i, k}X_{i, k}^\top, 
\qquad
\widehat V_{n, k}\;:=\;\frac{1}{n}\sum_{i=1}^n \widehat f_{\varepsilon|X, k}\!\left(0\, \big|\, X_{i, k}\right)\, X_{i, k}X_{i, k}^\top, 
\]
where \(\{(Y_{i, k}^{(r)},  X_{i, k}^{(r)}) : i=1, \ldots, n\}\) comes from  
 the nonparametric bootstrap  
enjoys the Golden-Partition Recovery property.
\end{prop}

\subsubsection{U-statistics}

We end this section on examples with the class of $U$-statistics.
For each centre \(k\in\{1, \ldots, K\}\),  let \(h_k:\mathcal X\times\mathcal X\to\mathbb R\) be a symmetric kernel and define the population target
\(
\theta_{0, k}\;:=\;\mathbb{E}\big[h_k(X_{1, k}, X_{2, k})\big].
\)
In this section,    take the estimator
\[
\widehat\theta_{n, k} \;:=\;\frac{2}{n(n-1)}\sum_{1\le i<j\le n} h_k(X_{i, k}, X_{j, k}), 
\]
for the nonparametric bootstrap replicate \(r\).

 Let
\(
\nu_k(x):=\mathbb{E}\!\big[h_k(x, X_{2, k})\big], \qquad
\psi_k(x):=\nu_k(x)-\theta_{0, k}, \qquad
\sigma_{1, k}^2:=\mathbb Var\!\big(\psi_k(X_{1, k})\big).
\) 
Then the  Hoeffding’s decomposition gives
\[
\widehat\theta_{n, k}-\theta_{0, k}
\;=\;\frac{2}{n}\sum_{i=1}^n \psi_k(X_{i, k})\;+\;R_{n, k}, 
\qquad R_{n, k}=o_p(n^{-1/2}), 
\]
so,  if \(\sigma_{1, k}^2>0\) and \(\mathbb{E}[h_k(X_{1, k}, X_{2, k})^2]<\infty\). The conditions used to establish the validity of the nonparametric bootstrap in the following proposition are standard; see,  for example,  \cite[Chap.~3]{shao2012jackknife}.

\begin{prop} 
\label{prop::Ustat}
Assume that the kernel \(h_k\) satisfies,   
\[
\mathbb{E}\!\big[h_k(X_{1, k}, X_{2, k})^2\big]<\infty, \qquad
\mathbb{E}\!\big|h_k(X_{1, k}, X_{1, k})\big|<\infty, \qquad
x\mapsto \mathbb{E}\!\big[h_k(x, X_{2, k})\big]\ \text{is not a.s.\ constant}
\]
and the assumption \textbf{A7} holds.
 Then the partition \(\widehat{C}_n^{(R(n))}\) produced by Algorithm~\ref{alg:CoC-algorithm_bootstrap_iterations},  with,  for each \(k\) and each bootstrap replicate \(r=1, \ldots, R(n)\),    
 $$
 \widehat\theta_{n, k} ^{(r)}\;:=\;\frac{2}{n(n-1)}\sum_{1\le i<j\le n} h_k(X_{i, k}^{(r)}, X_{j, k}^{(r)})
 $$
\[
\widehat\nu_{i, k}:=\frac{1}{n-1}\sum_{j\ne i} h_k(X_{i, k}, X_{j, k}), 
\qquad
\widehat V_{n, k}:=\sqrt{\frac{4}{n}\sum_{i=1}^n \big(\widehat\nu_{i, k}-\widehat\theta_{n, k}\big)^2}.
\]
 $\widehat Q_{n, k}:=1$
where \(\{X_{i, k}^{(r)} : i=1, \ldots, n\}\) comes from  
 the nonparametric bootstrap  
enjoys the Golden-Partition Recovery property.
\end{prop}

 \subsection{An alternative to weaken the independency assumption}\label{sec:weakAssumptions}

The independence assumptions \textbf{(A1)}--\textbf{(A4)} can be replaced by the following weaker conditions, which allow for cross-centre dependence.

\begin{itemize}
    \item[(A1)] For each \(k\), \(V_k\in\mathbb{R}^{p\times p}\) is symmetric positive definite.

    \item[(A2)] Let \(\overline{U}_n := (U_{n,1}^\top,\ldots,U_{n,K}^\top)^\top\in\mathbb{R}^{Kp}\). As \(n\to\infty\),
    \(\overline{U}_n \Rightarrow \mathcal{N}(0,\overline{Q})\),
    where \(\overline{Q}\in\mathbb{R}^{(Kp)\times(Kp)}\) is positive semidefinite. Moreover, the diagonal blocks \(Q_{kk}\) of \(\overline{Q}\) are positive definite.

    \item[(A3)] For each \(k\), \(\varepsilon_{n,k}\xrightarrow{p} 0\) as \(n\to\infty\).
\end{itemize}

\subsubsection{Sufficient conditions for (A2)}

To justify \textbf{(A2)}, we assume that cross-centre dependence is \emph{structural}: it appears only when the data streams are correctly aligned.

Write \(n=b_n m_n\) with \(b_n\to\infty\) and \(m_n\to\infty\). A \emph{matching partition} \(\mathcal{M}_n\) is a tuple \((\mathcal{P}_{n,1},\ldots,\mathcal{P}_{n,K})\), where each \(\mathcal{P}_{n,k}\) partitions \(I_k\) into \(m_n\) disjoint blocks of size \(b_n\). For a given \(\mathcal{M}_n\), define the matched block score vector
\(S_m(\mathcal{M}_n) := (S_{m,1}^\top,\ldots,S_{m,K}^\top)^\top\in\mathbb{R}^{Kp}\),
where, for \(m=1,\ldots,m_n\),
\[
S_{m,k} \;:=\; \frac{1}{\sqrt{b_n}}\sum_{i\in I_k^{m}} u_{i,k}\in\mathbb{R}^p,
\qquad
S_m \;:=\; \big((S_{m,1})^\top,\ldots,(S_{m,K})^\top\big)^\top .
\]

\paragraph{Dependence assumptions.}
We postulate the existence of an (unobserved) alignment that captures the joint distribution across centres.

\begin{itemize}
    \item[\textbf{(B0)}]
    There exists a specific (unobserved) sequence of matching partitions \(\mathcal{M}_n^*\) (the \emph{true matching}) such that the stacked block score vectors
    \(S_m^* := S_m(\mathcal{M}_n^*)\in\mathbb{R}^{Kp}\) are jointly well-defined on a common probability space, with
    \[
    S_{m,k}^* \;:=\; \frac{1}{\sqrt{b_n}}\sum_{i\in I_k^{*m}} u_{i,k}\in\mathbb{R}^p,
    \qquad
    S_m^* \;:=\; \big((S_{m,1}^*)^\top,\ldots,(S_{m,K}^*)^\top\big)^\top .
    \]
    \begin{itemize}
        \item Under \(\mathcal{M}_n^*\), the vectors \(\{S_m^*\}_{m=1}^{m_n}\) exhibit non-trivial cross-centre covariance, i.e., \(\mathrm{Cov}(S_{m,k}^*,S_{m,\ell}^*)\neq 0\) for some \(k\neq \ell\).
        \item Under any mismatched partition \(\mathcal{M}_n\neq \mathcal{M}_n^*\), cross-centre covariance based on paired blocks is attenuated, and (under random misalignment) converges to \(0\).
    \end{itemize}
\end{itemize}

\paragraph{Remark.}
If centres are independent, then all off-diagonal blocks of \(\mathbb{V}\mathrm{ar}(S_1^*)\) are zero, and any tuple of partitions across centres yields an equally valid matching from the covariance perspective.

\paragraph{Example: latent factor generative process (Regime 1).}
Cross-centre dependence is induced by a block-level common driver, while independence across blocks is enforced by renewing that driver from block to block. Concretely, for each \(m=1,\ldots,m_n\), draw a latent factor \(W_m\) i.i.d.\ from a fixed distribution, and independently draw an auxiliary innovation \(\xi_m\) i.i.d.\ with \(\xi_m\perp (W_1,\ldots,W_{m_n})\). Given \((W_m,\xi_m)\), generate the entire \(m\)-th matched block across centres,
\(\mathcal{Z}_m := \{Z_{i,k}: i\in I_k^{*m},\ k=1,\ldots,K\}\),
as \(\mathcal{Z}_m=g(W_m,\xi_m)\), where the map \(g\) does not depend on \(m\). This implies that the blocks \(\{\mathcal{Z}_m\}_{m=1}^{m_n}\) are i.i.d., hence \(\{S_m^*\}_{m=1}^{m_n}\) is i.i.d.\ as well. To encode conditional independence across centres given \(W_m\), one may take \(\xi_m=(\xi_{m,1},\ldots,\xi_{m,K})\) with \(\xi_{m,1},\ldots,\xi_{m,K}\) conditionally independent given \(W_m\), and set \(\{Z_{i,k}: i\in I_k^{*m}\}=g_k(W_m,\xi_{m,k})\). Then, even if \(S_{m,k}^*\perp S_{m,\ell}^*\mid W_m\), one typically has \(\mathrm{Cov}(S_{m,k}^*,S_{m,\ell}^*)\neq 0\) because both depend on \(W_m\), matching the first bullet of \textbf{(B0)}. A mismatched partition pairs blocks driven by \(W_m\) with blocks driven by \(W_{\tilde m}\) for \(\tilde m\neq m\); since \(W_m\perp W_{\tilde m}\), the shared driver is removed and cross-centre covariance is attenuated, converging to \(0\) under random misalignment.  

\paragraph{Example: time series with offsets (Regime 2).}
Cross-centre dependence is carried by a common latent time process, while misalignment destroys it because the latent process mixes over time. Let centre \(k\) observe a regular grid with offset \(t_{i,k}=\tau_k+i\), \(i=1,\ldots,n\), where \(\tau_k\in\mathbb{Z}\). Let \(\{W_t:t\in\mathbb{Z}\}\) be strictly stationary and mixing. Conditional on the path \(\{W_t\}\), assume the centres generate their series independently (e.g., via centre-specific dynamics driven by \(W_{t_{i,k}}\)), so cross-centre dependence can only arise through the shared \(\{W_t\}\). The true matching is determined by \emph{absolute time}: partition \(\mathbb{Z}\) into contiguous windows \(J^m=\{t\in\mathbb{Z}:(m-1)b_n<t\le mb_n\}\), and define \(I_k^{*m}=\{i:t_{i,k}\in J^m\}=\{i:(m-1)b_n-\tau_k<i\le mb_n-\tau_k\}\). Then \(S_{m,k}^*\) and \(S_{m,\ell}^*\) are built from observations driven by the same latent segment \(\{W_t:t\in J^m\}\), and the cross-centre covariance \(Q_{k\ell}\) is preserved. In contrast, ignoring \(\tau_k\) pairs blocks corresponding to different time windows \(J^m\) and \(J^{\tilde m}\); by mixing, these segments become asymptotically independent when \(|m-\tilde m|\) is large, so the induced cross-centre correlation vanishes. This explains why contiguous blocking can recover each \(Q_{kk}\) (marginally), whereas recovering \(Q_{k\ell}\) requires offset-aware alignment.

While the existence of \(\mathcal{M}_n^*\) ensures the validity of the asymptotic distribution in \textbf{(A2)}, identifying this partition in practice is often infeasible. Consequently, the specific off-diagonal blocks \(Q_{k\ell}\) of \(\overline{Q}\) remain unknown, while the diagonal blocks \(Q_{kk}\) (local asymptotic variances) are estimable regardless of the matching. Also, independence is important in the derivation of the results of the theorem \ref{thm::tail_bounds_WhatV2}.

\subsection{Algorithm}
{
 
\spacingset{1}
\begin{algorithm}[!htb]
\SetAlgoLined
\KwIn{Significance level \(\alpha\); plateau length \(N_{\max}\);
      fixed matrices \(\{\widehat V_{n, k}, \, \widehat Q_{n, k}\}_{k=1}^K\);
      \(R\) bootstrap sets of point estimators \(\{\widehat\theta^{(r)}_{n, k}\}_{k=1}^K\),  \(r=1, \dots, R\).}
\KwOut{Final partition \(C^\star\).}

\(C^{(1)} \leftarrow
  \mathrm{one\_shot\_CoC}\big(
    \{\widehat\theta^{(1)}_{n, k}\}_{k=1}^K, \, 
    \{\widehat V_{n, k}\}_{k=1}^K, \, 
    \{\widehat Q_{n, k}\}_{k=1}^K, \, 
    \alpha\big)\);\;
\(n_1 \leftarrow |C^{(1)}|\);\quad
\(\textit{runlen} \leftarrow 1\);\quad
\(r \leftarrow 2\);\;

\While{\(\textit{runlen} < N_{\max}\)}{
  \(i \leftarrow 1 + ((r-1) \bmod R)\) \tcp*{replicate index for this round}

  Let \(B_1, \dots, B_L\) be the blocks of \(C^{(r-1)}\)\;
  \(C^{(r)} \leftarrow \{B_1\}\) \tcp*{start new partition with \(B_1\)}

  \For{\(\ell = 2\) \KwTo \(L\)}{
      \(\mathcal E \leftarrow \varnothing\) \tcp*{clusters with \(p \ge \alpha\)}

      \ForEach{cluster \(D \in C^{(r)}\)}{
         \(p_D \leftarrow
            \mathrm{integration\_test}\big(
                D, \, B_\ell \, \big|\, 
                \{\widehat\theta^{(i)}_{n, k}\}_{k=1}^K, \, 
                \{\widehat V_{n, k}\}_{k=1}^K, \, 
                \{\widehat Q_{n, k}\}_{k=1}^K
            \big)\)\;
         \If{\(p_D \ge \alpha\)}{add \(D\) to \(\mathcal E\)\;}
      }

      \uIf{\(\mathcal E = \varnothing\)}{
           \(C^{(r)} \leftarrow C^{(r)} \cup \{B_\ell\}\)  \tcp*{new cluster}
      }\Else{
           \(D^\star \leftarrow \arg\max_{D \in \mathcal E} p_D\)\;
           Merge \(B_\ell\) into \(D^\star\)\;
      }
  }

  \(n_r \leftarrow |C^{(r)}|\)\;
  \If{\(n_r = n_{r-1}\)}{\(\textit{runlen} \leftarrow \textit{runlen} + 1\)}\Else{\(\textit{runlen} \leftarrow 1\)}\;
  \(r \leftarrow r + 1\)\;
}

\Return{\(C^\star \leftarrow C^{(r-1)}\)}
\caption{Cyclical multi-round bootstrap CoC algorithm  }
\label{alg:CoC-algorithm_cyclic_tiebreak}
\end{algorithm}
}

 {\spacingset{1}
\begin{algorithm}[!htb]
\SetAlgoLined
\KwIn{Estimators $(\widehat\theta_{n, k}, \widehat V_{n, k}, \widehat Q_{n, k})$ for $k=1, \dots, K$,  level $\alpha$.}
\KwOut{Clusters $C$.}

Compute $p_{\rm global}\gets\text{global\_homogeneity\_test}(\{\widehat\theta_{n, k}, \widehat V_{n, k}, \widehat Q_{n, k}\}_{k=1}^K)$ via Proposition~\ref{cor::testHomogemeity}\;
\uIf{$p_{\rm global}\ge \alpha$}{
  \Return{$\{\{1, \dots, K\}\}$}  \tcp*{All centres homogeneous}
}

Initialize $C\leftarrow\{\{1\}\}$\;
\For{$j\leftarrow 2$ \KwTo $K$}{
  \tcp{compute p-values against each existing cluster}
  Compute $L = |C|$\;
  \ForEach{cluster index $\ell=1, \dots, L$}{
    $p_\ell \leftarrow\mathrm{integration\_test}(C_\ell, j)$ via Proposition~\ref{prop::integration}\;
  }
  $\mathcal E \leftarrow \{\ell:\; p_\ell \ge \alpha\}$\;
  \uIf{$\mathcal E=\varnothing$}{
    $C\leftarrow C\cup\{\{j\}\}$\;  \tcp*{new singleton cluster}
  }
  \Else{
    \tcp{choose the cluster with largest p-value; tie-break: smallest index}
    $\ell^\star \leftarrow \displaystyle\arg\max_{\ell\in\mathcal E} p_\ell$ \;
    Merge centre $j$ into $C_{\ell^\star}$ (i.e.\ $C_{\ell^\star}\leftarrow C_{\ell^\star}\cup\{j\}$)\;
  }
}
\Return{$C$}
\caption{CoC-algorithm: Clusters of Centres Algorithm}
\label{alg:CoC-algorithm_updated}
\end{algorithm}
}

\end{document}